\documentclass{article}%
\usepackage{graphicx}
\usepackage{amsmath}
\usepackage{amsfonts}
\usepackage{amssymb}%
\providecommand{\U}[1]{\protect\rule{.1in}{.1in}}
\numberwithin{equation}{section}
\newtheorem{theorem}{Theorem}

\newtheorem{definition}[theorem]{Definition}

\newtheorem{lemma}[theorem]{Lemma}
\newtheorem{notation}[theorem]{Notation}

\newtheorem{proposition}[theorem]{Proposition}
\newtheorem{remark}[theorem]{Remark}

\newenvironment{proof}[1][Proof]{\textbf{#1.} }{\ \rule{0.5em}{0.5em}}
\begin{document}


\bigskip

\bigskip

\bigskip


{\Large Veiled singularities for the spherically symmetric massless
Einstein-Vlasov system.}

\begin{center}
\bigskip

Alan D. Rendall\footnote{Johannes Gutenberg Universit\"{a}t, Institut f\"{u}r
Mathematik, Staudingerweg 9, 55099 Mainz, Germany.}, Juan J.
L. Vel\'{a}zquez\footnote{Institute of Applied Mathematics, University of Bonn,
Endenicher Allee 60, 53115 Bonn, Germany. E-mail: velazquez@iam.uni-bonn.de}
\end{center}

\bigskip

\bigskip

\begin{abstract}
This paper continues the investigation of the formation of naked singularities
in the collapse of collisionless matter initiated in \cite{RV}. There the
existence of certain classes of non-smooth solutions of the Einstein-Vlasov
system was proved. Those solutions are self-similar and hence not
asymptotically flat. To obtain solutions which are more physically relevant it
makes sense to attempt to cut off these solutions in a suitable way so as to
make them asymptotically flat. This task, which turns out to be technically
challenging, will be carried out in this paper.

\end{abstract}

\section{Introduction}

In this paper we continue the construction of a class of singular solutions of
the Einstein-Vlasov system which was started in \cite{RV}. It is well known
that solutions of the Einstein equations coupled with suitable matter
models can yield singularities in finite time. More precisely, the meaning of
this statement is the following. Solutions of the Einstein equations are
certain spacetime manifolds, which are characterized by suitable
pseudoriemannian metrics. The usual terminology in general relativity is that
it is said that there is a singularity if the corresponding spacetime has a
metric that fails to be causally geodesically complete. By this it is
understood that there is a timelike or null geodesic which, at least in one
direction, cannot be extended and has a finite affine length.

One of the best known examples of singularities in general relativity is given
by black holes. They are characterized by the presence of a event horizon
which ensures that the singularity does not have an influence on distant
observers. A singularity which is not covered by an event horizon, so that it
is visible to distant observers, is known as a naked singularity. This type of
singularity is physically problematic for the following reason. To say that
the singularity is visible to distant observers means that there exist causal
geodesics which approach the singularity in the past time direction and which
in the future time direction enter regions where the density of matter and the
gravitational fields become arbitrarily small and the geometry resembles that
of the flat Minkowski spacetime. These physical notions can be formulated
mathematically using the concept of an asymptotically flat spacetime. In
general relativity causal influences propagate along causal geodesics and so a
naked singularity leads to a situation where the singularity can have an
influence on physical processes in distant regions. Since we do not have a
complete theory describing the physics at a singularity this means a breakdown
of the ability of physics to make predictions. There is a situation which is a
priori milder but which is problematic for a similar reason. This is where
although there are no geodesics of the type characterizing a naked singularity
there are families of causal geodesics which reach distant regions and come
from regions where physical conditions are arbitrarily extreme. This can be
made mathematically precise by saying that they come from regions where some
geometrical invariants, such as the Kretschmann scalar $R^{\alpha\beta
\gamma\delta}R_{\alpha\beta\gamma\delta}$ or the invariant $T^{\alpha\beta
}T_{\alpha\beta}$ built out of the energy-momentum tensor, become arbitrarily
large. From the point of view of physics this is just as problematic as a
naked singularity since we do not have a reliable theoretical control of
physical phenomena in regimes far beyond those accessible to experiments. We
denote this type of situation as a 'veiled singularity' since it means
intuitively that although there is not a singularity which is directly visible
we can nevertheless make observations of a region where known theory is in
danger of being invalid. This terminology is not intended to indicate that a
veiled singularity need be associated with geodesic incompleteness. A
situation like this is also present in the solutions of the Einstein equations
coupled to a massless scalar field with singular future light cones found by
Christodoulou \cite{Chr2}, \cite{Chr3}. Note, however, that Christodoulou's
solutions have a significantly different causal structure from those we construct.

There are three features of the solutions constructed in \cite{RV} which could
be seen as disadvantages from the point of view of their physical
applicability and which should be improved if possible. The first is the fact
which has already been mentioned that they are not asymptotically flat and
thus do not represent isolated systems. The second is that the initial data
from which they evolve is not smooth. The third is that the matter source used
models massless particles like photons rather than ordinary matter. In this
paper we remove the first disadvantage, leaving the other two for future
investigation. The hope is that eventually the rough data can be approximated
by smooth data and the massless particles by massive ones in such a way that
the key dynamical properties of the solutions are preserved. In fact there is
recent work which indicates that the case of massless particles may be of
considerable interest in its own right as a physical model. There has been a
detailed numerical study of gravitational collapse in that case which
concentrates on type I critical collapse \cite{AkbCh}. Type II critical
collapse, which might be related to the phenomena studied in the present
paper, is only briefly mentioned in \cite{AkbCh}. The fact that global
existence for small initial data has been proved for asymptotically flat
solutions of the Einstein-Vlasov system with massless particles \cite{Taylor}
is a strong indication of the robustness of this model.

Next the basic mathematical set-up used in the paper will be described. In
general relativity the dynamics of self-gravitating matter is described by
means of solutions of the Einstein equations coupled to other equations
describing the matter content. The type of singularities which could arise
depend very strongly on the type of matter model used in the system. In this
paper we will be concerned with collisionless matter, described by the Vlasov
equation. Moreover, we will assume in addition that the point particles
represented by the matter model have zero mass. The combined system of the
Einstein equations with this model of matter is the massless Einstein-Vlasov system.

In what follows we will consider solutions of the Einstein-Vlasov system where
the particle density is not a bounded function, but a measure concentrated on
some hypersurfaces that will be described in detail later. A consequence of
this is that the Einstein equations are not satisfied in classical form, but
in a suitable distributional sense. There is a class of distributional
solutions of the Einstein-Vlasov system which are equivalent to what is
usually known in the literature as dust. Dust solutions of Vlasov systems have
the property that there is a unique possible value of the velocity at each
point of spacetime. The solutions considered in this paper are somehow more
general than the usual dust solutions considered in the literature, because
they have a set of admissible velocities at each point of the spacetime but
the dimension of that admissible set of velocities is smaller than the total
dimension of the phase space. Moreover, the set of admissible velocities at a
given spatial point is qualitatively different in different regions. The
number of possible values of the radial velocity for fixed angular momentum is
two, one or zero. The number changes at some particular points which will be
referred to as turning points. There the support of the distribution function
fails to be transverse to the fibres of the tangent bundle. Some matter
variables like the density and pressure become unbounded in a neighbourhood of
the turning points. From this point of view the solutions considered in this
paper are intermediate between dust and smooth solutions and hence will be
called dust-like solutions. Note that in contrast to dust they do have some
velocity dispersion. The dimension of the support of $f$ in the tangent space
at a given spacetime point is zero for dust, two for the solutions in this
paper and three for smooth solutions of the Einstein-Vlasov system. There is a
particular type of solutions of the Einstein equations known as generalized
Einstein clusters, which were first studied in \cite{Bo}, \cite{Dat}. These
solutions can also be thought of as distributional solutions of the
Einstein-Vlasov system for which the support of $f$ in the tangent space at a
given spacetime point is one. A more detailed discussion about the relation
between generalized Einstein clusters and the Einstein-Vlasov system can be
found in \cite{RV}. For the solutions here it will be possible to describe the
distribution of velocities for the particles at a given point using a function
depending on one coordinate, while a general distribution of velocities
compatible with the assumption of spherical symmetry would depend on two coordinates.

The results of this paper are a continuation of those in \cite{RV}. In that
paper, a class of dust-like self-similar solutions of the Einstein-Vlasov
system were obtained. Those solutions do not have an event horizon anywhere in
the spacetime. The difficulty with those solutions is that, due to their
self-similar character, they cannot asymptotically resemble the Minkowski
metric at large distances from the center. In this paper we show that it is
possible to cut off the distribution of matter in a suitable manner outside a
compact set and to obtain a solution of the Einstein-Vlasov system whose
metric behaves asymptotically far away from the center like that of Minkowski
spacetime. In addition, the spacetime constructed in this paper will have the
property that it is not geodesically complete and that horizons do not appear at
any point of the spacetime.

It is interesting to remark that for the solutions constructed in this paper,
the distribution of matter away from the center for long times asymptotically
approaches the distribution of one Einstein cluster with all its mass
contained in an interval $r\in\left(  0,R_{\max}\right)  ,$ with a density of
matter $\rho$ which increases linearly with the distance to the center. A more
detailed description of the matter distribution for long times, away from the
center, will be found in Section 8.

Due to the mathematical complexity of the Einstein equations many of the
studies related to singularity formation for these equations have been carried
out for spherically symmetric solutions. In spherical symmetry the Einstein
vacuum equations are non-dynamical due to Birkhoff's theorem, which says that
any spherically symmetric vacuum solution is locally isometric to the
Schwarzschild solution and, in particular, static. Thus it is essential to
include matter of some kind. A matter model which has proved very useful for
this task is the scalar field. This is a real-valued function $\phi$ which
satisfies the wave equation $\nabla^{\alpha}\nabla_{\alpha}\phi=0$. In this
case the Einstein equations take the form $R_{\alpha\beta}=8\pi\nabla_{\alpha
}\phi\nabla_{\beta}\phi$ where $R_{\alpha\beta}$ is the Ricci curvature of
$g_{\alpha\beta}.$ The spherically symmetric Einstein-scalar field equations
were studied in great detail in a series of papers by D. Christodoulou. This
culminated in \cite{Chr2} and \cite{Chr3}. In \cite{Chr2} it was shown that in
this system naked singularities can evolve from regular asymptotically flat
initial data. It was shown in \cite{Chr3} that generic initial data do not
lead to naked singularities. (A precise definition of naked singularities can
be found in \cite{RV}).

For the spherically symmetric Einstein-scalar field equations it is known from
the work of Christodoulou \cite{Chr1} that small asymptotically flat initial
data lead to a solution which is geodesically complete and hence free of
singularities. This small data result has recently been extended to the case
without symmetry in \cite{LR}. On the other hand there are certain large
initial data for which it is known that a black hole is formed. The threshold
between these two types of behaviour was studied by Choptuik (cf. \cite{Ch})
and many other papers since. This area of research is known as critical
collapse. It is entirely numerical and heuristic and unfortunately
mathematically rigorous results are not yet available.

The plan of this paper is the following. In Section 2 we recall the system of
partial differential equations which describes the Einstein-Vlasov system in
the spherically symmetric case. We also proved that it is possible to
reformulate the problem as a system of equations where the angular momentum
does not appear explicitly and therefore the system can be reformulated in
terms of one variable less. Section 3 describes in a heuristic manner the
construction which will be carried out in this paper and we state the main
results. In this section we also give the precise definition of
measured-valued solution which will be used in this paper. Section 4
summarizes the main properties of the self-similar solutions constructed in
\cite{RV}. Moreover, some additional asymptotic properties of the solution are
also obtained. Section 5 contains a description of the functional analysis
properties which will be used in the proof of the main results of the paper.
This section also contains the description of some auxiliary PDEs which will
be used in the proof of the main result. Section 6 contains the fixed point
argument which proves the main result of the paper. Section 7 contains a
description of the properties of the spacetime constructed in the paper. In
particular, it is proved that the resulting metric is not geodesically
complete and the absence of a horizon in the spacetime obtained.

In order to simplify the notation we will use the following convention. We
will use generic functions $\Phi$ depending on the variables $\left(t,r\right)  ,$ i.e.$\ \Phi=\Phi\left(t, r\right)  .$ However, on several
occasions we will change to new variables $\left( \tau,y\right)  $ by means
of suitable diffeomorphisms $\left(t, r\right)  \rightarrow\left(\tau,y\right)  .$ This defines a new function $\tilde{\Phi}$ satisfying
$\tilde{\Phi}\left( \tau,y\right)  =\Phi\left(t, r\right)  .$ We will denote
$\tilde{\Phi}$ by $\Phi$ for simplicity, since no risk of confusion will
arise due to this.

\section{REFORMULATING THE EINSTEIN-VLASOV SYSTEM AS A SYSTEM OF PARTIAL
DIFFERENTIAL EQUATIONS.}

\subsection{\label{Section1}Einstein-Vlasov System in Schwarzschild
coordinates.}

We recall here the system of partial differential equations which describe the
solutions of the massless Einstein-Vlasov system in the spherically symmetric
case. These equations have been summarized in \cite{Rein} and we will just
refer to the corresponding formulas there.

A convenient way of writing the metric for spherically symmetric spacetimes
uses a modified version of the classical Schwarzschild coordinates (cf.
\cite{Rein}):%
\begin{equation}
ds^{2}=-e^{2\mu\left(  t,r\right)  }dt^{2}+e^{2\lambda\left(  t,r\right)
}dr^{2}+r^{2}\left(  d\theta^{2}+\sin^{2}\theta d\varphi^{2}\right)  .
\label{met1}%
\end{equation}
which is chosen so that $\mu\left(  t,0\right)  =0.$ This normalization
implies that the time variable $t$ is just the proper time at the center
$r=0.$

Due to the symmetry of the metric, a suitable way to describe the kinematic
characteristics of collisionless matter is by means of the following
quantities:%
\begin{equation}
r=\left\vert x\right\vert \;,\;w^{\left(  t\right)  }=\frac{x\cdot v}%
{r}\;,\;F=\left\vert x\wedge v\right\vert ^{2} \label{coord}%
\end{equation}
where the upper index $\left(  t\right)  $ in $w^{\left(  t\right)  }$ stands
for the tangential component.

A convenient feature of the choice of variables (\ref{coord}) is that the
angular momentum variable $F$ is constant along characteristics.

We will write the particle density as:%
\[
f=f\left(  t,r,w^{\left(  t\right)  },F\right)
\]

Then, using the fact that collisionless matter moves along the light rays
associated to the metric (\ref{met1}) we obtain, that the particle density
satisfies the first order PDE (cf. \cite{Rein}):%
\begin{equation}
\partial_{t}f+e^{\mu-\lambda}\frac{w^{\left(  t\right)  }}{E}\partial
_{r}f-\left(  \lambda_{t}w^{\left(  t\right)  }+e^{\mu-\lambda}\mu_{r}%
E-e^{\mu-\lambda}\frac{F}{r^{3}E}\right)  \partial_{w^{\left(  t\right)  }}f=0
\label{S1E1}%
\end{equation}
where:
\begin{equation}
E=\sqrt{\left(  w^{\left(  t\right)  }\right)  ^{2}+\frac{F}{r^{2}}}
\label{S1E2}%
\end{equation}

We are assuming in (\ref{S1E2}) that we are dealing with massless particles.
For massive particles we should replace (\ref{S1E2}) by $E=\sqrt{1+\left(
w^{\left(  t\right)  }\right)  ^{2}+\frac{F}{r^{2}}},$ but in that case some
invariance properties under rescalings which will be used in the construction
of the solutions in this paper would be lost. Notice, however, that in the
case of particles with velocities close to the speed of light (\ref{S1E2})
would be a good approximation for the particle energy, even in the case of
massive particles.

Using the energy-momentum tensor for collisionless matter (cf. \cite{Rein},
\cite{Rendall}), it turns out that Einstein equations for gravitational fields
become the following system of equations:%
\begin{align}
e^{-2\lambda}\left(  2r\lambda_{r}-1\right)  +1  &  =8\pi r^{2}\rho
,\label{S1E3}\\
e^{-2\lambda}\left(  2r\mu_{r}+1\right)  -1  &  =8\pi r^{2}p \label{S1E4}%
\end{align}
where, using suitable normalizations for $t$ and $r,$ we must use the
following boundary conditions:
\begin{align}
\mu\left(  t,0\right)   &  =0\;\;,\;\lambda\left(  t,0\right)
=0\ ,\;\label{S1E5}\\
\lambda\left(  t,\infty\right)   &  =0. \label{S1E6}%
\end{align}

The functions $\rho,\ p$ in (\ref{S1E3}) encode all the relevant information
in the energy-momentum tensor. In the case of collisionless matter they are
given by:%
\begin{align}
\rho &  =\rho\left(  t,r\right)  =\frac{\pi}{r^{2}}\int_{-\infty}^{\infty
}\left[  \int_{0}^{\infty}EfdF\right]  dw^{\left(  t\right)  },\label{S1E7}\\
p  &  =p\left(  t,r\right)  =\frac{\pi}{r^{2}}\int_{-\infty}^{\infty}\left[
\int_{0}^{\infty}\frac{(w^{(t)})^{2}}{E}fdF\right]  dw^{\left(  t\right)  }.
\label{S1E8}%
\end{align}

It is useful to notice that $p\leq\rho.$ The system (\ref{S1E1}),
(\ref{S1E3})-(\ref{S1E8}), (\ref{S1E2}) is invariant under the rescaling:
\begin{equation}
r\rightarrow\theta r\;\;,\;\;t\rightarrow\theta t\;\ \ \text{for\ }%
t<0\;,\;\;w^{\left(  t\right)  }\rightarrow\frac{1}{\sqrt{\theta}}w^{\left(
t\right)  }\;\;,\;\;F\rightarrow\theta F \label{S1E10}%
\end{equation}
for any $\theta>0.$ It is then natural to look for solutions of (\ref{S1E1}),
(\ref{S1E3})-(\ref{S1E8}) invariant under the rescaling (\ref{S1E10}). Such
solutions are the self-similar solutions studied in \cite{RV}. However, the
metric associated to those solutions is not asymptotically flat as
$r\rightarrow\infty.$ The solutions obtained in this paper will be obtained by
replacing the distribution of collisionless particles for $r\geq R_{+}\left(
t\right)  ,$ with $R_{+}$ of order one, by another distribution, with a
smaller number of particles. As a consequence, it will not be possible to
analyze the differential equations for the particle distribution $f$ by means
of ODEs, but on the contrary, a more involved analysis, which requires understanding the behaviour of some hyperbolic systems, will be required. This
analysis will be the main contribution of this paper.

\subsection{Elimination of the angular momentum for a massless system.}

Due to the fact that we consider a system of massless particles, we can
reformulate (\ref{S1E1})-(\ref{S1E8}) as a system of equations where the
variable $F$ does not appear. More precisely, we can obtain a simpler, but
equivalent PDE system, where the unknown function $f$ depending on the
variables $\left(  t,w^{\left(  t\right)  },r,F\right)  $ is replaced by a new
function $\zeta$ which depends only on three variables $\left(  t,v,r\right)
.$ To this end we define a new variable:
\begin{equation}
v=\frac{w^{\left(  t\right)  }}{\sqrt{F}}. \label{S3E1}%
\end{equation}

Making the change of variables $\left(t,  r,w^{\left(t\right)  },F\right)
\rightarrow\left(t, r,v,F\right)  $ and denoting the new distribution
function by $f$ with a slight abuse of notation we can transform the system
(\ref{S1E1})-(\ref{S1E8}) into:%

\begin{align}
&  \partial_{t}f+e^{\mu-\lambda}\frac{v}{\tilde{E}}\partial_{r}f-\left(
\lambda_{t}v+e^{\mu-\lambda}\mu_{r}\tilde{E}-e^{\mu-\lambda}\frac{1}%
{r^{3}\tilde{E}}\right)  \partial_{v}f=0,\label{S3E2}\\
&  \tilde{E}=\sqrt{v^{2}+\frac{1}{r^{2}}},\ \rho=\frac{\pi}{r^{2}}%
\int_{-\infty}^{\infty}\left[  \int_{0}^{\infty}\tilde{E}fFdF\right]
dv,\ p=\frac{\pi}{r^{2}}\int_{-\infty}^{\infty}\left[  \int_{0}^{\infty}%
f\frac{v^{2}}{\tilde{E}}FdF\right]  dv. \label{S3E5}%
\end{align}

Notice that the change of variables (\ref{S3E1}) eliminates the dependence on
the variable $F$ for the characteristic curves associated to the Vlasov
equation (cf. (\ref{S3E2})). Moreover, the functions $\rho$ and $p$ and
therefore the functions $\lambda,\;\mu$ characterizing the gravitational
fields depend on $f$ only through the reduced distribution function:
\begin{equation}
\zeta\left( t, r,v\right)  \equiv\int_{0}^{\infty}fFdF. \label{S3E6}%
\end{equation}

In particular, it is possible to write a closed problem for the reduced
distribution function that can be obtained by multiplying (\ref{S3E2}) by $F$
and integrating with respect to this variable:
\begin{align}
&  \partial_{t}\zeta+e^{\mu-\lambda}\frac{v}{\tilde{E}}\partial_{r}%
\zeta-\left(  \lambda_{t}v+e^{\mu-\lambda}\mu_{r}\tilde{E}-e^{\mu-\lambda
}\frac{1}{r^{3}\tilde{E}}\right)  \partial_{v}\zeta=0,\label{S3E7}\\
&  \tilde{E}=\sqrt{v^{2}+\frac{1}{r^{2}}},\ \ \rho=\frac{\pi}{r^{2}}%
\int_{-\infty}^{\infty}\tilde{E}\zeta dv,\ \ p=\frac{\pi}{r^{2}}\int_{-\infty
}^{\infty}\frac{v^{2}}{\tilde{E}}\zeta dv. \label{S3E10}%
\end{align}

This system must be complemented with the field equations (\ref{S1E3}),
(\ref{S1E4}).

Notice that, given any solution of the problem (\ref{S1E3}), (\ref{S1E4}),
(\ref{S3E7}), (\ref{S3E10}) we can obtain a solution of the original system
(\ref{S3E2}), (\ref{S3E5}) by choosing any function $f=f\left(t,r,v,F\right)  $ that gives the values of $\zeta\left(t_{0},  r,v\right)  $
for any $t_{0}\in\mathbb{R}$, by means of (\ref{S3E6}). Using the
characteristics associated to (\ref{S3E7}), complemented with the equation
$\frac{dF}{dt}=0$ it is then possible to define $f\left(t,  r,v,F\right)  $
for the same range of values of $\left(t, r,v\right)  $ as the distribution
$\zeta.$ These results are made completely precise in Propositions
\ref{Equiv1}, \ref{fExt} below.

\section{GENERAL STRATEGY AND MAIN RESULTS.\label{Plan}}

\subsection{Heuristic idea behind the construction.}

The solution that we construct in this paper is obtained by means of a
suitable perturbation of the self-similar solutions obtained in \cite{RV}.
\ More precisely, our goal is to obtain measures $f=f\left(t,  r,v,F\right)
,$ $\zeta=\zeta\left(t, r,v\right)  $ solving (\ref{S1E3}), (\ref{S1E4}),
(\ref{S3E2}), (\ref{S3E5}), (\ref{S3E7}), (\ref{S3E10}) in a suitable
distributional sense (cf. Definitions \ref{fWeak}, \ref{zetaWeak} below). The
support of $\zeta$ consists of two surfaces in the space $\left(
r,v,t\right)  $ having the form $\gamma_{1}\left(  t\right)  =\left\{
v=v_{1}\left(  t,r\right)  \ ,\ r\geq y_{0}\left(  -t\right)  ,\ t_{0}\leq
t<0\right\}  ,\ $\linebreak$\gamma_{2}\left(  t\right)  =\left\{
v=v_{2}\left(  t,r\right)  \ ,\ r\geq y_{0}\left(  -t\right)  ,\ t_{0}\leq
t<0\right\}  \ $\ for some suitable functions $v_{1}\leq v_{2}$ and numbers
$y_{0},\ t_{0}$ to be fixed. Therefore, the solutions constructed in this
paper will have the form:%
\begin{equation}
f\left(t, r,v,F\right)  =A_{1}\left(  t,r,F\right)  \delta\left(
v-v_{1}\left(  t,r\right)  \right)  +A_{2}\left(  t,r,F\right)  \delta\left(
v-v_{2}\left(  t,r\right)  \right)  \ \label{F1E0a}%
\end{equation}%
\begin{equation}
\zeta\left(t, r,v\right)  =B_{1}\left(  t,r\right)  \delta\left(
v-v_{1}\left(  t,r\right)  \right)  +B_{2}\left(  t,r\right)  \delta\left(
v-v_{2}\left(  t,r\right)  \right)  \ \label{F1E0}%
\end{equation}
where $A_{1}\geq0,\ A_{2}\geq0,$ $B_{1}>0,\ B_{2}\geq0.$

We will assume in the rest of the paper that $f=0$ for $\left(
t,r,v,F\right)  =\left(  t,r,v,0\right)  $ in order to avoid singularities in
(\ref{S3E1}). Actually, we will assume an even more stringent condition on
$f,$ namely $f=0$ for $0\leq F\leq\delta_{0}\left(  -t\right)  $ for some
$\delta_{0}>0.$ Concerning the support in the $r$ coordinate, the self-similar
solutions will vanish for $r\leq y_{0}\left(  -t\right)  $ for some $y_{0}>0.$

The self-similar solution constructed in \cite{RV} is a solution of
(\ref{S1E3}), (\ref{S1E4}), (\ref{S3E7}), (\ref{S3E10}) with the form:%
\begin{align}
\zeta\left(  t,r,v\right)   &  =\left(  -t\right)  ^{2}\Theta\left(
y,V\right)  \ ,\;\mu\left(  t,r\right)  =U\left(  y\right)  \;\;,\;\;\lambda
\left(  t,r\right)  =\Lambda\left(  y\right)  \ ,\label{F1E2}\\
y  &  =\frac{r}{\left(  -t\right)  }\;\;,\;\;V=\left(  -t\right)
v\ \label{F1E3}%
\end{align}
where the measure $\Theta\left(  y,V\right)  $ is supported in two curves in
the plane \linebreak$\left\{  \left(  y,V\right)  :y>0,\ V\in\mathbb{R}%
\right\}  $ given by $\Gamma_{1}=\left\{  V=V_{1}\left(  y\right)  :y\geq
y_{0}\right\}  ,\ $\linebreak$\Gamma_{2}=\left\{  V=V_{2}\left(  y\right)
:y\geq y_{0}\right\}  ,$ with $y_{0}>0$ and $V_{1}\left(  y\right)
<V_{2}\left(  y\right)  $ for $y>y_{0}.$

Due to its self-similar character, the solution obtained in \cite{RV} does not
define a spacetime which is asymptotically flat as $r\rightarrow\infty.$ We
describe in this paper a procedure that allows to cut off this self-similar
solution for sufficiently large radii $r$ and obtain in this way an
asymptotically flat spacetime. The rationale behind the cutoff procedure used
is that in the limit $t\rightarrow0^{-}$ most of the mass is concentrated in
the curve $\Gamma_{2}.$ Actually, the fact that the spacetime associated to
the self-similar solution is not asymptotically flat as $r\rightarrow\infty$
is due to the infinite amount of mass contained in $\Gamma_{2}.$ On the other
hand the curve $\Gamma_{1}$ contains a small fraction of the mass that tends
to zero as $t\rightarrow0^{-}.$ It is then natural, in order to obtain a
solution containing a finite amount of mass, to cut off the branch $\Gamma
_{2}$ at some value of the radius $r=R_{+}\left(  t\right)  .$ Due to the fact
that in spherically symmetric situations the gravitational fields at a given
radius $\bar{r}$ depend only in the distribution matter at radii $r\leq\bar
{r}$ it follows that the dynamics of the branches $\Gamma_{1},\ \Gamma_{2}$ is
not modified for $r\leq R_{+}\left(  t\right)  $ and it agrees with the
dynamics obtained for the self-similar solutions. However, the gravitational
fields are modified for $r>R_{+}\left(  t\right)  $ and as a result the
dynamics of the particles placed in the branch $\Gamma_{1}$ must be modified
for $r>R_{+}\left(  t\right)  .$ Our construction will then provide a
measured-valued solution of (\ref{S1E3}), (\ref{S1E4}), (\ref{S3E7}),
(\ref{S3E10}) supported in the union of two curves $\gamma_{1}\left(
t\right)  ,\gamma_{2}\left(  t\right)  \subset\left\{  \left(  r,v\right)
:r>0\ ,\ v\in\mathbb{R}\right\}  $ with $t_{0}\leq t<0$ satisfying:%
\begin{align}
\mathcal{U}\left(  t\right)  \gamma_{1}\left(  t\right)   &  =\Gamma
_{1}\ \ \ \ \text{if\ \ }r\leq R_{+}\left(  t\right)  \ ,\ \ \ \mathcal{U}%
\left(  t\right)  \gamma_{2}\left(  t\right)  =\Gamma_{2}\cap\left\{
y\leq\frac{R_{+}\left(  t\right)  }{\left(  -t\right)  }\right\}
\ \label{F1E4}\\
\lim_{t\rightarrow0^{-}}R_{+}\left(  t\right)   &  =R_{\max}>0 \label{F1E5}%
\end{align}
where $\mathcal{U}\left(  t\right)  $ is a transformation from the half-plane
$\left\{  \left(  r,v\right)  :r>0\ ,\ v\in\mathbb{R}\right\}  $ to $\left\{
\left(  y,V\right)  :y>0\ ,\ V\in\mathbb{R}\right\}  $ given by (\ref{F1E3})
for any $t<0.$

Notice that the intersection of the support of the solution obtained with the
set $\left\{  r>R_{+}\left(  t\right)  \right\}  $ is just $\gamma_{1}\left(
t\right)  \cap\left\{  r>R_{+}\left(  t\right)  \right\}  .$ The solution
obtained is not self-similar for $r>R_{+}\left(  t\right)  $ due to the
presence in the problem of a length scale $R_{\max}.$ On the other hand it is
also worth noticing that due to the change of the spacetime structure for
$r>R_{+}\left(  t\right)  $ the most suitable time variable to describe the
dynamics of the region $r>R_{+}\left(  t\right)  $ is not $t,$ but the new
time variable $\tau=-\log\left(  -t\right)  $ that basically corresponds to
the Minkowski time for $r\rightarrow\infty.$ In this time variable the
formation of the singularity takes place for proper times approaching
infinity. On the contrary, at the center $r=0,$ the formation of the
singularity will take place in finite proper time.

\begin{notation}
\label{Not2}From now on we will denote as $C$ a positive constant independent
of the variables $\tau,\bar{\tau},t,r,L,T$. However, the constant $C$ could
depend on $y_{0},\ R_{\max}.$ Some of these variables will be defined later.
\end{notation}

\subsection{Definition of measure-valued solutions.}

We need to make precise in which sense the measure $f=f\left(  t,r,v,F\right)
$ defines a solution of the Einstein-Vlasov system (\ref{S1E1})-(\ref{S1E8}),
(\ref{S3E2}), (\ref{S3E5}). The definition that we will give in this Section
has the advantage that it only requires few regularity conditions for the
functions $\lambda,\mu$ or equivalently for the densities $\rho,\ p.$

As a first step, we need the following auxiliary result concerning the
well-posedness of the problem (\ref{S1E3})-(\ref{S1E6}).

\begin{lemma}
\label{defFields}Suppose that $r^{2}\rho,r^{2}p\in L^{1}\left(  \mathbb{R}%
_{+}\right)  .$ Let us assume also that the function $R_{0}\left(  r\right)
=8\pi\int_{0}^{r}\xi^{2}\rho\left(  \xi\right)  d\xi$ satisfies:%
\begin{equation}
R_{0}\left(  r\right)  <r\text{\ if }r\in\left(  0,\infty\right)
,\ \lim_{r\rightarrow0}\frac{R_{0}\left(  r\right)  }{r}=0,\ \int_{0}^{1}%
\frac{R_{0}\left(  \xi\right)  }{\xi^{2}}d\xi<\infty,\ \int_{0}^{1}\xi
p\left(  \xi\right)  d\xi<\infty\label{Y1E1}%
\end{equation}
We define functions $\lambda$ and $\mu$ by means of:%
\begin{align}
\lambda\left(  r\right)   &  =\frac{1}{2}\log\left(  r\right)  -\frac{1}%
{2}\log\left(  r-R_{0}\left(  r\right)  \right)  \ \label{Y1E2}\\
\mu\left(  r\right)   &  =\int_{0}^{r}\frac{4\pi\xi^{2}p\left(  \xi\right)
d\xi}{\left(  \xi-R_{0}\left(  \xi\right)  \right)  }+\frac{1}{2}\int_{0}%
^{r}\frac{R_{0}\left(  \xi\right)  d\xi}{\left(  \xi-R_{0}\left(  \xi\right)
\right)  \xi} \label{Y1E3}%
\end{align}
Then the functions $\lambda,\mu$ are in $C\left[  0,\infty\right)  \cap
W_{\mathrm{loc}}^{1,1}\left(  0,\infty\right)  $ and satisfy (\ref{S1E5}),
(\ref{S1E6}). They solve (\ref{S1E3}), (\ref{S1E4}) for almost all
$r\in\left(  0,\infty\right)  .$

For any $\delta_{0}>0,$ let us denote as $\mathcal{Z}_{\delta_{0}}$ the set of
functions in $\left(  L^{1}\left(  \mathbb{R}_{+},r^{2}dr\right)  \right)
^{2}$ which are supported in the half line $\left[  \delta_{0},\infty\right)
$ and satisfy $R_{0}\left(  r\right)  <r$ \ for any $r\in\left(
0,\infty\right)  .$ Let us endow $\mathcal{Z}_{\delta_{0}}$ with the topology
of $\left(  L^{1}\left(  \mathbb{R}_{+},r^{2}dr\right)  \right)  ^{2}.$ Then,
the mapping $\left(  \rho,p\right)  \rightarrow\left(  \lambda,\mu\right)  $
defines a continuous mapping from $\mathcal{Z}_{\delta_{0}}$ to $\left(
W^{1,1}\left(  0,L\right)  \right)  ^{2}$ for any $L>0.$
\end{lemma}

\begin{proof}
The conditions (\ref{Y1E1}) imply that the functions $\lambda$ and $\mu$ in
(\ref{Y1E2}), (\ref{Y1E3}) are well defined for $r>0$ and they belong to
$W_{\mathrm{loc}}^{1,1}\left(  0,\infty\right)  .$ Then, they are continuous
in $\left(  0,\infty\right)  $ and they also satisfy:%
\begin{equation}
\lim_{r\rightarrow0^{+}}\lambda\left(  r\right)  =0\ \ ,\ \ \lim
_{r\rightarrow0^{+}}\mu\left(  r\right)  =0 \label{Y1E4}%
\end{equation}
whence $\lambda,\mu\in C\left[  0,\infty\right)  .$ We can differentiate
$\lambda$ and $\mu$ for almost all $r\in\left(  0,\infty\right)  $ and check
by means of one explicit computation that they solve (\ref{S1E3}),
(\ref{S1E4}) $a.e.\ r\in\left(  0,\infty\right)  .$

It only remains to check the continuity of the mapping $\left(  \rho,p\right)
\rightarrow\left(  \lambda,\mu\right)  $ defined from $\mathcal{Z}_{\delta
_{0}}$ to $\left(  W^{1,1}\left(  0,\infty\right)  \right)  ^{2}.$ Suppose
that $\left(  \bar{\rho},\bar{p}\right)  \in\mathcal{Z}_{\delta_{0}}$ and let
us write $\bar{R}_{0}\left(  r\right)  =8\pi\int_{0}^{r}\xi^{2}\bar{\rho
}\left(  \xi\right)  d\xi.$ By assumption $\bar{R}_{0}\left(  r\right)  =0$ if
$r\leq\delta_{0}$ and $\bar{R}_{0}\left(  r\right)  <r$ if $r>\delta_{0}.$
Moreover, since $\bar{\rho}\in L^{1}\left(  \mathbb{R}_{+},r^{2}dr\right)  $
we have that $\bar{R}_{0}\left(  r\right)  $ is bounded for large $r.$ Then,
there exists $\eta>0$ small such that $\bar{R}_{0}\left(  r\right)  <\left(
1-2\eta\right)  r$ if $r>\delta_{0}.$ If we choose $\rho,\ p$ supported in
$\left\{  r\geq\delta_{0}\right\}  $ such that $\int_{0}^{\infty}\left\vert
\rho\left(  r\right)  -\bar{\rho}\left(  r\right)  \right\vert r^{2}dr$ is
small, it then follows that $R_{0}\left(  r\right)  =0$ if $r\leq\delta_{0}$
and $R_{0}\left(  r\right)  <\left(  1-\eta\right)  r$ if $r>\delta$.
Moreover, we have also $\sup_{r\geq\delta_{0}}\frac{\left\vert R_{0}\left(
r\right)  -\bar{R}_{0}\left(  r\right)  \right\vert }{r}$ small. Then, if
$\left(  \bar{\rho},\bar{p}\right)  \rightarrow\left(  \bar{\lambda},\bar{\mu
}\right)  $ we obtain, after differentiating (\ref{Y1E2}), (\ref{Y1E3}) and
using Taylor's Theorem:%
\[
\left\vert \lambda_{r}\left(  r\right)  -\bar{\lambda}_{r}\left(  r\right)
\right\vert \leq\frac{1}{2}\left\vert \frac{1-R_{0}^{\prime}\left(  r\right)
}{r-R_{0}\left(  r\right)  }-\frac{1-\bar{R}_{0}^{\prime}\left(  r\right)
}{r-\bar{R}_{0}\left(  r\right)  }\right\vert \ \ ,\ \ a.e.\ \,r\in\left(
0,\infty\right)
\]%
\[
\left\vert \mu_{r}\left(  r\right)  -\bar{\mu}_{r}\left(  r\right)
\right\vert \leq4\pi\left\vert \frac{r^{2}p\left(  r\right)  }{r-R_{0}\left(
r\right)  }-\frac{r^{2}\bar{p}\left(  r\right)  }{r-\bar{R}_{0}\left(
r\right)  }\right\vert +\frac{1}{2}\left\vert \frac{R_{0}\left(  r\right)
}{\left(  r-R_{0}\left(  r\right)  \right)  r}-\frac{\bar{R}_{0}\left(
r\right)  }{\left(  r-\bar{R}_{0}\left(  r\right)  \right)  r}\right\vert
\]
whence:%
\[
\int_{0}^{\infty}\left\vert \lambda_{r}\left(  r\right)  -\bar{\lambda}%
_{r}\left(  r\right)  \right\vert dr\leq C\int_{\delta_{0}}^{\infty}%
\frac{\left\vert R_{0}\left(  r\right)  -\bar{R}_{0}\left(  r\right)
\right\vert }{r^{2}}dr+C\int_{\delta_{0}}^{\infty}\left\vert \rho\left(
r\right)  -\bar{\rho}\left(  r\right)  \right\vert dr
\]%
\begin{align*}
\int_{0}^{\infty}\left\vert \mu_{r}\left(  r\right)  -\bar{\mu}_{r}\left(
r\right)  \right\vert dr  &  \leq C\int_{\delta_{0}}^{\infty}r\left\vert
p\left(  r\right)  -\bar{p}\left(  r\right)  \right\vert dr+C\int_{\delta_{0}%
}^{\infty}p\left(  r\right)  \left\vert R_{0}\left(  r\right)  -\bar{R}%
_{0}\left(  r\right)  \right\vert dr\\
&  +C\int_{\delta_{0}}^{\infty}\frac{\left\vert R_{0}\left(  r\right)
-\bar{R}_{0}\left(  r\right)  \right\vert }{r^{2}}dr
\end{align*}

Then, if $\int_{0}^{\infty}\left\vert \rho\left(  r\right)  -\bar{\rho}\left(
r\right)  \right\vert r^{2}dr+\int_{0}^{\infty}\left\vert p\left(  r\right)
-\bar{p}\left(  r\right)  \right\vert r^{2}dr$ is small we have that $\int
_{0}^{\infty}\left\vert \lambda_{r}\left(  r\right)  -\bar{\lambda}_{r}\left(
r\right)  \right\vert dr$ and $\int_{0}^{\infty}\left\vert \mu_{r}\left(
r\right)  -\bar{\mu}_{r}\left(  r\right)  \right\vert dr$ are small too. Since
$\left(  \lambda\left(  \delta_{0}\right)  -\bar{\lambda}\left(  \delta
_{0}\right)  \right)  =\left(  \mu\left(  \delta_{0}\right)  -\bar{\mu}\left(
\delta_{0}\right)  \right)  =0$ it then follows that $\left(  \lambda
,\mu\right)  $ and $\left(  \bar{\lambda},\bar{\mu}\right)  $ are close in
$\left(  W^{1,1}\left(  0,L\right)  \right)  ^{2}$ for any $L>0$ and the
result follows.
\end{proof}

One of the technical difficulties that we have to deal with is the fact that
the support of the measure $f$ contains turning points. More precisely, there
are two admissible velocities $v_{1}\left(  r,t\right)  ,v_{2}\left(
r,t\right)  $ if $r>y_{0}\left(  -t\right)  $ and no admissible velocities if
$r<y_{0}\left(  -t\right)  .$ In a neighbourhood of $r=y_{0}\left(  -t\right)
$ quantities like $\rho$ and $p$ (and then $\mu_{r},\ \lambda_{t}$) are
unbounded. Due to this it is not clear in which sense a measure $f$ is a
solution of (\ref{S3E2}) unless some continuity assumptions are made in some
of the functions appearing in (\ref{S3E2}). These continuity assumptions will
provide a relation between the motion of the turning point $r=y_{0}\left(
-t\right)  $, and the functions $v_{1}\left(  r,t\right)  ,v_{2}\left(
r,t\right)  ,B_{1}\left(  r,t\right)  ,B_{2}\left(  r,t\right)  $. The precise
continuity assumptions needed to give a meaning to the solutions of
(\ref{S3E2}) are studied in Lemmas \ref{testCont} and \ref{UniqExt}.

\begin{lemma}
\label{testCont}Let $\mathcal{Z}_{\delta_{0}}$ be as in Lemma \ref{defFields}.
Suppose that $\rho,\ p\in C\left(  \left[  0,T\right]  ;\mathcal{Z}%
_{\delta_{0}}\right)  $ for $T<\infty$ and some $\delta_{0}>0.$ Let
$\lambda,\ \mu$ be as in (\ref{Y1E2}), (\ref{Y1E3}). Let us define a new
variable $\bar{v}$ by means of $v=\bar{v}e^{-\lambda}.$ Let us assume also
that the function
\begin{equation}
\Psi\left(  t,r,\bar{v}\right)  =\left[  -e^{\mu-2\lambda}\bar{v}^{2}%
\rho+\left(  \bar{v}^{2}e^{-2\lambda}+\frac{1}{r^{2}}\right)  p\right]
\label{Y1E7}%
\end{equation}
is continuous in a set $S\subset\left[  0,T\right]  \times\left(
0,\infty\right)  \times\left(  -\infty,\infty\right)  .$ Suppose that
$\varphi=\varphi\left(  t,r,\bar{v},F\right)  \in C_{0}^{1}\left(  \left[
0,T\right]  \times\left[  0,\infty\right)  \times\left(  -\infty
,\infty\right)  \times\left(  0,\infty\right)  \right)  .$ Then the function
$\Delta\left(  t,r,\bar{v},F\right)  $ defined by means of:%
\begin{equation}
\Delta\left(  t,r,\bar{v},F\right)  =\partial_{t}\varphi\left(  t,r,\bar
{v},F\right)  +\partial_{r}\left(  e^{\mu-2\lambda}\frac{\bar{v}}{\tilde{E}%
}\varphi\right)  -\partial_{\bar{v}}\left(  \left(  -\frac{\lambda_{r}%
e^{\mu-2\lambda}\bar{v}^{2}}{\tilde{E}}+e^{\mu}\mu_{r}\tilde{E}-e^{\mu}%
\frac{1}{r^{3}\tilde{E}}\right)  \varphi\right)  \label{S4E1a}%
\end{equation}
is continuous in $S\subset\left[  0,T\right]  \times\left[  0,\infty\right)
\times\left(  -\infty,\infty\right)  \times\left[  0,\infty\right)  .$
\end{lemma}

\begin{proof}
Notice that $\tilde{E}=\sqrt{\bar{v}^{2}e^{-2\lambda}+\frac{1}{r^{2}}}.$ The
function $\partial_{t}\varphi$ is continuous. Due to Lemma \ref{defFields}
$\lambda$ and $\mu$ are also continuous in $\left(  t,r\right)  \in\left[
0,T\right]  \times\left(  0,\infty\right)  $. Then, we only need to check the
continuity of%
\[
\partial_{r}\left(  e^{\mu-2\lambda}\frac{\bar{v}}{\tilde{E}}\varphi\right)
-\partial_{\bar{v}}\left(  \left(  -\frac{\lambda_{r}e^{\mu-2\lambda}\bar
{v}^{2}}{\tilde{E}}+e^{\mu}\mu_{r}\tilde{E}-e^{\mu}\frac{1}{r^{3}\tilde{E}%
}\right)  \varphi\right)
\]

Notice that, due to the differentiability of $\varphi$ and the continuity
of $\lambda,\ \mu$ we just need to prove the continuity of the functions:%
\begin{align}
\Delta_{1}^{\ast}\left(  t,r,\bar{v},F\right)   &  =\left(  -\frac{\lambda
_{r}e^{\mu-2\lambda}\bar{v}^{2}}{\tilde{E}}+e^{\mu}\mu_{r}\tilde{E}\right)
\label{Y1E5}\\
\Delta_{2}^{\ast}\left(  t,r,\bar{v},F\right)   &  =\partial_{r}\left(
e^{\mu-2\lambda}\frac{\bar{v}}{\tilde{E}}\right)  -\partial_{\bar{v}}\left(
-\frac{\lambda_{r}e^{\mu-2\lambda}\bar{v}^{2}}{\tilde{E}}+e^{\mu}\mu_{r}%
\tilde{E}\right)  \label{Y1E6}%
\end{align}

The continuity of $\Delta_{1}^{\ast}$ in $\left\{  r>0\right\}  $ is
equivalent to the continuity of \linebreak$\left(  -\lambda_{r}e^{\mu
-2\lambda}\bar{v}^{2}+\mu_{r}\left(  \bar{v}^{2}e^{-2\lambda}+\frac{1}{r^{2}%
}\right)  \right)  $ in the same region$.$ Using (\ref{S1E3}), (\ref{S1E4}) we
can rewrite this function as:%
\[
\frac{1}{2r}\left(  -e^{\mu-2\lambda}\bar{v}^{2}\left[  \left(  8\pi r^{2}%
\rho-1\right)  e^{2\lambda}+1\right]  +\left(  \bar{v}^{2}e^{-2\lambda}%
+\frac{1}{r^{2}}\right)  \left[  \left(  8\pi r^{2}p+1\right)  e^{2\lambda
}-1\right]  \right)
\]

Since the functions $\lambda$ and $\mu$ are continuous, we just need to check
the continuity in $\left\{  r>0\right\}  $ of
\[
8\pi r^{2}\left[  -e^{\mu-2\lambda}\bar{v}^{2}\rho+\left(  \bar{v}%
^{2}e^{-2\lambda}+\frac{1}{r^{2}}\right)  p\right]  =8\pi r^{2}\Psi\left(
t,r,\bar{v}\right)
\]
and due to the continuity of $\Psi$ it then follows that $\Delta_{1}^{\ast}$
is continuous.

On the other hand, expanding the derivatives in (\ref{Y1E6}) we can see that
the continuity of $\Delta_{2}^{\ast}$ is equivalent to the continuity of the
function:%
\[
\frac{\bar{v}e^{\mu-2\lambda}}{\tilde{E}}\left(  \mu_{r}-2\lambda_{r}\right)
+e^{\mu-2\lambda}\frac{\bar{v}^{3}e^{-2\lambda}\lambda_{r}}{\tilde{E}^{3}%
}+\frac{2\lambda_{r}e^{\mu-2\lambda}\bar{v}}{\tilde{E}}-\frac{\lambda
_{r}e^{\mu-2\lambda}\bar{v}^{3}e^{-2\lambda}}{\tilde{E}^{3}}-\frac{e^{\mu}%
\mu_{r}e^{-2\lambda}\bar{v}}{\tilde{E}}%
\]
which turns out to be identically zero. This concludes the proof of the Lemma.
\end{proof}

\begin{remark}
The continuity condition for $\Psi\left(  t,r,\bar{v}\right)  $ in
(\ref{Y1E7}) will be satisfied for the solutions obtained in this paper using
(\ref{S1E7}), (\ref{S1E8}).
\end{remark}

\begin{remark}
The assumption that the function $\Psi$ is continuous in the set $S$ is a very
strong constraint about the shape of this set. This assumption gives
information about the points of the support of $S$ where the coordinate $r$
reaches its minimum. Heuristically, these are the points where "shell
crossing" takes place. Notice that we cannot expect the functions $\Psi$ and
$\Delta$ to be continuous in any neighbourhood of one of such points. However,
the functions $\Psi$ and $\Delta$ restricted to the set $S$ can be continuous
if this set is chosen in a suitable way. Since the functions $p$ and $\rho$ as
well as the fields $\lambda,\mu$ depend on $S$ the continuity condition yields
information about the possible geometry of this set near the points where $r$
is minimum. As indicated before Lemma \ref{testCont}, the continuity of the
function $\Delta$ will be needed in order to define measure-valued solutions
supported in $S$ for (\ref{S3E2}), (\ref{S3E5}).
\end{remark}

We need to be able to define integrals of measures supported on sets
$S\subset\left[  T_{1},T_{2}\right]  \times\left[  0,\infty\right)
\times\left(  -\infty,\infty\right)  \times\left[  0,\infty\right)  $ with
$-\infty<T_{1}<T_{2}<\infty.$

\begin{definition}
\label{DefIntS}Suppose that $f$ is a Radon measure valued function $f=f\left(
t,r,v,F\right)  \in C\left(  \left[  T_{1},T_{2}\right]  ,\mathcal{M}%
_{+}\left(  \mathbb{R}_{+}\times\mathbb{R}\times\mathbb{R}_{+}\right)
\right)  ,$ $-\infty<T_{1}<T_{2}<\infty.$ Let $S\subset\left[  T_{1}%
,T_{2}\right]  \times\left[  0,\infty\right)  \times\left(  -\infty
,\infty\right)  \times\left[  0,\infty\right)  $ the support of $f.$ Suppose
that $\psi=\psi\left(  t,r,v,F\right)  \in C_{0}^{1}\left(  S\right)  $. Let
us denote as $\bar{\psi}$ any function $\bar{\psi}\in C_{0}^{1}\left(  \left[
T_{1},T_{2}\right]  ,\left[  0,\infty\right)  \times\left(  -\infty
,\infty\right)  \times\left(  0,\infty\right)  \right)  $ such that $\bar
{\psi}\left(  t,r,v,F\right)  =\psi\left(  t,r,v,F\right)  $ for any $\left(
t,r,v,F\right)  \in S.$ We define the integral $\int\int_{S}f\psi drdvdFdt$
as:%
\begin{equation}
\int\int_{S}f\psi drdvdFdt=\int_{T_{1}}^{T_{2}}\int_{\mathbb{R}_{+}%
\times\mathbb{R}\times\mathbb{R}_{+}}f\bar{\psi}drdvdFdt \label{DeI}%
\end{equation}

\end{definition}

The existence of at least one extension $\bar{\psi}$ of the function $\psi$ as
indicated in Definition \ref{DefIntS} follows from standard analysis results.
We now prove that the Definition \ref{DefIntS} is independent of the extension
$\bar{\psi}$ used, i.e. the function $\psi,$ which is defined in $S,$
characterizes uniquely the value of $\int\int_{S}f\psi.$ This is proved in the
following Lemma.

\begin{lemma}
\label{UniqExt}Suppose that $\bar{\psi}_{1},\bar{\psi}_{2}$ are two extensions
of the function $\psi$ as stated in Definition \ref{DefIntS}. Then:%
\[
\int_{T_{1}}^{T_{2}}\int_{\mathbb{R}_{+}\times\mathbb{R}\times\mathbb{R}_{+}%
}f\bar{\psi}_{1}drdvdFdt=\int_{T_{1}}^{T_{2}}\int_{\mathbb{R}_{+}%
\times\mathbb{R}\times\mathbb{R}_{+}}f\bar{\psi}_{2}drdvdFdt
\]

\end{lemma}

\begin{proof}
Let $\varepsilon>0.$ Suppose that $B_{R}\left(  0\right)  $ is a large ball in
$\left(  \mathbb{R}\right)  ^{4}$ containing the support of the functions
$\bar{\psi}_{1},\bar{\psi}_{2}.$ Using the continuity of these functions, as
well as the compactness of the set $S\cap\overline{B_{R}\left(  0\right)  }$
it follows that there exist a finite family of balls $B_{\delta}\left(
\xi_{j}\right)  \subset\left(  \mathbb{R}\right)  ^{4}%
,\ j=1,...,N_{\varepsilon},$ $\xi_{j}=\left(  t_{j},r_{j},v_{j},F_{j}\right)
$with $\delta>0$ depending on $\varepsilon,$ such that $S\cap\overline
{B_{R}\left(  0\right)  }\subset\bigcup_{j=1}^{N_{\varepsilon}}B_{\delta
}\left(  \xi_{j}\right)  $ and $\left\vert \bar{\psi}_{\ell}\left(
y_{1}\right)  -\bar{\psi}_{\ell}\left(  y_{2}\right)  \right\vert
<\varepsilon$ for $\ell=1,2,\ y_{1},y_{2}\in B_{\delta}\left(  \xi_{j}\right)
.$ Moreover, we can assume also that $B_{\delta}\left(  \xi_{j}\right)  \cap
S\neq\varnothing$ for any $j=1,...,N_{\varepsilon}.$ We construct a partition
of the unity $\left\{  \zeta_{k}\right\}  _{k=1}^{N_{\varepsilon}}$ such that
$\sum_{j=1}^{N_{\varepsilon}}\zeta_{j}=1$ in $\bigcup_{k=1}^{N_{\varepsilon}%
}\overline{B_{\delta}\left(  \xi_{k}\right)  },\ \zeta_{j}\geq0$ and
$\sum_{j=1}^{N_{\varepsilon}}\zeta_{j}=0$ at $\xi=\left(  t,r,v,F\right)  $ if
$\operatorname*{dist}\left(  \xi,\bigcup_{k=1}^{N_{\varepsilon}}%
\overline{B_{\delta}\left(  \xi_{k}\right)  }\right)  \geq1.$\ Notice that,
since $B_{\delta}\left(  \xi_{j}\right)  \cap S\neq\varnothing$ and $\bar
{\psi}_{1}\left(  y\right)  =\bar{\psi}_{2}\left(  y\right)  $ if $y\in S$ we
have that $\left\vert \bar{\psi}_{1}\left(  y\right)  -\bar{\psi}_{2}\left(
y\right)  \right\vert <2\varepsilon$ for any $y\in B_{\delta}\left(  \xi
_{j}\right)  .$ We then have:%
\begin{align*}
&  \left\vert \int_{T_{1}}^{T_{2}}\int_{\left(  \mathbb{R}_{+}\times
\mathbb{R}\times\mathbb{R}_{+}\right)  \cap B_{R}\left(  0\right)  }f\left(
\bar{\psi}_{1}-\bar{\psi}_{2}\right)  drdvdFdt\right\vert \\
&  =\left\vert \sum_{j=1}^{N_{\varepsilon}}\int_{T_{1}}^{T_{2}}\int_{\left(
\mathbb{R}_{+}\times\mathbb{R}\times\mathbb{R}_{+}\right)  \cap B_{R}\left(
0\right)  }f\left(  \bar{\psi}_{1}-\bar{\psi}_{2}\right)  \zeta_{j}%
drdvdFdt\right\vert \\
&  \leq\varepsilon\sum_{j=1}^{N_{\varepsilon}}\int_{T_{1}}^{T_{2}}%
\int_{\left(  \mathbb{R}_{+}\times\mathbb{R}\times\mathbb{R}_{+}\right)  \cap
B_{R}\left(  0\right)  }f\zeta_{j}drdvdFdt\\
&  \leq\varepsilon\int_{T_{1}}^{T_{2}}\int_{\left(  \mathbb{R}_{+}%
\times\mathbb{R}\times\mathbb{R}_{+}\right)  \cap B_{R}\left(  0\right)
}fdrdvdFdt
\end{align*}
and since $\varepsilon$ is arbitrarily small the result follows.
\end{proof}

We can now define our concept of measure-valued solution for the spherically
symmetric Einstein-Vlasov system.

\begin{definition}
\label{fWeak}Given a Radon measure $f=f\left(  t,r,v,F\right)  \in C\left(
\left[  T_{1},T_{2}\right]  ,\mathcal{M}_{+}\left(  \mathbb{R}_{+}%
\times\mathbb{R}\times\mathbb{R}_{+}\right)  \right)  ,$ $-\infty<T_{1}%
<T_{2}<\infty$ suppose that $\rho,\ p$ defined by means of (\ref{S1E7}),
(\ref{S1E8}) are in $C\left(  \left[  T_{1},T_{2}\right]  ;\mathcal{Z}%
_{\delta_{0}}\right)  $ for some $\delta_{0}>0,$ and that the functions
$\lambda,\ \mu$ are given by (\ref{Y1E2}), (\ref{Y1E3}) for each $t\in\left[
0,T\right]  .$ Suppose also that for any set of the form $\mathcal{U}=\left\{
\left(  t,r,v,F\right)  :\left(  t,r\right)  \in A,\ v\in\mathbb{R}%
,\ F\in\mathbb{R}_{+}\right\}  $ with $A$ measurable in $\left[  T_{1}%
,T_{2}\right]  \times\mathbb{R}_{+}$ we have $\int\int_{\mathcal{U}}%
Ff<\infty.$ Let us denote as $\tilde{f}\left(  t,r,\bar{v},F\right)  $ the
measure defined by means of:%
\begin{equation}
\tilde{f}\left(  t,r,\bar{v},F\right)  =f\left(  t,r,\bar{v}e^{-\lambda
},F\right)  \ \label{S4E0a}%
\end{equation}
and let us denote the support of $\tilde{f}$ as $S$.  We will say that $f$ is a
solution of (\ref{S1E1})-(\ref{S1E8}) in the sense of measures in the interval
$t\in\left[  T_{1},T_{2}\right]  $ if the function $\Psi\left(  t,r,\bar
{v}\right)  $ defined in (\ref{Y1E7}) is continuous in $S\subset\left[
T_{1},T_{2}\right]  \times\left(  0,\infty\right)  \times\left(
-\infty,\infty\right)  \times\left[  0,\infty\right)  $ and for any test
function $\varphi=\varphi\left(  t,r,\bar{v},F\right)  \in C_{0}^{1}\left(
T_{1},T_{2},\times\left[  0,\infty\right)  \times\left(  -\infty
,\infty\right)  \times\left[  0,\infty\right)  \right)  $ the following
identity holds:%
\begin{equation}
\int_{\mathbb{R}_{+}\times\mathbb{R}\times\mathbb{R}_{+}}\tilde{f}\left(
T_{1},r,\bar{v},F\right)  \varphi\left(  T_{1},r,\bar{v},F\right)  drd\bar
{v}dF+\int\int_{S}\tilde{f}\left(  t,r,\bar{v},F\right)  \Delta\left(
t,r,\bar{v}e^{-\lambda t},F\right)  drd\bar{v}dFdt=0\ \label{S4E1}%
\end{equation}
where $\tilde{E}=\sqrt{\bar{v}^{2}e^{-2\lambda}+\frac{1}{r^{2}}}$,
$\Delta\left(  r,v,F,t\right)  $ is as in (\ref{S4E1a}) and the integral
(\ref{S4E1}) is understood in the sense of Definition \ref{DefIntS}.
\end{definition}

\begin{remark}
\label{defCM}Changes of variables in measures are defined, in the usual
manner, by means of the change of variables over the test function, i.e., the
measure defined in (\ref{S4E0a}) must be understood as:%
\[
\int\tilde{f}\left(  t,r,\bar{v},F\right)  \varphi\left(  t,r,\bar
{v},F\right)  dtdrd\bar{v}dF=\int f\left(  t,r,v,F\right)  \tilde{\varphi
}\left(  t,r,v,F\right)  dtdrdvdF
\]
for any test function $\varphi$ where $\tilde{\varphi}\left(  t,r,v,F\right)
=\varphi\left(  t,r,ve^{\lambda},F\right)  e^{\lambda}.$
\end{remark}

\begin{remark}
Note that the second integral in (\ref{S4E1}) is well defined due to Lemmas
\ref{testCont} and \ref{DefIntS}.
\end{remark}

\begin{remark}
Notice that this definition excludes the possibility of the support of $f$
reaching $r=0.$
\end{remark}

\begin{remark}
The reason to use the variable $v=\bar{v}e^{-\lambda},$ in (\ref{S4E0a}) as
well as in (\ref{S4E1}) is because with this change of variables the
regularity assumptions required for the fields $\lambda$ and $\mu$ are
smaller. This will become apparent in Subsection \ref{vB}, because the use of
the variable $\bar{v}$ in Definition \ref{fWeak} is equivalent to the change
of variables (\ref{Y2E3}) there. This change of variables allows to eliminate
a term $\lambda_{\tau}$ in the equations for the evolution of the particle densities.
\end{remark}

We also need to define weak solutions for (\ref{S1E3})-(\ref{S1E6}),
(\ref{S3E7}), (\ref{S3E10}).

\begin{definition}
\label{zetaWeak}Given a Radon measure $\zeta=\zeta\left(  t,r,v\right)  \in
C\left(  \left[  T_{1},T_{2}\right]  ,\mathcal{M}_{+}\left(  \mathbb{R}%
_{+}\times\mathbb{R}\right)  \right)  ,$ $-\infty<T_{1}<T_{2}<\infty
,\ $suppose that $\rho,\ p$ defined by means of (\ref{S3E10}) are in $C\left(
\left[  T_{1},T_{2}\right]  ;\mathcal{Z}_{\delta_{0}}\right)  $ for some
$\delta_{0}>0,$ and that the functions $\lambda,\ \mu$ are given by
(\ref{Y1E2}), (\ref{Y1E3}) for each $t\in\left[  T_{1},T_{2}\right]  .$ Let us
denote as $\tilde{\zeta}\left(  t,r,\bar{v}\right)  $ the measure defined by
means of:%
\begin{equation}
\tilde{\zeta}\left(  t,r,\bar{v}\right)  =\zeta\left(  t,r,\bar{v}e^{-\lambda
}\right)  \ \ ,\ \ \ v=\bar{v}e^{-\lambda} \label{M2}%
\end{equation}
\ Let us denote the support of $\zeta$ as $S$.  We will say that $\zeta$ is a
solution of (\ref{S1E3})-(\ref{S1E6}), (\ref{S3E7}), (\ref{S3E10}) in the
sense of measures in the interval $t\in\left[  T_{1},T_{2}\right]  $ if the
function $\Psi\left(  t,r,\bar{v}\right)  $ defined in (\ref{Y1E7}) is
continuous in $S\subset\left[  T_{1},T_{2}\right]  \times\left(
0,\infty\right)  \times\left(  -\infty,\infty\right)  $ and for any test
function $\bar{\varphi}=\bar{\varphi}\left(  t,r,\bar{v}\right)  \in
C_{0}\left(  \left[  T_{1},T_{2}\right]  ,\left[  0,\infty\right)
\times\left(  -\infty,\infty\right)  \right)  $ the following identity holds:%
\begin{align}
&  \int_{\mathbb{R}_{+}\times\mathbb{R}}\tilde{\zeta}\left(  T_{1},r,\bar
{v}\right)  \bar{\varphi}\left(  0,r,\bar{v}\right)  drd\bar{v}+\int\int
_{S}\tilde{\zeta}\left(  t,r,\bar{v}\right)  \bar{\Delta}\left(  t,r,\bar
{v}\right)  drd\bar{v}dt\nonumber\\
&  =0\ \ \label{S4E2}%
\end{align}
where:%
\begin{equation}
\bar{\Delta}\left(  t,r,\bar{v}\right)  =\partial_{t}\bar{\varphi}\left(
t,r,\bar{v}\right)  +\partial_{r}\left(  e^{\mu-2\lambda}\frac{\bar{v}}%
{\tilde{E}}\bar{\varphi}\right)  -\partial_{\bar{v}}\left(  \left(
-\frac{\lambda_{r}e^{\mu-2\lambda}\bar{v}^{2}}{\tilde{E}}+e^{\mu}\mu_{r}%
\tilde{E}-e^{\mu}\frac{1}{r^{3}\tilde{E}}\right)  \bar{\varphi}\right)
\label{Y1E8}%
\end{equation}

\end{definition}

\begin{remark}
The measure $\tilde{\zeta}$ must be understood in a manner similar to the one
in Remark \ref{defCM}, with minor changes due to the fact that we consider
functions and measures in a space with one variable less.
\end{remark}

\begin{remark}
Lemma \ref{testCont} applied to the test function $\varphi,$ which is
independent of $F$, implies that the function $\Delta$ is continuous and
therefore, the second integral in (\ref{S4E2}) is well defined.
\end{remark}

\begin{remark}
A definition of weak solutions for the one-dimensional Vlasov-Poisson system
has been given in \cite{MMZ}, \cite{MZ}. The definition in that paper allows
to give a meaning to the solutions of that system in the cases in which the
density $f$ belongs to some particular class of measures, including Dirac
masses. It is not obvious if it is possible to adapt the definition used in
\cite{MMZ}, \cite{MZ} to the Einstein-Vlasov system, and if the resulting
definition would be equivalent to the concept of weak solution introduced in
the Definitions \ref{fWeak}, \ref{zetaWeak}. Notice that we do not try to
define a concept of solution for densities $f$ containing Dirac masses, but
just for measures $f$ supported in surfaces in the space $\left(
r,v,F\right)  $ or measures $\zeta$ supported in a curve in the plane $\left(
r,v\right)  $ but having possible turning points. The assumption about the
continuity of the function $\Psi$ in (\ref{Y1E7}) determines the motion of the
turning points.
\end{remark}

It is relevant to characterize the relation between the measure-valued
solutions of (\ref{S1E1})-(\ref{S1E8}) and the measure-valued solutions of
(\ref{S1E3})-(\ref{S1E6}), (\ref{S3E7}), (\ref{S3E10}).

\begin{proposition}
\label{Equiv1}Suppose that $f=f\left(  t,r,v,F\right)  \in C\left(  \left[
T_{1},T_{2}\right]  ,\mathcal{M}_{+}\left(  \mathbb{R}_{+}\times
\mathbb{R}\times\mathbb{R}_{+}\right)  \right)  ,\ $\linebreak$-\infty
<T_{1}<T_{2}<\infty,$ is a solution of (\ref{S1E1})-(\ref{S1E8}) in the sense
of measures in the interval $t\in\left[  T_{1},T_{2}\right]  $ with initial
datum $f_{0}\left(  r,v,F\right)  =f\left(  T_{1},r,v,F\right)  $ (cf.
Definition \ref{fWeak}). We define $\zeta=\zeta\left(  t,r,v\right)  \in
C\left(  \left[  T_{1},T_{2}\right]  ,\mathcal{M}_{+}\left(  \mathbb{R}%
_{+}\times\mathbb{R}\right)  \right)  $ by means of (\ref{S3E6}). Then $\zeta$
is a solution of (\ref{S1E3})-(\ref{S1E6}), (\ref{S3E7}), (\ref{S3E10}) in the
sense of measures in the interval $t\in\left[  T_{1},T_{2}\right]  $ with
initial datum $\zeta\left(  T_{1},r,v\right)  =\int_{0}^{\infty}Ff_{0}\left(
r,v,F\right)  dF$ (cf. Definition \ref{zetaWeak}).
\end{proposition}

\begin{proof}
Notice that, due to Definition \ref{fWeak} we have that the functions
$\rho,\ p$ defined by means of (\ref{S3E10}) are in $C\left(  \left[
0,T\right]  ;\mathcal{Z}_{\delta_{0}}\right)  $ for some $\delta_{0}>0.$ Then
the measure $\zeta$ given by (\ref{S3E6}) is well defined. The result then
follows taking as test function in (\ref{S4E1}) a sequence of test functions
$\varphi_{\varepsilon}\left(  t,r,\bar{v},F\right)  =\bar{\varphi}\left(
t,r,\bar{v}\right)  \zeta_{\varepsilon}\left(  F\right)  $ with $\bar{\varphi
}\left(  t,r,\bar{v}\right)  \in C_{0}\left(  \left[  0,\infty\right)
\times\left[  0,\infty\right)  \times\left(  -\infty,\infty\right)  \right)  $
and $\zeta_{\varepsilon}$ satisfying $0\leq\zeta_{\varepsilon}\left(
F\right)  \leq F,\ \zeta_{\varepsilon}\left(  F\right)  \leq C_{\varepsilon
}<\infty,$ $\lim_{\varepsilon\rightarrow0}\zeta_{\varepsilon}\left(  F\right)
=F.$ Since the integrals with the form $\int\int_{\mathcal{U}}Ff$ are finite,
due to Definition \ref{fWeak}. we can take the limit $\varepsilon
\rightarrow0,$ using Lebesgue's Theorem to obtain the result.
\end{proof}

Reciprocally, given $\zeta$ solution of (\ref{S1E3})-(\ref{S1E6}),
(\ref{S3E7}), (\ref{S3E10}) in the sense of the Definition \ref{zetaWeak}, we
can obtain a large class of measures $f$ which solve (\ref{S1E1})-(\ref{S1E8})
in the sense of Definition \ref{fWeak}. The key idea underlying the proof of
the following result is that the angular momentum is constant along the
characteristics associated to the equation (\ref{S1E1}).

\begin{proposition}
\label{fExt}Suppose that $\zeta=\zeta\left(  t,r,v\right)  \in C\left(
\left[  T_{1},T_{2}\right]  ,\mathcal{M}_{+}\left(  \mathbb{R}_{+}%
\times\mathbb{R}\right)  \right)  ,\ $\linebreak$-\infty<T_{1}<T_{2}<\infty,$
is a solution of (\ref{S1E3})-(\ref{S1E6}), (\ref{S3E7}), (\ref{S3E10}) in the
sense of Definition \ref{zetaWeak}. Let us assume that $\xi\in\mathcal{M}%
_{+}\left(  0,\infty\right)  $ is a compactly supported measure. We define
measures $f=f\left(  t,r,v,F\right)  \in C\left(  \left[  T_{1},T_{2}\right]
,\mathcal{M}_{+}\left(  \mathbb{R}_{+}\times\mathbb{R}\times\mathbb{R}%
_{+}\right)  \right)  $ by means of:%
\begin{equation}
f\left(  t,r,v,F\right)  =\zeta\left(  t,r,v\right)  \xi\left(  F\right)
\label{M1}%
\end{equation}

Then, $f$ is a solution of (\ref{S1E1})-(\ref{S1E8}) in the sense of
Definition \ref{fWeak}.
\end{proposition}

\begin{proof}
We choose a test function $\varphi=\varphi\left(  t,r,\bar{v},F\right)  \in
C_{0}^{1}\left(  \left[  T_{1},T_{2}\right]  \times\left[  0,\infty\right)
\times\left(  -\infty,\infty\right)  \times\left[  0,\infty\right)  \right)  $
(cf. Definition \ref{fWeak}). Using (\ref{M1}) we can rewrite the right-hand
side of (\ref{S4E1}) as:%
\begin{align}
J  & =\int_{\mathbb{R}_{+}\times\mathbb{R}\times\mathbb{R}_{+}}\tilde{\zeta
}\left(  T_{1},r,\bar{v}\right)  \xi\left(  F\right)  \varphi\left(
T_{1},r,\bar{v},F\right)  drd\bar{v}dF\label{M5}\\
& +\int\int_{S}\tilde{\zeta}\left(  t,r,\bar{v}\right)  \xi\left(  F\right)
\Delta\left(  t,r,\bar{v}e^{-\lambda t},F\right)  drd\bar{v}dFdt\nonumber
\end{align}
where $\tilde{\zeta}$ is as in (\ref{M2}). We define:%
\begin{equation}
\bar{\varphi}\left(  t,r,\bar{v}\right)  =\int\xi\left(  F\right)
\varphi\left(  t,r,\bar{v},F\right)  dF\label{M3}%
\end{equation}

Using (\ref{S4E1a}) it then follows that the function $\bar{\Delta}$
associated to $\bar{\varphi}$ is given by:%
\begin{equation}
\bar{\Delta}\left(  t,r,\bar{v}\right)  =\int\xi\left(  F\right)
\Delta\left(  t,r,\bar{v},F\right)  dF \label{M4}%
\end{equation}

Using (\ref{M3}), (\ref{M4}) we obtain that $J$ in (\ref{M5}) is given by:%
\[
J=\int_{\mathbb{R}_{+}\times\mathbb{R}\times\mathbb{R}_{+}}\tilde{\zeta
}\left(  T_{1},r,\bar{v}\right)  \bar{\varphi}\left(  T_{1},r,\bar{v}\right)
drd\bar{v}+\int\int_{S}\tilde{\zeta}\left(  t,r,\bar{v}\right)  \Delta\left(
t,r,\bar{v}e^{-\lambda t},F\right)  drd\bar{v}dt\
\]

Therefore, using (\ref{S4E2}) we obtain $J=0$ whence the result follows.
\end{proof}

\subsection{Main results.}

The main theorem that will be proved in this paper is the following:

\begin{theorem}
\label{main}There exists $t_{0}<0$ and a measure $\zeta=\zeta\left(
t,r,v\right)  \in C\left(  \left[  t_{0},0\right] :\mathcal{M}%
_{+}\left(  \mathbb{R}_{+}\times\mathbb{R}\right)  \right)  $ supported in two
curves $\gamma_{1}\left(  t\right)  ,$\ $\gamma_{2}\left(  t\right)  $ which
can be parametrized in the form:%
\begin{align*}
\gamma_{1}\left(  t\right)   &  =\left\{  v=v_{1}\left(  t,r\right)
\ ,\ r\geq y_{0}\left(  -t\right)  ,\ t_{0}\leq t<0\right\}  \ \\
\gamma_{2}\left(  t\right)   &  =\left\{  v=v_{2}\left(  t,r\right)
\ ,\ y_{0}\left(  -t\right)  \leq r\leq R_{+}\left(  t\right)  ,\ t_{0}\leq
t<0\right\}
\end{align*}
for some $y_{0}>0$ and suitable functions $v_{1}\in C^{1}\left(  \bigcup
_{t\in\left(  t_{0},0\right)  }\left[  \left(  y_{0}\left(  -t\right)
,\infty\right)  \times\left\{  t\right\}  \right]  \right)  ,\ $%
\linebreak$v_{2}\in C^{1}\left(  \bigcup_{t\in\left(  t_{0},0\right)  }\left[
\left(  y_{0}\left(  -t\right)  ,R_{+}\left(  t\right)  \right)
\times\left\{  t\right\}  \right]  \right)  ,\ R_{+}\in C^{1}\left(
t_{0},0\right)  $ and such that, the functions $\rho,\ p$ given by
(\ref{S3E10}) satisfy $\rho,\ p\in L_{\mathrm{loc}}^{\infty}\left(  \left(
y_{0}\left(  -t\right)  ,\infty\right)  \times\left(  t_{0},0\right)  \right)
$ and:%
\begin{align*}
0 &  \leq\lim\sup_{r\rightarrow\left(  y_{0}\left(  -t\right)  \right)  ^{+}%
}\sqrt{r-y_{0}\left(  -t\right)  }\rho\left(  t,r\right)  \leq C\left(
T\right)  <\infty\\
0 &  \leq\lim\sup_{r\rightarrow\left(  y_{0}\left(  -t\right)  \right)  ^{+}%
}\sqrt{r-y_{0}\left(  -t\right)  }p\left(  t,r\right)  \leq C\left(  T\right)
<\infty
\end{align*}
for any $t_{0}<T<0.$ Moreover, $\zeta$ solves (\ref{S1E3})-(\ref{S1E6}),
(\ref{S3E7}), (\ref{S3E10}) in the sense of the Definition \ref{zetaWeak}.

The distribution $\zeta\left(  r,v,t\right)  $ has the following asymptotics:%
\begin{equation}
\zeta\left(t, r,v\right)  =\frac{r\chi_{\left(  0,R_{\max}\right)  }\left(
r\right)  }{12\pi^{2}}\delta\left(  v\right)  +O\left(  \left(  -t\right)
^{b}\exp\left(  -ar\right)  \right)  \ \ \text{as\ \ }t\rightarrow
0\ \label{Dec1}%
\end{equation}
where $R_{\max}>0$ is a fixed number,\ $a>0,\ b>0$ and $\chi_{\left(
0,R_{\max}\right)  }\left(  r\right)  $ is the characteristic function
supported in the interval $\left(  0,R_{\max}\right)  .$The asymptotics
(\ref{Dec1}) must be understood in the sense of distributions, i.e. after
multiplying by a test function.

Moreover, the metric (\ref{met1}) defined using these functions behaves
asymptotically as:%
\[
ds^{2}=-e^{2\mu\left(  t,r\right)  }dt^{2}+e^{2\lambda\left(  t,r\right)
}dr^{2}+r^{2}\left(  d\theta^{2}+\sin^{2}\theta d\varphi^{2}\right)
\]
where:%
\begin{equation}
e^{2\lambda\left(  t,r\right)  }\rightarrow1\ ,\ \ e^{2\mu\left(  t,r\right)
}\rightarrow\ e^{2\mu_{\infty}}\text{\ as\ \ }r\rightarrow\infty
\ \ ,\ \ t\in\left(  t_{0},0\right)  \label{Dec2}%
\end{equation}
for some suitable $\mu_{\infty}>0.$
\end{theorem}

Combining Proposition \ref{fExt} and Theorem \ref{main} we obtain the
following result:

\begin{theorem}
\label{mainF} There exist $t_{0}<0$ and infinitely many measures $f=f\left(
t,r,v,F\right)  \in C\left(  \left(  t_{0},0\right)  ,\mathcal{M}_{+}\left(
\mathbb{R}_{+}\times\mathbb{R}\times\mathbb{R}_{+}\right)  \right)  $
supported in the surfaces $\gamma_{k}\left(  t\right)  \times\mathbb{R}_{+},$
$k=1,2,$ with $\gamma_{1}\left(  t\right)  ,$\ $\gamma_{2}\left(  t\right)  $
as in Theorem \ref{main}, such that $f$ solves (\ref{S1E1})-(\ref{S1E8}) in
the interval $t\in\left(  t_{0},0\right)  $ in the sense of Definition
\ref{fWeak}. The measure $f$ satisfies $\int_{\mathbb{R}_{+}}fdF=\zeta,$ with
$\zeta$ as in Theorem \ref{main}. The measure $\zeta$ satisfies (\ref{Dec1}).
The metric (\ref{met1}) is asymptotically flat as $r\rightarrow\infty$ in the
sense of (\ref{Dec2}).
\end{theorem}

\begin{remark}
It is relevant to remark that the asymptotics (\ref{Dec1}) shows that the
particle distribution behaves asymptotically, for times close to the
singularity, like a particular type of generalized Einstein clusters whose
speed approaches zero for times close to the onset of the singularity. The
gravitational field for large values of $r,$ in the region where the metric
becomes asymptotically flat, is, for long times, the corresponding one to that
generalized Einstein cluster.
\end{remark}

\section{MAIN PROPERTIES OF THE SELF-SIMILAR SOLUTION CONSTRUCTED IN
\cite{RV}. \label{selfsimilar}}

We summarize some of the main properties of the solution of (\ref{S1E3}),
(\ref{S1E4}), (\ref{S3E7}), (\ref{S3E10}) constructed in \cite{RV}. Most of
the results of the next Theorem have been proved in \cite{RV} and we just
reformulate some specific points in a form that is more convenient in order to
obtain the results of this paper.

\begin{theorem}
\label{RV}For any $y_{0}>0$ sufficiently small, there exists a solution of
(\ref{S1E3}), (\ref{S1E4}), (\ref{S3E7}), (\ref{S3E10}) of the form
(\ref{F1E2}), (\ref{F1E3}) where $\Theta\left(  y,V\right)  $ can be written
in the form:%
\begin{equation}
\Theta\left(  y,V\right)  =\beta_{0}e^{2\sigma}\delta\left(  H-h\right)
\ \ ,\ \ H=\frac{e^{U}}{y}\sqrt{V^{2}y^{2}+1}+yVe^{\Lambda},\ \ \label{F1E6}%
\end{equation}
where $\beta_{0}=\beta_{0}\left(  y_{0}\right)  >0,\ \,h=\frac{\sqrt
{1-y_{0}^{2}}}{y_{0}}.$ The fields $\lambda,\ \mu$ have the self-similar form
$\Lambda\left(  y\right)  ,\ U\left(  y\right)  $ in (\ref{F1E2}). The curve
$\left\{  H=h\right\}  \subset\left\{  \left(  y,V\right)  :y\geq
y_{0}\right\}  $ can be decomposed into two portions:%
\begin{equation}
\left\{  H=h\right\}  =\Gamma_{1}\cup\Gamma_{2}\ \ ,\Gamma_{k}=\left\{
\left(  y,V\right)  :V=V_{1}\left(  y\right)  ,\ y\geq y_{0}\right\}
\ \ ,\ \ k=1,2\ \label{F1E6a}%
\end{equation}
with$\ V_{1}\left(  y\right)  <V_{2}\left(  y\right)  $ for $y>y_{0}$ and
$V_{1}\left(  y_{0}\right)  =V_{2}\left(  y_{0}\right)  =-\frac{1}%
{\sqrt{1-y_{0}^{2}}}.$

Each of the curves $\Gamma_{1},\ \Gamma_{2}\ $\ can be parametrized using the
parameter $\sigma$ in (\ref{F1E6}):%
\begin{align*}
\Gamma_{1}  &  =\left\{  \left(  y,V\right)  :y=\bar{y}_{1}\left(
\sigma\right)  ,\ V=\bar{V}_{1}\left(  \sigma\right)  \ ,\ \sigma
\leq0\right\}  \ \\
\Gamma_{2}  &  =\left\{  \left(  y,V\right)  :y=\bar{y}_{2}\left(
\sigma\right)  ,\ V=\bar{V}_{2}\left(  \sigma\right)  \ ,\ \sigma
\geq0\right\}
\end{align*}
where the functions $\bar{y},\ \bar{V}_{k},$ $k=1,2$ solve the ODEs:%
\begin{align*}
\frac{d\bar{y}_{k}}{d\sigma}  &  =\bar{y}_{k}+e^{U\left(  \bar{y}_{k}\right)
-\Lambda\left(  \bar{y}_{k}\right)  }\frac{\bar{V}_{k}}{\mathcal{E}_{k}}\ \ \\
\frac{d\bar{V}_{k}}{d\sigma}  &  =-\bar{V}_{k}-\left(  \bar{y}_{k}\Lambda
_{y}\left(  \bar{y}_{k}\right)  \bar{V}_{k}+e^{U\left(  \bar{y}_{k}\right)
-\Lambda\left(  \bar{y}_{k}\right)  }\mathcal{E}_{k}U_{y}\left(  \bar{y}%
_{k}\right)  -\frac{e^{U\left(  \bar{y}_{k}\right)  -\Lambda\left(  \bar
{y}_{k}\right)  }}{\bar{y}_{k}^{3}\mathcal{E}_{k}}\right)
\end{align*}
with:%
\[
\mathcal{E}_{k}=\sqrt{\frac{1}{\bar{y}_{k}^{2}}+\bar{V}_{k}^{2}}%
\ \ ,\ \ \ \sigma\left(  y_{0};y_{0}\right)  =0
\]
for $k=1,2.$

The measure $\Theta\left(  y,V\right)  $ in (\ref{F1E6}) can be rewritten in
the form:%
\begin{equation}
\Theta\left(  y,V\right)  =b_{1}\left(  y\right)  \delta\left(  V-V_{1}\left(
y\right)  \right)  +b_{2}\left(  y\right)  \delta\left(  V-V_{2}\left(
y\right)  \right)  \ \label{F1E7b}%
\end{equation}
where the functions $V_{1}\left(  \cdot\right)  ,\ V_{2}\left(  \cdot\right)
$ are as in (\ref{F1E6a}) and:
\begin{equation}
b_{1}\left(  y\right)  =\frac{\beta_{0}e^{2\sigma}}{\left\vert \frac{\partial
H}{\partial V}\left(  y,V_{1}\left(  y\right)  \right)  \right\vert
}\ \ \ ,\ \ \ b_{2}\left(  y\right)  =\frac{\beta_{0}e^{2\sigma}}{\left\vert
\frac{\partial H}{\partial V}\left(  y,V_{2}\left(  y\right)  \right)
\right\vert } \label{F1E7}%
\end{equation}

Moreover, if we write $\zeta$ in (\ref{F1E2}) in the form (\ref{F1E0}) we
obtain:%
\begin{equation}
B_{k}\left(  t,r\right)  =\left(  -t\right)  b_{k}\left(  y\right)
\ \ ,\ \ \ v_{k}\left(  t,r\right)  =\frac{V_{k}\left(  y\right)  }{\left(
-t\right)  }\ \ ,\ \ \ k=1,2\ . \label{F1E7a}%
\end{equation}

The following asymptotics hold:%
\begin{align}
U\left(  y\right)   &  =\log\left(  \frac{y}{y_{0}}\right)  +\log\left(
\sqrt{1-y_{0}^{2}}\right)  +O\left(  \frac{1}{y^{\delta}}\right)
\ \ \ \text{as\ \ }y\rightarrow\infty\label{F2E1}\\
\Lambda\left(  y\right)   &  =\log\left(  \sqrt{3}\right)  +O\left(  \frac
{1}{y^{\delta}}\right)  \ \ \ \text{as\ \ }y\rightarrow\infty\label{F2E2}\\
V_{1}\left(  y\right)   &  =-\frac{2y_{0}\sqrt{3\left(  1-y_{0}^{2}\right)  }%
}{\left(  1-4y_{0}^{4}\right)  y}\left(  1+O\left(  \frac{1}{y^{\delta}%
}\right)  \right)  \ \ \ \ \text{as\ \ }y\rightarrow\infty\ \label{F2E3}\\
V_{2}\left(  y\right)   &  =-\frac{\sqrt{1-y_{0}^{2}}}{\sqrt{3}y_{0}}%
\frac{C_{1}}{y}\left(  \frac{y_{0}}{y}\right)  ^{2}\left(  1+o\left(
1\right)  \right)  \ \ \ \ \text{as\ \ }y\rightarrow\infty\label{F2E4}%
\end{align}%
\begin{align}
b_{1}\left(  y\right)   &  \sim A_{1}\left(  y\right)  ^{1-\gamma\left(
y_{0}\right)  }\ \ \ ,\ \ \ b_{2}\left(  y\right)  \sim\frac{y}{12\pi^{2}%
}\ \text{as\ \ }y\rightarrow\infty\label{F2E5}\\
A_{1}  &  =\frac{\beta_{0}\left(  C_{2}\right)  ^{2}\left(  y_{0}\right)
^{\frac{4\left(  1-y_{0}^{2}\right)  }{\left(  1-4y_{0}^{2}\right)  }-2}%
}{\left\vert \frac{\zeta_{1}\sqrt{1-y_{0}^{2}}}{y_{0}\sqrt{\left(  \zeta
_{1}\right)  ^{2}+1}}+\sqrt{3}\right\vert }\ ,\ \gamma\left(  y_{0}\right)
=\frac{4\left(  1-y_{0}^{2}\right)  }{\left(  1-4y_{0}^{2}\right)  }%
\ ,\ \zeta_{1}=-\frac{2y_{0}\sqrt{3\left(  1-y_{0}^{2}\right)  }}{\left(
1-4y_{0}^{2}\right)  } \label{F2E5a}%
\end{align}
for some constants $C_{1},\ C_{2}\in\mathbb{R}$ and $\delta>0,$ depending on
$y_{0}.$ Moreover, for any compact $\mathcal{K}\subset\left(  0,\infty\right)
$ and any positive integer $m,$ there exists a constant $C$ depending on
$\mathcal{K}$ and $m$ such that:%
\begin{equation}
B_{1}\left(  t,r\right)  =A_{1}\left(  -t\right)  ^{\gamma\left(
y_{0}\right)  }r^{1-\gamma\left(  y_{0}\right)  }+O\left(  \left(  -t\right)
^{\gamma\left(  y_{0}\right)  +\delta}\right)  ,\ B_{2}\left(  t,r\right)
=\frac{r}{12\pi^{2}}+O\left(  \left(  -t\right)  ^{\delta}\right)
\ \text{as\ \ }t\rightarrow0^{-} \label{F2E6}%
\end{equation}%
\begin{equation}
\left\vert \frac{\partial^{\ell}}{\partial r^{\ell}}\left(  v_{1}\left(
t,r\right)  -\frac{\zeta_{1}}{r}\right)  \right\vert \leq C\left(  -t\right)
^{\delta},\ \left\vert \frac{\partial^{\ell}}{\partial r^{\ell}}\left(
v_{2}\left(  t,r\right)  +\frac{\sqrt{1-y_{0}^{2}}}{\sqrt{3}y_{0}}\frac
{C_{1}\left(  -t\right)  ^{2}}{r}\left(  \frac{y_{0}}{r}\right)  ^{2}\right)
\right\vert \leq C\left(  -t\right)  ^{\delta}) \label{F2E6b}%
\end{equation}%
\begin{align}
\left\vert \frac{\partial^{\ell}}{\partial r^{\ell}}\left(  \lambda\left(
t,r\right)  -\log\left(  \sqrt{3}\right)  \right)  \right\vert  &  \leq
C\left(  -t\right)  ^{\delta}\ \ \ \text{as\ }t\rightarrow0\ \label{F3E1}\\
\left\vert \frac{\partial^{\ell}}{\partial r^{\ell}}\left(  \mu\left(
t,r\right)  -\log\left(  \frac{1}{\left(  -t\right)  }\right)  -\log\left(
r\right)  -\log\left(  \frac{\sqrt{1-y_{0}^{2}}}{y_{0}}\right)  \right)
\right\vert  &  \leq C\left(  -t\right)  ^{\delta}\ \ \ \ \text{as\ \ }%
t\rightarrow0 \label{F3E2}%
\end{align}
for $\ell=0,1,...,m,$ $r\in\mathcal{K}$ and $\left\vert t\right\vert $ small.
\end{theorem}

\begin{proof}
We first refer to the specific points of Theorem \ref{RV} which have been
proved in \cite{RV}. The representation formula (\ref{F1E7}) just follows from
(\ref{F1E6}) and the definition of the functions $V_{1}\left(  y\right)
,\ V_{2}\left(  y\right)  .$ The asymptotics (\ref{F2E5}) is a consequence of
a representation formula that has been obtained in \cite{RV}, namely:%
\begin{equation}
e^{\sigma_{k}}=\frac{y}{y_{0}}Q_{k}\ \ ,\ \ \ k=1,2\ \label{F2E4b}%
\end{equation}
as well as the following asymptotics (cf. also \cite{RV}):%
\begin{equation}
Q_{1}\sim C_{2}e^{-\frac{2\left(  1-y_{0}^{2}\right)  }{\left(  1-4y_{0}%
^{2}\right)  }s}\ \ \ \text{as\ \ }s\rightarrow\infty\ \ \ ,\ \ \ Q_{2}%
\rightarrow Q_{2,\infty}=\frac{2\sqrt{y_{0}}}{3^{\frac{1}{4}}\sqrt{\theta}%
}\ \ ,\ \ \theta=\frac{16\pi^{2}\beta_{0}}{y_{0}}\ \ ,\ \ \frac{y}{y_{0}%
}=e^{s}\ \label{F2E4a}%
\end{equation}
for some $C_{2}\in\mathbb{R}$. Moreover, using (\ref{F1E6}), (\ref{F2E1}%
)-(\ref{F2E4}) we obtain:%
\begin{align*}
\frac{\partial H}{\partial V}\left(  y,V_{2}\left(  y\right)  \right)   &
\sim\sqrt{3}y\\
\frac{\partial H}{\partial V}\left(  y,V_{1}\left(  y\right)  \right)   &
\sim\left(  \frac{\zeta_{1}\sqrt{1-y_{0}^{2}}}{y_{0}\sqrt{\left(  \zeta
_{1}\right)  ^{2}+1}}+\sqrt{3}\right)  y\ \ ,\ \ \zeta_{1}=-\frac{2y_{0}%
\sqrt{3\left(  1-y_{0}^{2}\right)  }}{\left(  1-4y_{0}^{2}\right)  }%
\end{align*}
as $y\rightarrow\infty.$ Combining (\ref{F2E4b}), (\ref{F2E4a}) we then obtain
(\ref{F2E5}). Using (\ref{F1E7a}) we then obtain (\ref{F2E6}). Finally, the
asymptotics (\ref{F2E6b})-(\ref{F3E2}) is a consequence of (\ref{F2E1}%
)-(\ref{F2E4}).

We now check that the measure $\zeta$ given in (\ref{F1E2}), (\ref{F1E3})
satisfies Definition \ref{zetaWeak}. To this end we rewrite the functions
$\Psi\left(  t,r,\bar{v}\right)  ,$ $\Delta\left(  t,r,\bar{v}\right)  $ in
(\ref{Y1E7}), (\ref{Y1E8}) respectively using the self-similar variables
(\ref{F1E3}). Then:%
\begin{align}
\Psi\left(  t,r,\bar{v}\right)   &  =\frac{\pi}{y^{2}}\frac{1}{\left(
-t\right)  ^{4}}\left[  -e^{U\left(  y\right)  }V^{2}\int_{-\infty}^{\infty
}\sqrt{V^{2}+\frac{1}{y^{2}}}\Theta\left(  y,V\right)  dV\right.
\label{Y2E1}\\
&  +\left.  \left(  V^{2}+\frac{1}{y^{2}}\right)  \int_{-\infty}^{\infty}%
\frac{V^{2}}{\sqrt{V^{2}+\frac{1}{y^{2}}}}\Theta\left(  y,V\right)  dV\right]
\nonumber
\end{align}

On the other hand, suppose that we define $\Phi$ by means of:%
\[
\varphi\left(  t,r,\bar{v}\right)  =\Phi\left(  \tau,y,V\right)
\]
where $y,V$ are as in (\ref{F1E2}) and $\tau=\log\left(  \frac{1}{\left(
-t\right)  }\right)  .$ We can then compute $\Delta\left(  t,r,\bar{v}\right)
$ defined in (\ref{S4E1a}):%
\begin{align}
\Delta\left(  t,r,\bar{v}\right)   &  =\frac{1}{\left(  -t\right)  }\left[
\Phi_{\tau}+y\Phi_{y}-V\Phi_{V}-y\Lambda_{y}V\Phi_{V}\right] \nonumber\\
&  +\left[  \frac{1}{\left(  -t\right)  }\left(  U_{y}-2\Lambda_{y}\right)
e^{U-\Lambda}\frac{V}{\sqrt{V^{2}+\frac{1}{y^{2}}}}+\frac{e^{U-\Lambda}%
}{\left(  -t\right)  }\frac{V}{\sqrt{\left(  V^{2}+\frac{1}{y^{2}}\right)
^{3}}}\left(  \Lambda_{y}V^{2}+\frac{1}{y^{3}}\right)  \right]  \Phi
\nonumber\\
&  +\frac{1}{\left(  -t\right)  }e^{U-\Lambda}\frac{V}{\sqrt{V^{2}+\frac
{1}{y^{2}}}}\left[  \Phi_{y}-\Lambda_{y}V\Phi_{V}\right] \nonumber\\
&  -\frac{\Phi e^{U-\Lambda}}{\left(  -t\right)  }\left[  -\frac{2\Lambda
_{y}V}{\sqrt{V^{2}+\frac{1}{y^{2}}}}+\frac{\Lambda_{y}V^{3}}{\sqrt{\left(
V^{2}+\frac{1}{y^{2}}\right)  ^{3}}}+\frac{U_{y}V}{\sqrt{V^{2}+\frac{1}{y^{2}%
}}}+\frac{V}{y^{3}\sqrt{\left(  V^{2}+\frac{1}{y^{2}}\right)  ^{3}}}\right]
\nonumber\\
&  -\frac{e^{U-\Lambda}}{\left(  -t\right)  }\left(  -\frac{\Lambda_{y}V^{2}%
}{\sqrt{V^{2}+\frac{1}{y^{2}}}}+U_{y}\sqrt{V^{2}+\frac{1}{y^{2}}}-\frac
{1}{y^{3}\sqrt{V^{2}+\frac{1}{y^{2}}}}\right)  \Phi_{V} \label{Y2E2}%
\end{align}

Using (\ref{F1E6}) and (\ref{Y2E1}) we can check that $\Psi\left(  t,r,\bar
{v}\right)  $ is a continuous function on the support of $S$. Indeed, using
(\ref{F1E7b}) we obtain:%
\begin{align*}
\Psi\left(  t,r,\bar{v}\right)   &  =\frac{\pi}{y^{2}}\frac{1}{\left(
-t\right)  ^{4}}\left[  -e^{U\left(  y\right)  }V^{2}\sum_{k=1}^{2}%
b_{k}\left(  y\right)  \sqrt{\left(  V_{k}\left(  y\right)  \right)
^{2}+\frac{1}{y^{2}}}\right. \\
&  \left.  +\left(  V^{2}+\frac{1}{y^{2}}\right)  \sum_{k=1}^{2}b_{k}\left(
y\right)  \frac{\left(  V_{k}\left(  y\right)  \right)  ^{2}}{\sqrt{\left(
V_{k}\left(  y\right)  \right)  ^{2}+\frac{1}{y^{2}}}}\right]
\end{align*}
and now using (\ref{F1E7}) we obtain that $\Psi$ is continuous for $y>y_{0}.$
It only remains to check the continuity of $\Psi$ at the point $\left(
y,V\right)  =\left(  y_{0},V_{0}\right)  ,$ with $V_{0}=-\frac{1}%
{\sqrt{1-y_{0}^{2}}}.$ This can be seen, using the asymptotics of the
functions $U\left(  y\right)  ,\ V_{k}\left(  y\right)  $ and $b_{k}\left(
y\right)  ,$ in a manner similar to the proof of Proposition 5 of \cite{RV}.
\end{proof}

\section{DUST-LIKE SOLUTIONS: DIFFERENTIAL EQUATIONS AND FUNCTION SPACES.}

\subsection{Some definitions.}

As indicated in Section \ref{Plan}, the solution that we will construct agrees
with the self-similar solution described in Section \ref{selfsimilar} for
$r\leq R_{+}\left(  t\right)  .$ On the other hand, the solution that we will
construct has the form (\ref{F1E0}) with $B_{2}\left(  t,r\right)  =0$ for
$r>R_{+}\left(  t\right)  $ and therefore it is not self-similar for $r\geq
R_{+}\left(  t\right)  .$ We now derive the equations that must be satisfied
by the functions $R_{+}\left(  t\right)  ,\ v_{1}\left(  t,r\right)
,\ B_{1}\left(  t,r\right)  \ $in order to obtain a solution of (\ref{S1E3}),
(\ref{S1E4}), (\ref{S3E7}), (\ref{S3E10}). We also prove the existence of some
auxiliary functions which will be needed in the following.

We first observe that the asymptotics (\ref{F3E2}) suggests to introduce a new
time scale in order to describe the region where $r$ is of order one. We
define:%
\begin{equation}
\tau=\log\left(  \frac{1}{\left(  -t\right)  }\right)  \label{F3E5}%
\end{equation}
as well as:%
\begin{equation}
\mu\left(  t,r\right)  =\log\left(  \frac{1}{\left(  -t\right)  }\right)
+\bar{\mu}\left(  t,r\right)  \label{F3E3}%
\end{equation}

It is also convenient to define an auxiliary function $\hat{U}$ by means of
(cf. (\ref{F2E1})):%
\begin{equation}
U\left(  y\right)  =\log\left(  y\right)  +\hat{U}\left(  y\right)
\label{F3E4}%
\end{equation}

It is interesting to remark that for the solution constructed in this paper,
$t$ will be the proper time for a particle fixed at the center $r=0,$ while
$\tau$ is the proper time of a particle at rest at $r=\infty.$

We first describe the behaviour of the function $R_{+}\left(  t\right)
=r_{+}\left(  \tau\right)  .$ The point $\left(  r,v\right)  =\left(
R_{+}\left(  t\right)  ,v_{2}\left(  t,R_{+}\left(  t\right)  \right)
\right)  $ is a point where there is a discontinuity of the density $B_{2}$
associated to the measure $\zeta.$ In order to obtain a weak solution of
(\ref{S3E7}) the point $\left(  R_{+}\left(  t\right)  ,v_{2}\left(
t,R_{+}\left(  t\right)  \right)  \right)  $ must move along characteristics.
Given that the fields $\lambda,\ \mu$ are continuous at the point
$r=R_{+}\left(  t\right)  $ we can use the values of the self-similar fields
in (\ref{F1E2}), (\ref{F1E3}) (cf. also Theorem \ref{RV}). In order to fix the
form of this function we need to impose an additional condition. We will
assume that:%
\begin{equation}
\lim_{t\rightarrow0^{+}}R_{+}\left(  t\right)  =R_{\max}\ \label{F2E7a}%
\end{equation}
for some $R_{\max}>0.$ Given the form of the characteristic curves associated
to (\ref{S3E7}) we define a function $R_{+}\left(  t\right)  $ by means of the
ODE problem:%
\begin{align}
\frac{dr_{+}\left(  \tau\right)  }{d\tau}  &  =\frac{\exp\left(  \hat
{U}\left(  r_{+}\left(  \tau\right)  e^{\tau}\right)  -\Lambda\left(
r_{+}\left(  \tau\right)  e^{\tau}\right)  \right)  V_{2}\left(  r_{+}\left(
\tau\right)  e^{\tau}\right)  r_{+}\left(  \tau\right)  }{\sqrt{\left(
V_{2}\left(  r_{+}\left(  \tau\right)  e^{\tau}\right)  \right)  ^{2}%
+\frac{e^{2\tau}}{\left(  r_{+}\left(  \tau\right)  \right)  ^{2}}}%
}\label{F2E7}\\
\lim_{\tau\rightarrow\infty}r_{+}\left(  \tau\right)   &  =R_{\max
}\ \label{F2E7d}%
\end{align}
where the function $V_{2}\left(  y\right)  $ is as in Theorem \ref{RV}. We
then define $R_{+}\left(  t\right)  $ by means of $R_{+}\left(  t\right)
=r_{+}\left(  \tau\right)  .$

We need to define in a precise manner some of the functions needed for the
fixed point argument. As a first step we construct the function $r_{+}\left(
\tau\right)  $ which solves (\ref{F2E7}), (\ref{F2E7d}).

Due to (\ref{F2E4}) it follows that the right-hand side of (\ref{F2E7})
behaves like $Ce^{-2\tau}$ as $\tau\rightarrow\infty$ for some suitable
$C\in\mathbb{R}$ if $r_{+}\left(  \tau\right)  \rightarrow R_{\max}$ as
$\tau\rightarrow\infty.$ This will imply the existence of at least one
solution of (\ref{F2E7}), (\ref{F2E7a}). Similar estimates might be derived
for the derivatives of $r_{+}\left(  \tau\right)  .$ More precisely:

\begin{proposition}
\label{Rmas} For any $R_{\max}>0,$ there exists $t_{0}<0$ such that there
exists a unique solution of (\ref{F2E7}), (\ref{F2E7a}) defined for
$-\log\left(  t_{0}\right)  \leq\tau<0$. Moreover, we have:%
\begin{equation}
\left\vert r_{+}\left(  \tau\right)  -R_{\max}\right\vert +\left\vert
\frac{dr_{+}\left(  \tau\right)  }{d\tau}\right\vert \leq Ce^{-4\tau}
\label{F6E4a}%
\end{equation}
for any $\tau\geq-\log\left(  t_{0}\right)  $ where $C$ is a constant that
depends in general on $y_{0},$ but not on $\tau.$ We will write $R_{+}\left(
t\right)  =r_{+}\left(  \tau\right)  .$
\end{proposition}

\begin{proof}
The asymptotics (\ref{F2E1})-(\ref{F2E4}) as well as (\ref{F2E6})-(\ref{F3E2})
imply that (\ref{F2E7}) can be written in the form:%
\begin{equation}
\frac{dr_{+}\left(  \tau\right)  }{d\tau}=e^{-4\tau}F\left(  \tau,r_{+}\left(
\tau\right)  \right)  \ \label{F2E7c}%
\end{equation}
where $F$ is a bounded function as well as its first derivatives if
$r_{+}\left(  \tau\right)  \in\left[  \frac{R_{\max}}{2},R_{\max}\right]  $
and $\tau\geq0.$ We can then reformulate (\ref{F2E7}), (\ref{F2E7a}) as the
integral equation:%
\begin{equation}
r_{+}\left(  \tau\right)  =R_{\max}-\int_{\tau}^{\infty}e^{-4s}F\left(
s,r_{+}\left(  s\right)  \right)  ds \label{F2E7b}%
\end{equation}

A fixed point argument then shows that there exists a unique solution of
(\ref{F2E7b}) defined for $-\log\left(  t_{0}\right)  \leq\tau<0$ if
$\left\vert t_{0}\right\vert $ is sufficiently small. The estimates for the
derivatives are then obtained using (\ref{F2E7c}).
\end{proof}

\subsection{Evolution equations satisfied by $v_{k}\left(  t,r\right)
,\ B_{k}\left(  t,r\right)  ,\ $\protect\linebreak$k=1,2.$\label{vB}}

In order to obtain solutions of (\ref{S1E3}), (\ref{S1E4}), (\ref{S3E7}),
(\ref{S3E10}) with the form (\ref{F1E0}) in the region where $r>R_{+}\left(
t\right)  $ we need to derive the evolution equations satisfied by the
functions $v_{k},\ B_{k}$ for $k=1,2.$ Since in the region $r>R_{+}\left(
t\right)  $ we assume that $B_{2}\left(  t,r\right)  =0$ we need to obtain
only the evolution equations for $v_{1},\ B_{1}.$ To this end we impose that
$\zeta$ given in (\ref{F1E0}) solves (\ref{S3E7}) in the sense of
distributions. We derive formally in this subsection the system of
differential equations that must be satisfied by $B_{1}\left(  t,r\right)
,\ v_{1}\left(  t,r\right)  $ and we will check later that the resulting
measure $\zeta$ satisfies (\ref{S1E3})-(\ref{S1E6}), (\ref{S3E7}),
(\ref{S3E10}) in the sense of Definition \ref{zetaWeak}. Notice that, by
assumption:%
\begin{equation}
\zeta\left(  t,r,v\right)  =B_{1}\left(  t,r\right)  \delta\left(
v-v_{1}\left(  t,r\right)  \right)  \ ,\ \ r>R_{+}\left(  t\right)
\label{Y2E8}%
\end{equation}

Then, the following identities hold in the sense of distributions:%
\begin{align*}
\partial_{\alpha}\zeta &  =\left(  \partial_{\alpha}B_{1}\right)
\delta\left(  v-v_{1}\left(  t,r\right)  \right)  -B_{1}\partial_{\alpha}%
v_{1}\left(  t,r\right)  \delta^{\prime}\left(  v-v_{1}\left(  t,r\right)
\right)  \ \ ,\ \ \alpha=t,r\\
\partial_{v}\zeta &  =B_{1}\delta^{\prime}\left(  v-v_{1}\left(  t,r\right)
\right)
\end{align*}

We \ will use now the following distributional identity:%
\begin{align*}
A\left(  t,r,v\right)  \delta^{\prime}\left(  v-v_{1}\left(  t,r\right)
\right)    & =A\left(  t,r,v_{1}\left(  t,r\right)  \right)  \delta^{\prime
}\left(  v-v_{1}\left(  t,r\right)  \right)  \\
& -A_{v}\left(  t,r,v_{1}\left(  t,r\right)  \right)  \delta\left(
v-v_{1}\left(  t,r\right)  \right)
\end{align*}

Then, using also (\ref{S3E7}) and (\ref{F3E3}), (\ref{F3E5}):%
\begin{equation}
\partial_{\tau}v_{1}\left(  t,r\right)  +e^{\bar{\mu}-\lambda}\frac{v_{1}%
}{\tilde{E}}\partial_{r}v_{1}+\left(  \lambda_{\tau}v_{1}+e^{\bar{\mu}%
-\lambda}\mu_{r}\tilde{E}-e^{\bar{\mu}-\lambda}\frac{1}{r^{3}\tilde{E}%
}\right)  =0\ \ \text{if\ \ }B_{1}\neq0\ \label{T1E1}%
\end{equation}%
\begin{equation}
\partial_{\tau}B_{1}+e^{\bar{\mu}-\lambda}\frac{v_{1}}{\tilde{E}}\partial
_{r}B_{1}+\left(  e^{\bar{\mu}-\lambda}\frac{v_{1}}{\tilde{E}}\left(
\partial_{r}v_{1}\right)  +\lambda_{\tau}v+e^{\bar{\mu}-\lambda}\mu_{r}%
\tilde{E}-e^{\bar{\mu}-\lambda}\frac{1}{r^{3}\tilde{E}}\right)  _{v}\left(
t,r,v_{1}\right)  B_{1}=0 \label{T1E2}%
\end{equation}
with $\tilde{E}$ as in (\ref{S3E10}).

We need to complement the equations (\ref{T1E1}), (\ref{T1E2}) with the
equations that determine the functions $\lambda,\ \mu$ (cf. (\ref{S1E3}),
(\ref{S1E4})). Notice that the functions $\rho$ and $p$ have the form:%
\begin{equation}
\rho=\frac{\pi}{r^{2}}\tilde{E}B_{1}\ \ ,\ \ p=\frac{\pi}{r^{2}}\frac
{v_{1}^{2}}{\tilde{E}}B_{1}\ \ ,\ \ \ \ r>R_{+}\left(  t\right)  \label{T1E3}%
\end{equation}

Then (\ref{S1E3}), (\ref{S1E4}) become:%
\begin{equation}
e^{-2\lambda}\left(  2r\lambda_{r}-1\right)  +1=8\pi^{2}\tilde{E}%
B_{1},\ \ e^{-2\lambda}\left(  2r\mu_{r}+1\right)  -1=\frac{8\pi^{2}v_{1}%
^{2}B_{1}}{\tilde{E}}\ \label{T1E3a}%
\end{equation}

We need to complement the system (\ref{T1E1})-(\ref{T1E3a}) with suitable
boundary conditions. Imposing continuity for the velocities at $r=R_{+}\left(
t\right)  $ as well as for the fields $\lambda,\ \bar{\mu}$ and the densities
we obtain:%
\begin{align}
v_{1}\left(  \tau,r_{+}\left(  \tau\right)  \right)   &  =e^{\tau}V_{1}\left(
r_{+}\left(  \tau\right)  e^{\tau}\right)  \ \ ,\ \ B_{1}\left(  \tau
,r_{+}\left(  \tau\right)  \right)  =e^{-\tau}b_{1}\left(  r_{+}\left(
\tau\right)  e^{\tau}\right) \label{T1E4}\\
\lambda\left(  t,r_{+}\left(  \tau\right)  \right)   &  =\Lambda\left(
r_{+}\left(  \tau\right)  e^{\tau}\right)  \ \ \ ,\ \ \ \bar{\mu}\left(
t,r_{+}\left(  \tau\right)  \right)  =\log\left(  r_{+}\left(  \tau\right)
\right)  +\hat{U}\left(  r_{+}\left(  \tau\right)  e^{\tau}\right)
\label{T1E5}%
\end{align}

In order to require the weakest possible differentiability properties for the
fields $\lambda,\mu$, it is convenient to rewrite the equations (\ref{T1E1}),
(\ref{T1E2}). Adding and subtracting $e^{\bar{\mu}-\lambda}\frac{v_{1}}%
{\tilde{E}}\partial_{r}\lambda$ to the left-hand side of (\ref{T1E1}) we
obtain:%
\begin{equation}
\partial_{\tau}v_{1}+e^{\bar{\mu}-\lambda}\frac{v_{1}}{\tilde{E}}\partial
_{r}v_{1}+Z\left(  t,r,v_{1}\right)  v_{1}+e^{\bar{\mu}-\lambda}\left(
\mu_{r}\tilde{E}-\frac{v_{1}^{2}}{\tilde{E}}\lambda_{r}\right)  -e^{\bar{\mu
}-\lambda}\frac{1}{r^{3}\tilde{E}}=0\label{T1E6}%
\end{equation}
with $Z\left(  t,r,v_{1}\right)  =\left(  \lambda_{\tau}+e^{\bar{\mu}-\lambda
}\frac{v_{1}}{\tilde{E}}\partial_{r}\lambda\right)  .$ Notice now that the
term $\left(  \mu_{r}\tilde{E}-\frac{v_{1}^{2}}{\tilde{E}}\lambda_{r}\right)
$ can be rewritten, using (\ref{S1E3}), (\ref{S1E4}) and (\ref{T1E3}) in the
form:%
\[
\mu_{r}\tilde{E}-\frac{v_{1}^{2}}{\tilde{E}}\lambda_{r}=\frac{\left(
e^{2\lambda}-1\right)  }{2r\tilde{E}}\left(  \frac{1}{r^{2}}+2v_{1}%
^{2}\right)
\]
where we use the cancellation of the terms containing $\tilde{B}_{1}.$ Using
this identity in (\ref{T1E6}) we obtain:%
\begin{equation}
\partial_{\tau}v_{1}+e^{\bar{\mu}-\lambda}\frac{v_{1}}{\tilde{E}}\partial
_{r}v_{1}+Z\left(  t,r,v_{1}\right)  v_{1}+\frac{\left(  e^{2\lambda
}-1\right)  e^{\bar{\mu}-\lambda}}{2r\tilde{E}}\left(  \frac{1}{r^{2}}%
+2v_{1}^{2}\right)  -e^{\bar{\mu}-\lambda}\frac{1}{r^{3}\tilde{E}%
}=0\label{T1E6a}%
\end{equation}

On the other hand adding \ and subtracting $e^{\bar{\mu}-\lambda}\mu_{r}%
\frac{v_{1}}{\tilde{E}}$ in the left-hand side of (\ref{T1E2}) we obtain,
after some rearrangement of terms:%
\begin{equation}
\partial_{\tau}B_{1}+e^{\bar{\mu}-\lambda}\frac{v_{1}}{\tilde{E}}\partial
_{r}B_{1}+\left(  \frac{e^{\bar{\mu}-\lambda}}{\tilde{E}^{3}r^{2}}\left(
\partial_{r}v_{1}\right)  +Z\left(  t,r,v_{1}\right)  +\frac{e^{\bar{\mu
}-\lambda}\left(  \mu_{r}-\lambda_{r}\right)  v_{1}}{\tilde{E}}+\frac
{e^{\bar{\mu}-\lambda}v_{1}}{r^{3}\tilde{E}^{3}}\right)  B_{1}=0\ \label{T1E7}%
\end{equation}

A relevant property of (\ref{T1E6a}) is that the terms containing derivatives
of the fields $\lambda,\ \bar{\mu}$ appear in the form of the convective
derivative $Z\left(  t,r,v_{1}\right)  =\left(  \lambda_{\tau}+e^{\bar{\mu
}-\lambda}\frac{v_{1}}{\tilde{E}}\partial_{r}\lambda\right)  .$ Then, we can
remove this term from the equation by means of a change of variables, namely:%
\begin{equation}
v_{1}\left(  t,r\right)  =\exp\left(  -\lambda\right)  \bar{w}\left(
t,r\right)  \ \ ,\ \ \ B_{1}\left(  t,r\right)  =\exp\left(  -\lambda\right)
\bar{D}\left(  t,r\right)  \label{Y2E3}%
\end{equation}

We remark that the change of variables (\ref{Y2E3}) will play a role similar
to (\ref{M2}) in Definition \ref{zetaWeak}. Its goal is to eliminate terms
like $\lambda_{t}$ in the differential equations under consideration.

Then (\ref{T1E6a}), (\ref{T1E7}) become:%
\begin{equation}
\partial_{\tau}\bar{w}+e^{\bar{\mu}-2\lambda}\frac{\bar{w}}{\tilde{E}}%
\partial_{r}\bar{w}+\frac{e^{\bar{\mu}}}{r\tilde{E}}\left(  \frac{\left(
e^{2\lambda}-1\right)  }{2}\left(  \frac{1}{r^{2}}+2e^{-2\lambda}\bar{w}%
^{2}\right)  -\frac{1}{r^{2}}\right)  =0\ \label{Y2E4}%
\end{equation}%
\begin{equation}
\partial_{\tau}\bar{D}+e^{\bar{\mu}-2\lambda}\frac{\bar{w}}{\tilde{E}}%
\partial_{r}\bar{D}+\frac{e^{\bar{\mu}-2\lambda}}{\tilde{E}}\left(
\frac{\partial_{r}\bar{w}}{\tilde{E}^{2}r^{2}}-\frac{\lambda_{r}\bar{w}%
}{\tilde{E}^{2}r^{2}}+\left(  \mu_{r}-\lambda_{r}\right)  \bar{w}+\frac
{w}{r^{3}\tilde{E}^{2}}\right)  \bar{D}=0\label{Y2E5}%
\end{equation}
with:%
\begin{equation}
\tilde{E}=\sqrt{\bar{w}^{2}e^{-2\lambda}+\frac{1}{r^{2}}}\ \ ,\ \ \ \tau
=\log\left(  \frac{1}{\left(  -t\right)  }\right)  \label{Y2E6}%
\end{equation}

The system (\ref{Y2E4}), (\ref{Y2E5}) must be solved with the following
boundary conditions (cf. (\ref{T1E4}), (\ref{T1E5}), (\ref{Y2E3})):%
\begin{equation}
\bar{w}\left(  \tau,r_{+}\left(  \tau\right)  \right)  =e^{\tau-\Lambda\left(
r_{+}\left(  \tau\right)  e^{\tau}\right)  }V_{1}\left(  r_{+}\left(
\tau\right)  e^{\tau}\right)  \ \ ,\ \ \bar{D}\left(  \tau,r_{+}\left(
\tau\right)  \right)  =e^{-\tau-\Lambda\left(  r_{+}\left(  \tau\right)
e^{\tau}\right)  }b_{1}\left(  r_{+}\left(  \tau\right)  e^{\tau}\right)
\label{Y2E7}%
\end{equation}

It is relevant to remark that the equation for $\bar{D}$ contains a term involving a first derivative of $v_{1}$, or more precisely $\frac
{e^{\bar{\mu}-2\lambda}}{\tilde{E}^{3}r^{2}}\left(  \partial_{r}\bar
{w}\right)  .$ As a consequence we will need to consider function spaces which
estimate one derivative more for $\bar{D}$ than for $\bar{w}.$

\bigskip

A similar computation yields:%
\begin{equation}
\left(  \partial_{\tau}\bar{w}_{k}+e^{\bar{\mu}-2\lambda}\frac{\bar{w}_{k}%
}{\tilde{E}_{k}}\partial_{r}\bar{w}_{k}\right)  +\frac{e^{\bar{\mu}}}%
{r\tilde{E}_{k}}\left(  \frac{\left(  e^{2\lambda}-1\right)  }{2}\left(
\frac{1}{r^{2}}+2e^{-2\lambda}\bar{w}_{k}^{2}\right)  -\frac{1}{r^{2}}\right)
=0\label{J3E3}%
\end{equation}%
\begin{equation}
\left(  \partial_{\tau}\bar{D}_{k}+e^{\bar{\mu}-2\lambda}\frac{\bar{w}_{k}%
}{\tilde{E}_{k}}\partial_{r}\bar{D}_{k}\right)  +\frac{e^{\bar{\mu}-2\lambda}%
}{\tilde{E}_{k}}\left(  \frac{\partial_{r}\bar{w}_{k}}{\tilde{E}_{k}^{2}r^{2}%
}-\frac{\lambda_{r}\bar{w}_{k}}{\tilde{E}_{k}^{2}r^{2}}+\left(  \mu
_{r}-\lambda_{r}\right)  \bar{w}_{k}+\frac{w_{k}}{r^{3}\tilde{E}_{k}^{2}%
}\right)  \bar{D}_{k}=0\ \label{J3E4}%
\end{equation}
for $r<R_{+}\left(  t\right)  ,$ $k=1,2$ with:%
\begin{equation}
\tilde{E}_{k}=\sqrt{\bar{w}_{k}^{2}e^{-2\lambda}+\frac{1}{r^{2}}%
}\ \ ,\ \ \ \tau=\log\left(  \frac{1}{\left(  -t\right)  }\right)
\ \label{J3E5}%
\end{equation}%
\begin{equation}
\frac{e^{-2\lambda}\left(  2r\lambda_{r}-1\right)  +1}{8\pi^{2}}=\left(
\tilde{E}_{1}B_{1}+\tilde{E}_{2}B_{2}\right)  ,\ \ \frac{e^{-2\lambda}\left(
2r\mu_{r}+1\right)  -1}{8\pi^{2}}=\left(  \frac{v_{1}^{2}B_{1}}{\tilde{E}_{1}%
}+\frac{v_{2}^{2}B_{2}}{\tilde{E}_{2}}\right)  \ \label{J3E6}%
\end{equation}
for $\ r<R_{+}\left(  t\right)  $ and\
\begin{equation}
v_{k}\left(  t,r\right)  =\exp\left(  -\lambda\right)  \bar{w}_{k}\left(
t,r\right)  \ \ ,\ \ \ B_{k}\left(  t,r\right)  =\exp\left(  -\lambda\right)
\bar{D}_{k}\left(  t,r\right)  \label{J3E7}%
\end{equation}
if $k=1,2,\ \ \ \ r<R_{+}\left(  t\right)  .$

\subsection{Characteristic curves for (\ref{Y2E4})-(\ref{Y2E6}).}

Our next goal is to solve the equations (\ref{Y2E4})-(\ref{Y2E6}). We also
want to define a concept of solution of (\ref{Y2E4})-(\ref{Y2E7}) using the
weakest possible regularity. To this end, we will integrate these equations
using characteristics. We now formulate the characteristic equations. The
solvability of these equations will be proved later.

Let us denote as $\left(  r\left(  \tau;\bar{\tau}\right)  ,w\left(  \tau
;\bar{\tau}\right)  ,D\left(  \tau;\bar{\tau}\right)  \right)  $ the
characteristic curves associated to the system (\ref{Y2E4}), (\ref{Y2E5})
defined for $\tau\leq\bar{\tau}.$ We will assume that the curve $\left\{
r=r\left(  \tau;\bar{\tau}\right)  \right\}  $ reaches the boundary of the
domain $\mathcal{D}\left(  T\right)  =\left\{  \left(  \tau,r\right)
:r>r_{+}\left(  \tau\right)  ;\tau\geq T\right\}  $ for $\tau=\bar{\tau}.$
Suppose that the curves $\left\{  r=r\left(  \tau;\bar{\tau}\right)
:T\leq\tau\leq\bar{\tau}\right\}  $ cover the whole domain $\mathcal{D}\left(
T\right)  .$ Then the derivative $\partial_{r}w$ in (\ref{Y2E5}) evaluated at
$\left(  r\left(  \tau;\bar{\tau}\right)  ,\tau\right)  $ can be computed,
using the Implicit Function Theorem, by means of $\left(  \frac{\partial
w}{\partial\bar{\tau}}\left(  \tau,\bar{\tau}\right)  \right)  /\left(
\frac{\partial r}{\partial\bar{\tau}}\left(  \tau,\bar{\tau}\right)  \right)
.$ We can then write the characteristic equations associated to the equations
(\ref{Y2E4})-(\ref{Y2E6}) as:%
\begin{align}
\frac{\partial r\left(  \tau;\bar{\tau}\right)  }{\partial\tau}  &
=\frac{e^{\bar{\mu}\left(  r\left(  \tau;\bar{\tau}\right)  ,\tau\right)
-2\lambda\left(  r\left(  \tau;\bar{\tau}\right)  ,\tau\right)  }w\left(
\tau;\bar{\tau}\right)  r\left(  \tau;\bar{\tau}\right)  }{\Xi\left(
\tau;\bar{\tau}\right)  }\ ,\ \ r\left(  \bar{\tau};\bar{\tau}\right)
=r_{+}\left(  \bar{\tau}\right) \label{F5E1}\\
\frac{\partial w\left(  \tau;\bar{\tau}\right)  }{\partial\tau}  &
=-\frac{e^{\bar{\mu}\left(  r\left(  \tau;\bar{\tau}\right)  ,\tau\right)
}\left(  e^{2\lambda\left(  r\left(  \tau;\bar{\tau}\right)  \right)
}-1\right)  }{2\Xi\left(  \tau;\bar{\tau}\right)  }\left[  \frac{1}{\left(
r\left(  \tau;\bar{\tau}\right)  \right)  ^{2}}+2e^{-2\lambda\left(  r\left(
\tau;\bar{\tau}\right)  ,\tau\right)  }\left(  w\left(  \tau;\bar{\tau
}\right)  \right)  ^{2}\right] \nonumber\\
&  +\frac{e^{\bar{\mu}\left(  r\left(  \tau;\bar{\tau}\right)  ,\tau\right)
}}{\left(  r\left(  \tau;\bar{\tau}\right)  \right)  ^{2}\Xi\left(  \tau
;\bar{\tau}\right)  }\label{F5E2}\\
w\left(  \bar{\tau};\bar{\tau}\right)   &  =\exp\left(  \bar{\tau}%
-\Lambda\left(  r_{+}\left(  \bar{\tau}\right)  e^{\bar{\tau}}\right)
\right)  V_{1}\left(  r_{+}\left(  \bar{\tau}\right)  e^{\bar{\tau}}\right)
\label{F5E3}%
\end{align}%
\begin{align}
&  \frac{\partial D\left(  \tau;\bar{\tau}\right)  }{\partial\tau}%
\label{Y3E4}\\
&  =-\left(  \frac{e^{\bar{\mu}\left(  r\left(  \tau;\bar{\tau}\right)
,\tau\right)  -2\lambda\left(  r\left(  \tau;\bar{\tau}\right)  ,\tau\right)
}\cdot\left[  \frac{\left(  \frac{\partial w}{\partial\bar{\tau}}\left(
\tau,\bar{\tau}\right)  \right)  }{\left(  \frac{\partial r}{\partial\bar
{\tau}}\left(  \tau,\bar{\tau}\right)  \right)  }-\left(  \partial_{r}%
\lambda\right)  \left(  r\left(  \tau;\bar{\tau}\right)  ,\tau\right)
w\left(  \tau;\bar{\tau}\right)  \right]  r\left(  \tau;\bar{\tau}\right)
}{\left(  \Xi\left(  \tau;\bar{\tau}\right)  \right)  ^{3}}\right. \nonumber\\
&  \left.  +\frac{e^{\bar{\mu}\left(  r\left(  \tau;\bar{\tau}\right)
,\tau\right)  -2\lambda\left(  r\left(  \tau;\bar{\tau}\right)  ,\tau\right)
}\left[  \bar{\mu}_{r}\left(  r\left(  \tau;\bar{\tau}\right)  ,\tau\right)
-\lambda_{r}\left(  r\left(  \tau;\bar{\tau}\right)  ,\tau\right)  \right]
w\left(  \tau;\bar{\tau}\right)  r\left(  \tau;\bar{\tau}\right)  }{\Xi\left(
\tau;\bar{\tau}\right)  }\right.  \ \nonumber\\
&  \left.  +\frac{e^{\bar{\mu}\left(  r\left(  \tau;\bar{\tau}\right)
,\tau\right)  -2\lambda\left(  r\left(  \tau;\bar{\tau}\right)  ,\tau\right)
}w\left(  \tau;\bar{\tau}\right)  }{\left(  \Xi\left(  \tau;\bar{\tau}\right)
\right)  ^{3}}\right)  D\left(  \tau;\bar{\tau}\right) \nonumber\\
D\left(  \bar{\tau};\bar{\tau}\right)   &  =\exp\left(  -\bar{\tau}%
-\Lambda\left(  r_{+}\left(  \bar{\tau}\right)  e^{\bar{\tau}}\right)
\right)  b_{1}\left(  r_{+}\left(  \bar{\tau}\right)  e^{\bar{\tau}}\right)
\ \nonumber
\end{align}
where the boundary values have been chosen using (\ref{Y2E7}) and where:%
\begin{equation}
\Xi\left(  \tau;\bar{\tau}\right)  =\sqrt{1+e^{-2\lambda\left(  r\left(
\tau;\bar{\tau}\right)  ,\tau\right)  }\left(  w\left(  \tau;\bar{\tau
}\right)  \right)  ^{2}\left(  r\left(  \tau;\bar{\tau}\right)  \right)  ^{2}}
\label{Yrho}%
\end{equation}

We are interested in obtaining solutions of (\ref{F5E1})-(\ref{Y3E4}) in the
domain:%
\begin{equation}
\mathcal{U}\left(  T\right)  =\left\{  \left(  \tau,\bar{\tau}\right)
\in\left(  T,\infty\right)  \times\left(  T,\infty\right)  :\bar{\tau}\geq
\tau\right\}  \label{Y3E5}%
\end{equation}
where from now on $T=\log\left(  \frac{1}{\left(  -t_{0}\right)  }\right)  .$

In order to solve the system (\ref{F5E1})-(\ref{Y3E4}) we need to define
suitable function spaces. As a preliminary step we study the asymptotic
behavior or the solutions of a system of equations which will describe the
asymptotics of the solutions of (\ref{F5E1})-(\ref{F5E3}) for large values of
$\tau$ and $\bar{\tau}.$

\subsection{Formal asymptotic behaviour of the characteristic curves
(\ref{F5E1})-(\ref{F5E3}). \label{auxF}}

In this Section we derive by means of formal computations the asymptotics of
the solutions of (\ref{T1E1}), (\ref{T1E2}), (\ref{T1E3a}), (\ref{T1E4}),
(\ref{T1E5}) that we construct in this paper. These computations, although
formal, will be useful to get some intuitive understanding of the form of the
solutions and also to give some justification of the function spaces that we
will use to construct the solutions. We also study some auxiliary functions
which will be used in the following to define suitable function spaces.

Notice that we can expect $\bar{D}\left(  t,r\right)  $ to approach to zero as
\thinspace$t\rightarrow0,$ due to (\ref{Y2E7}). Then, using (\ref{Y2E3}) as
well as the fact that the fields $\lambda,\ \mu$ are given as in Lemma
\ref{defFields} and the asymptotics (\ref{F6E4a}), we would obtain the
following approximations for $\bar{\tau}\geq\tau\geq T$ if $\tau_{0}$ is large
enough:%
\[
r_{+}\left(  \bar{\tau}\right)  =R_{\max}\ \ ,\ \ R_{0}\left(  \tau,r\right)
=\frac{2R_{\max}}{3}%
\]%
\begin{align}
\lambda\left(  \tau,r\right)   &  =\lambda_{0}\left(  r\right)  =\frac{1}%
{2}\log\left(  r\right)  -\frac{1}{2}\log\left(  r-\frac{2R_{\max}}{3}\right)
\label{F4E5}\\
\bar{\mu}\left(  \tau,r\right)   &  =\bar{\mu}_{0}\left(  r\right)
=\log\left(  \frac{R_{\max}\sqrt{3\left(  1-y_{0}^{2}\right)  }}{y_{0}%
}\right)  +\frac{1}{2}\log\left(  1-\frac{2R_{\max}}{3r}\right)
\ \label{F4E6}%
\end{align}

With these approximations, (\ref{F5E1})-(\ref{F5E3}) become:%
\begin{align}
\frac{\partial r\left(  \tau;\bar{\tau}\right)  }{\partial\tau}  &
=\frac{e^{\bar{\mu}_{0}\left(  r\left(  \tau;\bar{\tau}\right)  \right)
-2\lambda_{0}\left(  r\left(  \tau;\bar{\tau}\right)  \right)  }w\left(
\tau;\bar{\tau}\right)  r\left(  \tau;\bar{\tau}\right)  }{\Xi\left(
\tau;\bar{\tau}\right)  }\ \ ,\ \ r\left(  \bar{\tau};\bar{\tau}\right)
=R_{\max}\label{F4E7}\\
\frac{\partial w\left(  \tau;\bar{\tau}\right)  }{\partial\tau}  &
=-\frac{e^{\bar{\mu}_{0}\left(  r\left(  \tau;\bar{\tau}\right)  \right)
}\left(  e^{2\lambda_{0}\left(  r\left(  \tau;\bar{\tau}\right)  \right)
}-1\right)  }{2\Xi\left(  \tau;\bar{\tau}\right)  }\left[  \frac{1}{\left(
r\left(  \tau;\bar{\tau}\right)  \right)  ^{2}}+2e^{-2\lambda_{0}\left(
r\left(  \tau;\bar{\tau}\right)  ,\tau\right)  }\left(  w\left(  \tau
;\bar{\tau}\right)  \right)  ^{2}\right] \nonumber\\
&  +\frac{e^{\bar{\mu}_{0}\left(  r\left(  \tau;\bar{\tau}\right)  \right)  }%
}{\left(  r\left(  \tau;\bar{\tau}\right)  \right)  ^{2}\Xi\left(  \tau
;\bar{\tau}\right)  }\label{F4E9}\\
w\left(  \bar{\tau};\bar{\tau}\right)   &  =-\frac{6y_{0}\sqrt{\left(
1-y_{0}^{2}\right)  }}{\left(  1-4y_{0}^{4}\right)  R_{\max}}\ \nonumber
\end{align}
with $\Xi\left(  \tau;\bar{\tau}\right)  $ as in (\ref{Yrho}). The solution of
(\ref{F4E7}), (\ref{F4E9}) can be obtained in the form:%

\[
r\left(  \tau;\bar{\tau}\right)  =\mathcal{R}\left(  \tau-\bar{\tau}\right)
\ \ ,\ \ w\left(  \tau;\bar{\tau}\right)  =\mathcal{W}\left(  \tau-\bar{\tau
}\right)  \ \
\]
where $\mathcal{R}$ and $\mathcal{W}$ solve:%
\begin{align}
\frac{\partial\mathcal{R}\left(  \tau\right)  }{\partial\tau}  &
=\frac{e^{\bar{\mu}_{0}\left(  \mathcal{R}\left(  \tau\right)  \right)
-2\lambda_{0}\left(  \mathcal{R}\left(  \tau\right)  \right)  }\mathcal{W}%
\left(  \tau\right)  \mathcal{R}\left(  \tau\right)  }{\sqrt{1+e^{-2\lambda
_{0}\left(  \mathcal{R}\left(  \tau\right)  \right)  }\left(  \mathcal{R}%
\left(  \tau\right)  \right)  ^{2}\left(  \mathcal{W}\left(  \tau\right)
\right)  ^{2}}}\ ,\ \mathcal{R}\left(  0\right)  =R_{\max}\label{F5E5}\\
\frac{\partial\mathcal{W}\left(  \tau\right)  }{\partial\tau}  &
=-\frac{e^{\bar{\mu}_{0}\left(  \mathcal{R}\left(  \tau\right)  \right)
}\left(  e^{2\lambda_{0}\left(  \mathcal{R}\left(  \tau\right)  \right)
}-1\right)  }{2\sqrt{1+e^{-2\lambda_{0}\left(  \mathcal{R}\left(  \tau\right)
\right)  }\left(  \mathcal{R}\left(  \tau\right)  \right)  ^{2}\left(
\mathcal{W}\left(  \tau\right)  \right)  ^{2}}}\left[  \frac{1}{\left(
\mathcal{R}\left(  \tau\right)  \right)  ^{2}}+2e^{-2\lambda_{0}\left(
\mathcal{R}\left(  \tau\right)  \right)  }\left(  \mathcal{W}\left(
\tau\right)  \right)  ^{2}\right] \nonumber\\
&  +\frac{e^{\bar{\mu}_{0}\left(  \mathcal{R}\left(  \tau\right)  \right)  }%
}{\left(  \mathcal{R}\left(  \tau\right)  \right)  ^{2}\sqrt{1+e^{-2\lambda
_{0}\left(  \mathcal{R}\left(  \tau\right)  \right)  }\left(  \mathcal{R}%
\left(  \tau\right)  \right)  ^{2}\left(  \mathcal{W}\left(  \tau\right)
\right)  ^{2}}}\label{F5E6}\\
\mathcal{W}\left(  0\right)   &  =-\frac{6y_{0}\sqrt{\left(  1-y_{0}%
^{2}\right)  }}{\left(  1-4y_{0}^{2}\right)  R_{\max}} \label{F5E7}%
\end{align}

The functions $\mathcal{R}$ and $\mathcal{W}$ will play an important role in
the following in order to describe the function spaces used in the solution of
(\ref{F5E1})-(\ref{Y3E4}). In the next proposition we describe their
asymptotic properties.

\begin{proposition}
\label{appFunctions}Suppose that $\left(  \mathcal{R},\mathcal{W}\right)  $
solve (\ref{F5E5})-(\ref{F5E7}) for $\tau\leq0$. Then:
\begin{equation}
\mathcal{W}\left(  \tau\right)  \rightarrow-\sqrt{\frac{4y_{0}^{2}\left(
1-y_{0}^{2}\right)  }{\left(  1-4y_{0}^{2}\right)  ^{2}R_{\max}^{2}}+\frac
{1}{R_{\max}}}\ \ \ \text{as\ \ }\tau\rightarrow-\infty\label{F6E3}%
\end{equation}%
\begin{equation}
\mathcal{R}\left(  \tau\right)  \sim-\frac{R_{\max}\sqrt{3\left(  1-y_{0}%
^{2}\right)  }}{y_{0}}\tau\ \ \text{as\ \ }\tau\rightarrow-\infty\label{F6E3b}%
\end{equation}

We have also:%
\begin{equation}
\left\vert \mathcal{R}^{\prime}\left(  \tau\right)  \right\vert \leq C\left(
y_{0},R_{\max}\right)  \ \ \text{for\ \ }\tau\leq0 \label{Y6E3c}%
\end{equation}

\begin{equation}
\left\vert \mathcal{W}\left(  \tau\right)  \right\vert \leq C\left(
y_{0},R_{\max}\right)  \ \ \text{for\ \ }\tau\leq0 \label{Y6E3d}%
\end{equation}
for some constant $C\left(  y_{0},R_{\max}\right)  $ depending only on
$y_{0},\ R_{\max}.$ Moreover, there exists $\Gamma\left(  y_{0,}R_{\max
}\right)  >0,$ depending only on $y_{0},R_{\max}$ such that:%
\begin{equation}
\mathcal{W}\left(  \tau\right)  \leq-\Gamma\left(  y_{0,}R_{\max}\right)
\ \ \text{for\ \ }\tau\leq0\ \label{Y6E3e}%
\end{equation}

\end{proposition}

\begin{proof}
The system of equations (\ref{F5E5})-(\ref{F5E7}) can be solved explicitly.
Indeed, this system can be reformulated as a Hamiltonian system. To this end
we define a new variable by means of:%
\begin{equation}
\mathcal{Z}\left(  \tau\right)  =\mathcal{W}\left(  \tau\right)  \frac{\left(
\mathcal{R}\left(  \tau\right)  -\frac{2R_{\max}}{3}\right)  }{\mathcal{R}%
\left(  \tau\right)  } \label{F5E8}%
\end{equation}

Using (\ref{F5E5}), (\ref{F5E6}) and (\ref{F5E8})\ we then obtain, after some
computations:%
\begin{equation}
\frac{\partial\mathcal{Z}\left(  \tau\right)  }{\partial\tau}=\frac
{e^{\bar{\mu}_{0}\left(  \mathcal{R}\left(  \tau\right)  \right)  }}%
{\sqrt{1+e^{-2\lambda_{0}\left(  \mathcal{R}\left(  \tau\right)  \right)
}\left(  \mathcal{R}\left(  \tau\right)  \right)  ^{2}\left(  \mathcal{W}%
\left(  \tau\right)  \right)  ^{2}}}\left[  -\frac{R_{\max}}{\left(
\mathcal{R}\left(  \tau\right)  \right)  ^{2}}+\frac{1}{\mathcal{R}\left(
\tau\right)  }\right]  \label{F5E9}%
\end{equation}

We rewrite (\ref{F5E5}) as:%
\begin{equation}
\frac{\partial\mathcal{R}\left(  \tau\right)  }{\partial\tau}=\frac
{e^{\bar{\mu}_{0}\left(  \mathcal{R}\left(  \tau\right)  \right)  }%
\mathcal{Z}\left(  \tau\right)  \mathcal{R}\left(  \tau\right)  }%
{\sqrt{1+e^{-2\lambda_{0}\left(  \mathcal{R}\left(  \tau\right)  \right)
}\left(  \mathcal{R}\left(  \tau\right)  \right)  ^{2}\left(  \mathcal{W}%
\left(  \tau\right)  \right)  ^{2}}}\ \label{F5E5a}%
\end{equation}

Combining (\ref{F5E5a}) and (\ref{F5E9}) we can obtain a conserved quantity
along characteristics, namely:%
\[
\frac{\left(  \mathcal{Z}\left(  \tau\right)  \right)  ^{2}}{2}-\frac{R_{\max
}}{2\left(  \mathcal{R}\left(  \tau\right)  \right)  ^{2}}+\frac
{1}{\mathcal{R}\left(  \tau\right)  }%
\]

Using the value of $\mathcal{R}\left(  0\right)  $ in (\ref{F5E5}),
$\mathcal{W}\left(  0\right)  $ in (\ref{F5E7}) and (\ref{F5E8}) we obtain:%
\begin{equation}
\frac{\left(  \mathcal{Z}\left(  \tau\right)  \right)  ^{2}}{2}-\frac{R_{\max
}}{2\left(  \mathcal{R}\left(  \tau\right)  \right)  ^{2}}+\frac
{1}{\mathcal{R}\left(  \tau\right)  }=\frac{2y_{0}^{2}\left(  1-y_{0}%
^{2}\right)  }{\left(  1-4y_{0}^{2}\right)  ^{2}R_{\max}^{2}}+\frac
{1}{2R_{\max}} \label{F6E1}%
\end{equation}

Notice that, since $\mathcal{Z}\left(  0\right)  =-\frac{2y_{0}\sqrt{\left(
1-y_{0}^{2}\right)  }}{\left(  1-4y_{0}^{2}\right)  R_{\max}}<0$ it follows
from (\ref{F5E9}) that $\mathcal{Z}\left(  \tau\right)  <0$ for $\tau
<\bar{\tau}.$ On the other hand, (\ref{F5E5a}) implies that $\mathcal{R}%
\left(  \tau\right)  $ is decreasing and then $\mathcal{R}\left(  \tau\right)
>R_{\max}$ for $\tau<\bar{\tau}.$ Moreover, (\ref{F5E5a}) implies also that
$\frac{\partial\mathcal{R}\left(  \tau\right)  }{\partial\tau}$ remains of
order one for $\tau<\bar{\tau}$ and then $\mathcal{R}\left(  \tau\right)
\rightarrow\infty$ as $\left(  \bar{\tau}-\tau\right)  \rightarrow\infty.$
Therefore:%
\begin{equation}
\mathcal{Z}\left(  \tau\right)  \rightarrow-\sqrt{\frac{4y_{0}^{2}\left(
1-y_{0}^{2}\right)  }{\left(  1-4y_{0}^{2}\right)  ^{2}R_{\max}^{2}}+\frac
{1}{R_{\max}}}\ \ \text{as\ \ }\tau\rightarrow-\infty\label{F6E2}%
\end{equation}

Then (\ref{F5E8}) implies (\ref{F6E3}). It then follows from (\ref{F5E5a}),
using also that $\lambda_{0}\left(  \mathcal{R}\left(  \tau\right)  \right)
\rightarrow0$ as $\tau\rightarrow-\infty$ and $\bar{\mu}_{0}\left(
\mathcal{R}\left(  \tau\right)  \right)  \rightarrow\log\left(  \frac{R_{\max
}\sqrt{3\left(  1-y_{0}^{2}\right)  }}{y_{0}}\right)  $ as $\tau
\rightarrow-\infty$ that:%
\begin{equation}
\frac{\partial\mathcal{R}\left(  \tau\right)  }{\partial\tau}\rightarrow
-\frac{R_{\max}\sqrt{3\left(  1-y_{0}^{2}\right)  }}{y_{0}}\text{\ \ as\ \ }%
\tau\rightarrow-\infty\ \label{F6E2a}%
\end{equation}
whence (\ref{F6E3b}) follows. Moreover, since $\frac{\partial\mathcal{R}%
\left(  \tau\right)  }{\partial\tau}$ is bounded in bounded regions, this
implies also (\ref{Y6E3c}). Using also that $\mathcal{W}$ is bounded as well
as (\ref{F6E3}) we obtain (\ref{Y6E3d}). To prove (\ref{Y6E3e}) we use the
fact that since $\mathcal{Z}\left(  0\right)  =-\frac{2y_{0}\sqrt{\left(
1-y_{0}^{2}\right)  }}{\left(  1-4y_{0}^{2}\right)  R_{\max}}<0$, and
(\ref{F6E2}) as well as the fact that $\mathcal{Z}\left(  \tau\right)  <0$ for
$\tau<\bar{\tau}$ imply that $\mathcal{Z}\left(  \tau\right)  \leq
-\Gamma_{\ast}\left(  y_{0,}R_{\max}\right)  <0$ for $\tau<\bar{\tau}.$ Using
then (\ref{F5E8}) as well as the fact that $\frac{\left(  \mathcal{R}\left(
\tau\right)  -\frac{2R_{\max}}{3}\right)  }{\mathcal{R}\left(  \tau\right)
}<1$ for $\mathcal{R}\left(  \tau\right)  \geq R_{\max},$ we obtain
(\ref{Y6E3e}). This concludes the proof.
\end{proof}

\subsection{Function spaces.}

We define a function space $\mathcal{X}_{L,T}$ as follows. Suppose that%
\begin{equation}
r,w,D\in W_{\mathrm{loc}}^{1,\infty}\left(  \mathcal{U}\left(  T\right)
\right)  \ \label{A1E0}%
\end{equation}
where $\mathcal{U}\left(  T\right)  $ is as in (\ref{Y3E5}). Let us assume
also that the functions $r,w,D$ satisfy:%
\begin{equation}
0\leq D\left(  \tau;\bar{\tau}\right)  \exp\left(  2\bar{\tau}\right)  \leq
L\ \ ,\ \ a.e.\ \left(  \tau,\bar{\tau}\right)  \in\mathcal{U}\left(
T\right)  \label{A1E1}%
\end{equation}
\
\begin{equation}
\left[  \frac{\left\vert \rho^{\ast}\left(  \tau;\bar{\tau}\right)
\right\vert }{1+\left(  \bar{\tau}-\tau\right)  }+\left\vert z\left(
\tau;\bar{\tau}\right)  \right\vert \right]  \leq\frac{1}{L}%
\ \ ,\ a.e.\ \left(  \tau,\bar{\tau}\right)  \in\mathcal{U}\left(  T\right)
\label{A1E2}%
\end{equation}
where $\rho^{\ast},\ z$ are defined by means of:%
\begin{equation}
\rho^{\ast}\left(  \tau;\bar{\tau}\right)  =\left(  r\left(  \tau;\bar{\tau
}\right)  -\mathcal{R}\left(  \tau-\bar{\tau}\right)  \right)
\ \ ,\ \ z\left(  \tau;\bar{\tau}\right)  =\left(  w\left(  \tau;\bar{\tau
}\right)  -\mathcal{W}\left(  \tau-\bar{\tau}\right)  \right)  \ \ ,\ \ \bar
{\tau}\geq\tau\geq T\ \label{F1E1}%
\end{equation}

On the other hand, since $\ r,w\in W_{\mathrm{loc}}^{1,\infty}\left(
\mathcal{U}\left(  T\right)  \right)  $ they are differentiable $a.e.\ $in
$\mathcal{U}\left(  T\right)  .$ We will assume also that:%
\begin{align}
\max\left\{  \left\vert \frac{\partial\rho^{\ast}\left(  \tau;\bar{\tau
}\right)  }{\partial\bar{\tau}}\right\vert ,\left\vert \frac{\partial z\left(
\tau;\bar{\tau}\right)  }{\partial\bar{\tau}}\right\vert \right\}   &
\leq\frac{1}{L}\ \ ,\ \ a.e.\ \text{\ in \ }\mathcal{U}\left(  T\right)
\label{A1E3}\\
\frac{\partial r}{\partial\tau}\left(  \tau;\bar{\tau}\right)   &  \leq
-\frac{1}{L},\ \ a.e.\ \text{\ in \ }\mathcal{U}\left(  T\right)
\ \ \label{A1E3a}%
\end{align}

We will assume also the following estimates for the derivatives $\frac
{\partial r}{\partial\tau},\ \frac{\partial w}{\partial\tau}:$%
\begin{equation}
\max\left\{  \left\vert \frac{\partial r}{\partial\tau}\right\vert ,\left\vert
\frac{\partial w}{\partial\tau}\right\vert \right\}  \leq\sqrt{L}%
\ \ ,\ \ a.e.\ \text{\ in \ }\mathcal{U}\left(  T\right)  \label{A1E5}%
\end{equation}

We will assume also that for the functions in the space $\mathcal{X}_{L,T}$ we
have:%
\begin{equation}
r\left(  \bar{\tau};\bar{\tau}\right)  =r_{+}\left(  \bar{\tau}\right)
\ \ ,\ \ w\left(  \bar{\tau};\bar{\tau}\right)  =\exp\left(  \bar{\tau
}-\Lambda\left(  r_{+}\left(  \bar{\tau}\right)  e^{\bar{\tau}}\right)
\right)  V_{1}\left(  r_{+}\left(  \bar{\tau}\right)  e^{\bar{\tau}}\right)
\ \ ,\ \ T\leq\bar{\tau}<\infty\label{Y3E6}%
\end{equation}
It is relevant to remark that the right-hand sides of these equations are
smooth functions. Therefore, the Lipschitz property implied by (\ref{A1E3}),
(\ref{A1E5}) is satisfied on the line $\left\{  \left(  \tau,\bar{\tau
}\right)  \in\mathcal{U}\left(  T\right)  :\tau=\bar{\tau}\right\}  .$

\begin{definition}
We will denote the space of functions satisfying (\ref{A1E0}), (\ref{A1E2}),
(\ref{A1E3}), (\ref{A1E3a}), (\ref{A1E5}) and (\ref{Y3E6}) by $\mathcal{X}%
_{L,T}.$
\end{definition}

\begin{remark}
\label{gam2}We will use in the following that, since $y_{0}$ can be assumed to
be small, the exponent $\gamma\left(  y_{0}\right)  =\left(  1-\frac{4\left(
1-y_{0}^{2}\right)  }{\left(  1-4y_{0}^{2}\right)  }\right)  \ $is larger than
$2$ in absolute value.
\end{remark}

\bigskip

We remark also that (\ref{F1E1}), (\ref{A1E3}), (\ref{A1E5}) imply that the
functions $r,\ w$ are Lipschitz continuous. Then they are differentiable
$a.e.$

Given any $\left(  r,w,D\right)  \in\mathcal{X}_{L,T}$ we can construct some
auxiliary functions $\bar{w}\left(  \tau,r\right)  ,\ \bar{D}\left(
\tau,r\right)  $ defined for $r\geq r_{+}\left(  \tau\right)  ,\ T\leq
\tau<\infty.$ This is proved in the following Lemma.

\begin{lemma}
\label{auxFunctions}There exist $L_{0}>0$ and $T_{0}>0$ sufficiently large
such that, for any function $\left(  r,w,D\right)  \in\mathcal{X}_{L,T}$ with
$L>L_{0}$ and $T>T_{0}$ there exist functions $\bar{D}\in L^{\infty}$ and
$\bar{w}\in W^{1,\infty}$ defined for $r\geq r_{+}\left(  \tau\right)  ,$
$\tau\geq T$ such that:%
\begin{align}
\bar{D}\left(  \tau,r\left(  \tau;\bar{\tau}\right)  \right)   &  =D\left(
\tau;\bar{\tau}\right)  \ \ ,\ \ r=r\left(  \tau;\bar{\tau}\right)
\ \ \ ,\ \ \bar{\tau}\geq\tau\geq T\ \ ,\ \ \ r\geq r_{+}\left(  \tau\right)
\label{A3E0a}\\
\bar{w}\left(  \tau,r\left(  \tau;\bar{\tau}\right)  ,\tau\right)   &
=w\left(  \tau;\bar{\tau}\right)  \ \ ,\ \ r=r\left(  \tau;\bar{\tau}\right)
\ \ \ ,\ \ \bar{\tau}\geq\tau\geq T\ \ ,\ \ \ r\geq r_{+}\left(  \tau\right)
\label{A3E0b}%
\end{align}

There exist positive constants $C_{0},\ a$ depending only on $y_{0}$ but
independent of $L$ such that:%
\begin{equation}
0\leq\bar{D}\left(  \tau,r\right)  \leq C_{0}Le^{-2\tau}e^{-ar}%
\ \ \ ,\ \ \left\vert \bar{w}\left(  \tau,r\right)  \right\vert \leq
C_{0}\ \ ,\ \ r\geq r_{+}\left(  \tau\right)  \ \label{A3E1}%
\end{equation}

Moreover, we have also the inequality:%
\begin{equation}
r\left(  \tau;\bar{\tau}\right)  \leq r_{+}\left(  \bar{\tau}\right)
+2C_{0}\left(  \bar{\tau}-\tau\right)  \ \ \ ,\ \ \bar{\tau}\geq\tau\geq T
\label{B1E2a}%
\end{equation}

\end{lemma}

\begin{proof}
We will assume in all the following that $L\geq1.$ Given $r>r_{+}\left(
\tau\right)  $ we define $\bar{\tau}\left(  \tau,r\right)  $ by means of the
formula:%
\[
r=r\left(  \tau;\bar{\tau}\left(  \tau,r\right)  \right)  =\mathcal{R}\left(
\tau-\bar{\tau}\left(  \tau,r\right)  \right)  +\left[  r\left(  \tau
;\bar{\tau}\left(  \tau,r\right)  \right)  -\mathcal{R}\left(  \tau-\bar{\tau
}\left(  \tau,r\right)  \right)  \right]
\]

In order to prove that the function $\bar{\tau}\left(  \tau,r\right)  $ is
well defined, we remark that, due to the definition of the space
$\mathcal{X}_{L,T}$ we have, for $T\leq\tau\leq\bar{\tau}<\infty$ and $L$
sufficiently large, the inequality $\frac{\partial r\left(  \tau;\bar{\tau
}\right)  }{\partial\bar{\tau}}\geq\theta>0$ where the number $\theta$ is
independent of $L.$ Indeed, since $\mathcal{R}^{\prime}\left(  s\right)
\leq-\delta<0$ for $s\leq0\ $it follows that, choosing $L_{0}$ sufficiently
large we obtain:%
\begin{equation}
\frac{\partial r\left(  \tau;\bar{\tau}\right)  }{\partial\bar{\tau}}\geq
\frac{\delta}{2}>0\ \label{B1E3}%
\end{equation}
where we use again the properties of the functions $r\left(  \cdot
,\cdot\right)  \in\mathcal{X}_{L,T}.$

Therefore using that $r_{+}\left(  \bar{\tau}\right)  =r\left(  \bar{\tau
};\bar{\tau}\right)  $, assuming that $\tau\geq T$ is fixed, it follows that
the image of the mapping defined by means of $\bar{\tau}\rightarrow r\left(
\tau;\bar{\tau}\right)  $ covers the whole range of values $\left[
r_{+}\left(  \tau\right)  ,\infty\right)  .$ Therefore there exists a unique
value $\bar{\tau}=\bar{\tau}\left(  \tau,r\right)  $ such that:%
\begin{equation}
r=r\left(  \tau;\bar{\tau}\left(  \tau,r\right)  \right)  \label{B1E4}%
\end{equation}

We then define functions $\bar{D}$ and $\bar{w}$ by means of:%
\begin{equation}
\bar{D}\left(  \tau,r\right)  =D\left(  \tau;\bar{\tau}\left(  \tau,r\right)
\right)  \ \ \ ,\ \ \ \bar{w}\left(  \tau,\bar{\tau}\left(  \tau,r\right)
\right)  =w\left(  \tau;\bar{\tau}\right)  \label{B1E1}%
\end{equation}

It remains to check the regularity properties of the function $\bar{v}_{1}.$
To this end we need to prove regularity for $\bar{\tau}\left(  \tau,r\right)
.$ Notice that due to the definition of $\mathcal{X}_{L,T}$ we have estimates
of the form $0<C_{1}\leq\frac{\partial r\left(  \tau;\bar{\tau}\right)
}{\partial\bar{\tau}}\leq C_{2}$ assuming that $L_{0}$ is large and $T\leq
\tau.$ This implies that the function $\bar{\tau}\left(  \tau,r\right)  $ is
uniformly Lipschitz in the variable $r$ for $\tau\geq T.$ Suppose
that:%
\[
r\left(  \tau;\bar{\tau}_{1}\right)  =r_{1}\ \ ,\ \ r\left(  \tau;\bar{\tau
}_{2}\right)  =r_{2}%
\]

Suppose that $r_{2}\geq r_{1}.$ Then:%
\[
r_{2}-r_{1}=r\left(  \tau;\bar{\tau}_{2}\right)  -r\left(  \tau;\bar{\tau}%
_{1}\right)  \geq C_{1}\left(  \bar{\tau}_{2}-\bar{\tau}_{1}\right)
\]
whence $\left(  \bar{\tau}_{2}-\bar{\tau}_{1}\right)  \leq\frac{1}{C_{1}%
}\left(  r_{2}-r_{1}\right)  .$ Therefore, since the choice of the largest
$r_{k}$ is arbitrary:%
\[
\left\vert \bar{\tau}\left(  \tau,r_{2}\right)  -\bar{\tau}\left(  \tau
,r_{1}\right)  \right\vert \leq C\left\vert r_{2}-r_{1}\right\vert
\]

On the other hand, the function $v_{1}$ is Lipschitz by assumption. Using the
strict monotonicity of the function $\bar{\tau}\left(  \tau,\cdot\right)  $
and the Lipschitz property of both $v_{1}$ and $\bar{\tau}\left(  \cdot
,\tau\right)  $ we can prove, arguing as in the proof of the chain rule, that
$\bar{v}_{1}\left(  \tau,\cdot\right)  $ is Lipschitz uniformly for $\tau\geq
T.$

We also need to prove the Lipschitz property in the variable $\tau.$ Arguing
again as in the proof of the chain rule, we can see that the problem reduces
to prove that the function $\bar{\tau}\left(  \tau,r\right)  $ is Lipschitz
with respect to $\tau.$ To this end we begin with the formula:%
\[
r=r\left(  \tau_{1};\bar{\tau}\left(  \tau_{1},r\right)  \right)
\ \ ,\ \ r=r\left(  \tau_{2};\bar{\tau}\left(  \tau_{2},r\right)  \right)
\]

Then:%
\[
r\left(  \tau_{1};\bar{\tau}\left(  \tau_{1},r\right)  \right)  =r\left(
\tau_{2};\bar{\tau}\left(  \tau_{2},r\right)  \right)
\]%
\[
r\left(  \tau_{1};\bar{\tau}\left(  \tau_{1},r\right)  \right)  -r\left(
\tau_{1};\bar{\tau}\left(  \tau_{2},r\right)  \right)  =r\left(  \tau_{2}%
;\bar{\tau}\left(  \tau_{2},r\right)  \right)  -r\left(  \tau_{1};\bar{\tau
}\left(  \tau_{2},r\right)  \right)
\]

Using the Lipschitz estimate of $r\left(  \tau,\bar{\tau}\right)  $ in $\tau$
we can estimate the right-hand side. On the other hand, suppose that
$\bar{\tau}\left(  \tau_{1},r\right)  \geq\bar{\tau}\left(  \tau_{2},r\right)
.$ The estimates for $r$ imply:%
\[
C_{1}\left(  \bar{\tau}\left(  \tau_{1},r\right)  -\bar{\tau}\left(  \tau
_{2},r\right)  \right)  \leq C\left\vert \tau_{2}-\tau_{1}\right\vert
\]

Then:%
\begin{equation}
\left\vert \bar{\tau}\left(  \tau_{1},r\right)  -\bar{\tau}\left(  \tau
_{2},r\right)  \right\vert \leq C\left\vert \tau_{2}-\tau_{1}\right\vert
\label{B1E2}%
\end{equation}

Therefore $\bar{\tau}$ is globally Lipschitz and then $\bar{D},\bar{w}\in
W^{1,\infty}.$

In order to prove (\ref{A3E1}) we notice that (\ref{Y6E3c}) and (\ref{A1E3})
imply:%
\[
\left\vert r\left(  \tau;\bar{\tau}\right)  -r\left(  \bar{\tau};\bar{\tau
}\right)  \right\vert \leq2C\left(  y_{0}\right)  \left(  \bar{\tau}%
-\tau\right)  \ \ ,\ \ \bar{\tau}\geq\tau
\]
since $L\geq1$. Using (\ref{Y3E6}) as well as (\ref{B1E4}) we derive
(\ref{B1E2a}).

Due to (\ref{A1E1}), (\ref{B1E1}) we have $0\leq\bar{D}\left(  \tau,r\right)
\leq L\exp\left(  -2\bar{\tau}\left(  \tau,r\right)  \right)  .$ Then, using
the boundedness of $r_{+}\left(  \bar{\tau}\right)  $ as well as the fact that
$L\geq1$ we obtain $\bar{D}\left(  \tau,r\right)  \leq CLe^{-2\tau}e^{-ar}$
for some $a>0,$ depending on $y_{0}$ but independent of $L.$\ This gives the
first inequality in (\ref{A3E1}). The second inequality follows from
(\ref{Y6E3d}) and (\ref{A1E2}).
\end{proof}

Lemma \ref{auxFunctions} suggests introducing the following function spaces:

\begin{definition}
We define the space $\mathcal{Y}_{L,T,a}$ as the space of functions  \linebreak$\bar
{D}\in L^{\infty}\left(  \left\{  \left(  \tau,r\right)  :r\geq r_{+}\left(
\tau\right)  ,\tau\geq T\right\}  \right)  $ satisfying the inequality:%
\begin{equation}
0\leq\bar{D}\left(  \tau,r\right)  \leq L^{\frac{3}{2}}e^{-2\tau}%
e^{-ar}\ \label{A1E1a}%
\end{equation}
with $a>0$.
\end{definition}

\begin{remark}
In the remainder of this paper it will be assumed that, given $L,T$ the
constant $a$ used in the definition of the spaces $\mathcal{Y}_{L,T,a}$ is
chosen as indicated in Lemma \ref{auxFunctions} whenever these spaces are
referred to.
\end{remark}

\subsection{Solutions of (\ref{Y2E4})-(\ref{Y2E7}) in the sense of
characteristics.}

We now define a suitable concept of solution of (\ref{Y2E4})-(\ref{Y2E7}).

\begin{definition}
\label{delSolChar}Suppose that $L>L_{0},\ T>T_{0}$ with $L_{0},T_{0}$ as in
Lemma \ref{auxFunctions}. We will say that $\left(  \left(  r,w\right)
,D\right)  \in\mathcal{X}_{L,T}$ is a solution of (\ref{Y2E4})-(\ref{Y2E7}) in
the sense of characteristics if it satisfies (\ref{F5E1})-(\ref{Y3E4}) $a.e.$
$\left(  \tau,\bar{\tau}\right)  \in\mathcal{U}\left(  T\right)  .$
\end{definition}

Our next goal is to prove that a solution of (\ref{Y2E4})-(\ref{Y2E7}) in the
sense of characteristics allows to obtain a solution of (\ref{S1E3}%
)-(\ref{S1E6}), (\ref{S3E7}), (\ref{S3E10}) in the sense of Definition
\ref{zetaWeak}. As a first step we derive a formula relating $D\left(
\tau;\bar{\tau}\right)  $ and $\frac{\partial r\left(  \tau;\bar{\tau}\right)
}{\partial\bar{\tau}}.$ This formula provides some geometric interpretation
for the function $D$ which is related with the stretching of the
characteristic curves. It will be used repeatedly in the following.

\subsection{A representation formula for $D\left(  \tau;\bar{\tau}\right)  $
in terms of $\frac{\partial r\left(  \tau;\bar{\tau}\right)  }{\partial
\bar{\tau}}.$}

\subsubsection{Heuristics.\label{Heur}}

We first need to obtain a formula for $D\left(  \tau,\bar{\tau}\right)  $ (or
equivalently $D\left(  t,\bar{t}\right)  $ with some abuse of notation). The
equation of $D\left(  \tau,\bar{\tau}\right)  $ can be integrated explicitly
for given fields $\lambda,\ \mu$ if the corresponding functions $r\left(
\tau,\bar{\tau}\right)  ,\ w\left(  \tau,\bar{\tau}\right)  $ are known. We
first obtain formally a derivation of the desired formula. To this end, we
first write the equation for $\zeta$ in divergence form. Using (\ref{S3E7}) we
can compute $\partial_{t}\left(  e^{\lambda}\zeta\right)  .$ After some
simple, but tedious algebraic computations we obtain:\
\begin{align}
&  \partial_{t}\left(  e^{\lambda}\zeta\right)  +\partial_{r}\left(  e^{\mu
}\frac{v}{\tilde{E}}\zeta\right)  -\partial_{r}\left(  \frac{e^{\mu}}%
{\tilde{E}}\right)  v\zeta-\partial_{v}\left[  e^{\lambda}\left(  \lambda
_{t}v+e^{\mu-\lambda}\mu_{r}\tilde{E}-e^{\mu-\lambda}\frac{1}{r^{3}\tilde{E}%
}\right)  \zeta\right] \label{Q1E1}\\
&  +\zeta\left(  e^{\mu}\mu_{r}\partial_{v}\tilde{E}-\frac{e^{\mu}}{r^{3}%
}\partial_{v}\left(  \frac{1}{\tilde{E}}\right)  \right) \nonumber\\
&  =0\nonumber
\end{align}

Using the definition of $\tilde{E}$ in (\ref{S3E10}) we obtain:%
\begin{equation}
-\partial_{r}\left(  \frac{e^{\mu}}{\tilde{E}}\right)  +\left(  e^{\mu}\mu
_{r}\partial_{v}\tilde{E}-\frac{e^{\mu}}{r^{3}}\partial_{v}\left(  \frac
{1}{\tilde{E}}\right)  \right)  =0 \label{Q1E2}%
\end{equation}

Combining (\ref{Q1E1}), (\ref{Q1E2}) we obtain the following equation which
has a suitable divergence form structure:%
\begin{equation}
\partial_{t}\left(  e^{\lambda}\zeta\right)  +\partial_{r}\left(  e^{\mu}%
\frac{v}{\tilde{E}}\zeta\right)  -\partial_{v}\left[  e^{\lambda}\left(
\lambda_{t}v+e^{\mu-\lambda}\mu_{r}\tilde{E}-e^{\mu-\lambda}\frac{1}%
{r^{3}\tilde{E}}\right)  \zeta\right]  =0 \label{Q1E3}%
\end{equation}

Notice that this formula is valid for any pair of functions $\lambda,\mu$ even
if they are not related with the function $\bar{D}$ by means of (\ref{T1E3a}),
(\ref{Y2E3}).

We can apply (\ref{Q1E3}) now formally to a measure $\zeta\left(
t,r,v\right)  $ with the form (\ref{Y2E8}). Our goal now is to obtain a
constant quantity along characteristics. Using (\ref{Y2E8}) and the second
identity in (\ref{Y2E3}) we obtain:
\begin{equation}
\int_{r\left(  t,\bar{t}\right)  }^{r\left(  t,\bar{t}+\Delta\bar{t}\right)
}\int_{0}^{\infty}e^{\lambda}\zeta\left(  t,r,v\right)  drdv=\int_{r\left(
t,\bar{t}\right)  }^{r\left(  t,\bar{t}+\Delta\bar{t}\right)  }\bar{D}\left(
t,r\right)  dr\ \label{Q1E4}%
\end{equation}
where $\Delta\bar{t}$ is arbitrary. Using now (\ref{Q1E3}) and integrating in
the variable $v$ we arrive at:%
\[
\partial_{t}\left(  \int_{0}^{\infty}e^{\lambda}\zeta dv\right)  +\partial
_{r}\left(  \int_{0}^{\infty}e^{\mu}\frac{v}{\tilde{E}}\zeta dv\right)  =0
\]
which combined with (\ref{Y2E8}) yields:%
\[
\partial_{t}\left(  \bar{D}\left(  t,r\right)  \right)  +\partial_{r}\left(
e^{\mu-\lambda}\frac{v_{1}\left(  t,r\right)  }{\tilde{E}}\bar{D}\left(
t,r\right)  \right)  =0
\]

We can use this formula to differentiate $\int_{r\left(  t,\bar{t}\right)
}^{r\left(  t,\bar{t}+\Delta\bar{t}\right)  }\bar{D}\left(  t,r\right)  dr.$
Using that:%
\[
\frac{\partial r}{\partial t}=e^{\mu-\lambda}\frac{v_{1}\left(  t,r\right)
}{\tilde{E}}%
\]
we obtain:%
\begin{align*}
\partial_{t}\left(  \int_{r\left(  t,\bar{t}\right)  }^{r\left(  t,\bar
{t}+\Delta\bar{t}\right)  }\bar{D}\left(  t,r\right)  dr\right)   &  =\bar
{D}\left(  t,r\left(  t,\bar{t}+\Delta\bar{t}\right)  \right)  \frac{\partial
r\left(  t,\bar{t}+\Delta\bar{t}\right)  }{\partial t}-\bar{D}\left(
t,r\left(  t,\bar{t}\right)  \right)  \frac{\partial r\left(  t,\bar
{t}\right)  }{\partial t}\\
&  -\int_{r\left(  t,\bar{t}\right)  }^{r\left(  t,\bar{t}+\Delta\bar
{t}\right)  }\partial_{r}\left(  e^{\mu-\lambda}\frac{v_{1}\left(  t,r\right)
}{\tilde{E}}\bar{D}\left(  t,r\right)  \right)
\end{align*}
and, using the formula for $\frac{\partial r}{\partial t}$ it then follows
that:%
\[
\partial_{t}\left(  \int_{r\left(  t,\bar{t}\right)  }^{r\left(  t,\bar
{t}+\Delta\bar{t}\right)  }\bar{D}\left(  t,r\right)  dr\right)  =0
\]

This formula yields in an integral form the desired conservation law. In order
to derive a pointwise conserved quantity we use the change of variables
$r=r\left(  t,\bar{t}\right)  ,\ dr=\frac{\partial r\left(  t,\bar{t}\right)
}{\partial\bar{t}}d\bar{t}$. Then:%
\[
\partial_{t}\left(  \int_{\bar{t}}^{\bar{t}+\Delta\bar{t}}\bar{D}\left(
t,r\left(  t,\bar{t}\right)  \right)  \frac{\partial r\left(  t,\bar
{t}\right)  }{\partial\bar{t}}d\bar{t}\right)  =0
\]
and since $\Delta\bar{t}$ is arbitrary we deduce that:%
\begin{equation}
\bar{D}\left(  t,r\left(  t,\bar{t}\right)  \right)  \frac{\partial r\left(
t,\bar{t}\right)  }{\partial\bar{t}}=D\left(  t,\bar{t}\right)  \frac{\partial
r\left(  t,\bar{t}\right)  }{\partial\bar{t}}=D\left(  \bar{t},\bar{t}\right)
\frac{\partial r\left(  \bar{t},\bar{t}\right)  }{\partial\bar{t}}
\label{Q1E5}%
\end{equation}

This is the desired conservation law which we will derive now in a more
precise and rigorous way.

\subsubsection{Rigorous proof of the representation formula for $D\left(
\tau;\bar{\tau}\right)  $.}

Our goal is to prove (\ref{Q1E5}) rigorously. The precise result is the following.

\begin{lemma}
\label{geomConst}Suppose that $\lambda,\mu\in W^{1,\infty}\left(  \left\{
r\geq r_{+}\left(  \tau\right)  ,\ \tau\geq T\right\}  \right)  $ and that
$\left(  r,w,D\right)  \in\mathcal{X}_{L,T}$ solve (\ref{F5E1})-(\ref{Y3E4}).
Then, there exists $f\left(  \bar{\tau}\right)  $ such that:%
\begin{equation}
\frac{\partial r\left(  \tau;\bar{\tau}\right)  }{\partial\bar{\tau}}D\left(
\tau;\bar{\tau}\right)  =f\left(  \bar{\tau}\right)  \ \ ,\ \ a.e.\ \left(
\tau;\bar{\tau}\right)  \in\mathcal{U}\left(  T\right)  \ \label{L1E2a}%
\end{equation}

\end{lemma}

\begin{proof}
If the functions $\frac{\partial r\left(  \tau;\bar{\tau}\right)  }%
{\partial\bar{\tau}},\ D\left(  \tau;\bar{\tau}\right)  $ were smooth, the
result would follow from the cancellation of the derivative $\frac{d}{d\tau
}\left[  \frac{\partial r\left(  \tau;\bar{\tau}\right)  }{\partial\bar{\tau}%
}D\left(  \tau;\bar{\tau}\right)  \right]  .$ Given that our regularity
assumptions do not guarantee the existence of this derivative in a classical
sense, we need to compute this derivative in weak form. We will assume first
that all the required derivatives exist and describe later how to adapt the
argument to a weak formulation. We compute first $\frac{d}{d\tau}\left(
\frac{\partial r\left(  \tau;\bar{\tau}\right)  }{\partial\bar{\tau}}\right)
.$

Differentiating (\ref{F5E1}) and rearranging some terms we obtain:%
\begin{align*}
\frac{\partial}{\partial\tau}\left(  \frac{\partial r\left(  \tau;\bar{\tau
}\right)  }{\partial\bar{\tau}}\right)   &  =\frac{e^{\bar{\mu}\left(
\tau,r\left(  \tau;\bar{\tau}\right)  \right)  -2\lambda\left(  \tau,r\left(
\tau;\bar{\tau}\right)  \right)  }w\left(  \tau;\bar{\tau}\right)  r\left(
\tau;\bar{\tau}\right)  }{\Xi\left(  \tau;\bar{\tau}\right)  }\left[  \bar
{\mu}_{r}\left(  \tau,r\left(  \tau;\bar{\tau}\right)  \right)  -2\lambda
_{r}\left(  \tau,r\left(  \tau;\bar{\tau}\right)  \right)  \right]
\frac{\partial r\left(  \tau;\bar{\tau}\right)  }{\partial\bar{\tau}}\\
&  +\frac{e^{\bar{\mu}\left(  \tau,r\left(  \tau;\bar{\tau}\right)  \right)
-2\lambda\left(  \tau,r\left(  \tau;\bar{\tau}\right)  \right)  }r\left(
\tau;\bar{\tau}\right)  }{\left(  \Xi\left(  \tau;\bar{\tau}\right)  \right)
^{3}}\frac{\partial w\left(  \tau;\bar{\tau}\right)  }{\partial\bar{\tau}}\\
&  +\frac{e^{\bar{\mu}\left(  \tau,r\left(  \tau;\bar{\tau}\right)  \right)
-2\lambda\left(  \tau,r\left(  \tau;\bar{\tau}\right)  \right)  }w\left(
\tau;\bar{\tau}\right)  }{\left(  \Xi\left(  \tau;\bar{\tau}\right)  \right)
^{3}}\frac{\partial r\left(  \tau;\bar{\tau}\right)  }{\partial\bar{\tau}}\\
&  +\frac{e^{\bar{\mu}\left(  \tau,r\left(  \tau;\bar{\tau}\right)  \right)
-4\lambda\left(  \tau,r\left(  \tau;\bar{\tau}\right)  \right)  }\left(
w\left(  \tau;\bar{\tau}\right)  \right)  ^{3}\left(  r\left(  \tau;\bar{\tau
}\right)  \right)  ^{3}\lambda_{r}\left(  \tau,r\left(  \tau;\bar{\tau
}\right)  \right)  }{\left(  \Xi\left(  \tau;\bar{\tau}\right)  \right)  ^{3}%
}\frac{\partial r\left(  \tau;\bar{\tau}\right)  }{\partial\bar{\tau}}%
\end{align*}
with $\Xi\left(  \tau;\bar{\tau}\right)  $ as in (\ref{Yrho}).

Combining this equation with (\ref{Y3E4}) we obtain:%
\begin{align}
&  \frac{\partial}{\partial\tau}\left[  \frac{\partial r\left(  \tau;\bar
{\tau}\right)  }{\partial\bar{\tau}}D\left(  \tau;\bar{\tau}\right)  \right]
\ \nonumber\\
&  =-\frac{e^{\bar{\mu}\left(  \tau,r\left(  \tau;\bar{\tau}\right)  \right)
-2\lambda\left(  \tau,r\left(  \tau;\bar{\tau}\right)  \right)  }w\left(
\tau;\bar{\tau}\right)  r\left(  \tau;\bar{\tau}\right)  \lambda_{r}\left(
\tau,r\left(  \tau;\bar{\tau}\right)  \right)  }{\Xi\left(  \tau;\bar{\tau
}\right)  }\frac{\partial r\left(  \tau;\bar{\tau}\right)  }{\partial\bar
{\tau}}D\left(  \tau;\bar{\tau}\right) \nonumber\\
&  +\frac{e^{\bar{\mu}\left(  \tau,r\left(  \tau;\bar{\tau}\right)  \right)
-2\lambda\left(  \tau,r\left(  \tau;\bar{\tau}\right)  \right)  }w\left(
\tau;\bar{\tau}\right)  r\left(  \tau;\bar{\tau}\right)  \lambda_{r}\left(
\tau,r\left(  \tau;\bar{\tau}\right)  \right)  }{\left(  \Xi\left(  \tau
;\bar{\tau}\right)  \right)  ^{3}}\frac{\partial r\left(  \tau;\bar{\tau
}\right)  }{\partial\bar{\tau}}D\left(  \tau;\bar{\tau}\right) \nonumber\\
&  +\frac{e^{\bar{\mu}\left(  \tau,r\left(  \tau;\bar{\tau}\right)  \right)
-2\lambda\left(  \tau,r\left(  \tau;\bar{\tau}\right)  \right)  }r\left(
\tau;\bar{\tau}\right)  }{\left(  \Xi\left(  \tau;\bar{\tau}\right)  \right)
^{3}}\frac{\partial w\left(  \tau;\bar{\tau}\right)  }{\partial\bar{\tau}%
}D\left(  \tau;\bar{\tau}\right) \nonumber\\
&  -\frac{e^{\bar{\mu}\left(  \tau,r\left(  \tau;\bar{\tau}\right)  \right)
-2\lambda\left(  \tau,r\left(  \tau;\bar{\tau}\right)  \right)  }\cdot\left(
\partial_{r}w\right)  \left(  \tau,r\left(  \tau;\bar{\tau}\right)  \right)
r\left(  \tau;\bar{\tau}\right)  }{\left(  \Xi\left(  \tau;\bar{\tau}\right)
\right)  ^{3}}\frac{\partial r\left(  \tau;\bar{\tau}\right)  }{\partial
\bar{\tau}}D\left(  \tau;\bar{\tau}\right) \label{S1E9}\\
&  +\frac{e^{\bar{\mu}\left(  \tau,r\left(  \tau;\bar{\tau}\right)  \right)
-4\lambda\left(  \tau,r\left(  \tau;\bar{\tau}\right)  \right)  }\left(
w\left(  \tau;\bar{\tau}\right)  \right)  ^{3}\left(  r\left(  \tau;\bar{\tau
}\right)  \right)  ^{3}\lambda_{r}\left(  \tau,r\left(  \tau;\bar{\tau
}\right)  \right)  }{\left(  \Xi\left(  \tau;\bar{\tau}\right)  \right)  ^{3}%
}\frac{\partial r\left(  \tau;\bar{\tau}\right)  }{\partial\bar{\tau}}D\left(
\tau;\bar{\tau}\right) \nonumber
\end{align}
where:%
\[
\left(  \partial_{r}w\right)  \left(  \tau,r\left(  \tau;\bar{\tau}\right)
\right)  =\left(  \frac{\partial w}{\partial\bar{\tau}}\left(  \tau,\bar{\tau
}\right)  \right)  /\left(  \frac{\partial r}{\partial\bar{\tau}}\left(
\tau,\bar{\tau}\right)  \right)
\]

Combining the terms on the right-hand side of (\ref{S1E9}) we obtain
$\frac{\partial}{\partial\tau}\left[  \frac{\partial r\left(  \tau;\bar{\tau
}\right)  }{\partial\bar{\tau}}D\left(  \tau;\bar{\tau}\right)  \right]
=0,\ $whence (\ref{L1E2a}) follows.

In order to prove the result if the derivative $\frac{\partial}{\partial\tau
}\left[  \frac{\partial r\left(  \tau;\bar{\tau}\right)  }{\partial\bar{\tau}%
}D\left(  \tau;\bar{\tau}\right)  \right]  $ does not exist, we use the
following argument. The weak formulation for the equations satisfied by
$\frac{\partial r\left(  \tau;\bar{\tau}\right)  }{\partial\bar{\tau}}$ and
$D\left(  \tau;\bar{\tau}\right)  $ have the form:%
\begin{align}
\int\int\frac{\partial r\left(  \tau;\bar{\tau}\right)  }{\partial\bar{\tau}%
}\frac{\partial\varphi}{\partial\tau}\left(  \tau;\bar{\tau}\right)  d\tau
d\bar{\tau}  &  =\int\int A_{1}\left(  \tau;\bar{\tau}\right)  \varphi\left(
\tau;\bar{\tau}\right)  d\tau d\bar{\tau}\ \label{B1}\\
\int\int D\left(  \tau;\bar{\tau}\right)  \frac{\partial\varphi}{\partial\tau
}\left(  \tau;\bar{\tau}\right)  d\tau d\bar{\tau}  &  =\int\int A_{2}\left(
\tau;\bar{\tau}\right)  \varphi\left(  \tau;\bar{\tau}\right)  d\tau
d\bar{\tau} \label{B2}%
\end{align}
for any compactly supported test function $\varphi.$ Moreover, the following
identity holds:%
\begin{equation}
A_{1}\left(  \tau;\bar{\tau}\right)  D\left(  \tau;\bar{\tau}\right)
+A_{2}\left(  \tau;\bar{\tau}\right)  \frac{\partial r\left(  \tau;\bar{\tau
}\right)  }{\partial\bar{\tau}}=0\ \label{B3}%
\end{equation}
Suppose that $\zeta_{\varepsilon}\in C^{\infty}\left(  \mathbb{R}\right)  $ is
a mollifier, compactly supported, nonnegative, which satisfies $\zeta
_{\varepsilon}\left(  s\right)  =\zeta_{\varepsilon}\left(  -s\right)  $ and
converges to a Dirac mass in the sense of measures. We will assume also that
$s\zeta_{\varepsilon}\left(  s\right)  \leq0.$ We take the test function
$\varphi=\zeta_{\varepsilon}\ast\left[  D\cdot\psi\right]  $ in (\ref{B1}) and
$\varphi=\zeta_{\varepsilon}\ast\left[  \frac{\partial r}{\partial\bar{\tau}%
}\cdot\psi\right]  ,$ where $\ast$ denotes the convolution in the variable
$\tau$ and $\psi\in C^{\infty}$ is a compactly supported test function. We
then obtain, exchanging the roles of the variables $\tau$ and $\xi$ and using
the symmetry properties of $\zeta_{\varepsilon}$ that:%
\begin{align}
&  \int\int\int\frac{\partial r\left(  \tau;\bar{\tau}\right)  }{\partial
\bar{\tau}}D\left(  \xi;\bar{\tau}\right)  \zeta_{\varepsilon}^{\prime}\left(
\tau-\xi\right)  \left[  \psi\left(  \xi;\bar{\tau}\right)  -\psi\left(
\tau;\bar{\tau}\right)  \right]  d\tau d\bar{\tau}d\xi\label{B4}\\
&  =\int\int\int\zeta_{\varepsilon}^{\prime}\left(  \tau-\xi\right)
\psi\left(  \xi;\bar{\tau}\right)  \left[  A_{1}\left(  \tau;\bar{\tau
}\right)  D\left(  \xi;\bar{\tau}\right)  +A_{2}\left(  \tau;\bar{\tau
}\right)  \frac{\partial r\left(  \xi;\bar{\tau}\right)  }{\partial\bar{\tau}%
}\right]  d\tau d\bar{\tau}d\xi\nonumber
\end{align}

Taking the limit $\varepsilon\rightarrow0,$ using Lebesgue's dominated
convergence Theorem and using (\ref{B3}) we obtain that the right-hand side of
(\ref{B4}) converges to zero. On the other hand $\zeta_{\varepsilon}^{\prime
}\left(  \tau-\cdot\right)  \left[  \psi\left(  \cdot;\bar{\tau}\right)
-\psi\left(  \tau;\bar{\tau}\right)  \right]  $ converges to a Dirac mass
multiplied by $\frac{\partial\varphi}{\partial\tau}\left(  \bar{\tau}%
;\bar{\tau}\right)  $ at the point $\xi=\tau.$ Using again Lebesgue's Theorem
we obtain:%
\[
\int\int\frac{\partial r\left(  \tau;\bar{\tau}\right)  }{\partial\bar{\tau}%
}D\left(  \tau;\bar{\tau}\right)  \frac{\partial\varphi}{\partial\tau}\left(
\bar{\tau};\bar{\tau}\right)  d\tau d\bar{\tau}=0
\]

This is the weak formulation of the identity (\ref{L1E2a}) whence the result follows.
\end{proof}

\subsection{Relation between the solutions of (\ref{Y2E4})-(\ref{Y2E7}) in the
sense of characteristics and the weak solutions of the problem.}

We now prove that it is possible to obtain a solution of (\ref{S1E3}%
)-(\ref{S1E6}), (\ref{S3E7}), (\ref{S3E10}) in the sense of Definition
\ref{zetaWeak} with the form (\ref{F1E0}) by glueing a self-similar solution
with the form given in Theorem \ref{RV} with another solution with the form
(\ref{Y2E8}) for $r>R_{+}\left(  t\right)  $ and $v_{1},\ B_{1}$ as in
(\ref{Y2E3}) with $w,D$ satisfying (\ref{Y2E4}), (\ref{Y2E5}), (\ref{Y2E7}) in
the sense of characteristics. We will denote as $\Omega_{t_{0}}$ the domain
$\left\{  \left(  t,r\right)  :t>t_{0},\ r>R_{+}\left(  t\right)  \right\}  .$

\begin{proposition}
\label{ws}Suppose that the the functions $w\in W_{\mathrm{loc}}^{1,\infty
}\left(  \Omega_{t_{0}}\right)  ,\ D\in W_{\mathrm{loc}}^{1,\infty}\left(
\Omega_{t_{0}}\right)  $ solve the equations (\ref{Y2E4}), (\ref{Y2E5}) for
\thinspace$a.e.\ \left(  t,r\right)  \in\Omega_{t_{0}},$ with boundary
conditions (\ref{Y2E7}), where $R_{+}\left(  t\right)  $ is as in Proposition
\ref{Rmas}. Let us define $v_{1},\ B_{1}$ for $r>R_{+}\left(  t\right)  ,$
$t\geq t_{0}$ by means of (\ref{Y2E3}). Suppose that $\lambda,\mu\in
W^{1,\infty}\left(  \left\{  t>t_{0}\right\}  \right)  $ are defined as in
Lemma \ref{defFields}. Let us assume that we extend $v_{1},\ B_{1}%
,\ \lambda,\ \mu$ to $r\leq r_{+}\left(  \tau\right)  $ as in Theorem
\ref{RV}. We define also $v_{2}$ as in Theorem \ref{RV} for all $r>0.$ We
define $B_{2}$ as in Theorem \ref{RV} for $r$ $\leq R_{+}\left(  t\right)  $
and we extend $B_{2}$ as zero for $r>R_{+}\left(  t\right)  .$ Then, the
measure $\zeta$ defined as in (\ref{F1E0}), as well as the fields
$\lambda,\ \mu$ in (\ref{Y1E2}), (\ref{Y1E3}) solve (\ref{S1E3})-(\ref{S1E6}),
(\ref{S3E7}), (\ref{S3E10}) in the sense of Definition \ref{zetaWeak} for
$t>t_{0}$.
\end{proposition}

The proof of Proposition \ref{ws} will be the content of the rest of this
subsection. We first need to rewrite some of the terms appearing in
(\ref{S4E2}), (\ref{Y1E8}) and Definition \ref{zetaWeak}.

\begin{lemma}
\label{Delta1}Suppose that $\Delta$ is as in (\ref{S4E1a}) with $\lambda
,\ \mu$ satisfying (\ref{T1E3a}) and $r\left(  \tau;\bar{\tau}\right)
,w\left(  \tau;\bar{\tau}\right)  \in W_{\mathrm{loc}}^{1,\infty}\left(
\mathcal{U}\left(  T\right)  \right)  $ solve (\ref{F5E1}), (\ref{F5E2}).
Then, the following identity holds:%
\begin{equation}
\Delta\left(  t,r\left(  \tau;\bar{\tau}\right)  ,w\left(  \tau;\bar{\tau
}\right)  \right)  e^{-\tau}=\frac{d}{d\tau}\left(  \varphi\left(  t,r\left(
\tau;\bar{\tau}\right)  ,w\left(  \tau;\bar{\tau}\right)  \right)  \right)
\ \ a.e.\ \left(  \tau;\bar{\tau}\right)  \in\mathcal{U}\left(  T\right)
\label{J1E1}%
\end{equation}
where $\tau$ is as in (\ref{F3E5}).
\end{lemma}

\begin{proof}
We first compute the right-hand side of (\ref{J1E1}):%
\begin{align*}
&  \frac{d}{d\tau}\left(  \varphi\left(  t,r\left(  \tau;\bar{\tau}\right)
,w\left(  r\left(  \tau;\bar{\tau}\right)  \right)  \right)  \right) \\
&  =\varphi_{t}\left(  t,r\left(  \tau;\bar{\tau}\right)  ,w\left(  r\left(
\tau;\bar{\tau}\right)  \right)  \right)  e^{-\tau}+\varphi_{r}\left(
t,r\left(  \tau;\bar{\tau}\right)  ,w\left(  r\left(  \tau;\bar{\tau}\right)
\right)  \right)  \frac{\partial r\left(  \tau;\bar{\tau}\right)  }%
{\partial\tau}\\
&  +\varphi_{w}\left(  t,r\left(  \tau;\bar{\tau}\right)  ,w\left(  r\left(
\tau;\bar{\tau}\right)  \right)  \right)  \frac{\partial w\left(  \tau
;\bar{\tau}\right)  }{\partial\tau}%
\end{align*}

Using (\ref{F5E1}), (\ref{F5E2}) we obtain:%
\begin{align}
&  \frac{d}{d\tau}\left(  \varphi\left(  t,r\left(  \tau;\bar{\tau}\right)
,w\left(  r\left(  \tau;\bar{\tau}\right)  \right)  \right)  \right)
\label{J1E4}\\
&  =\varphi_{t}\left(  t,r\left(  \tau;\bar{\tau}\right)  ,w\left(  r\left(
\tau;\bar{\tau}\right)  \right)  \right)  e^{-\tau}\nonumber\\
&  +\varphi_{r}\left(  t,r\left(  \tau;\bar{\tau}\right)  ,w\left(  r\left(
\tau;\bar{\tau}\right)  \right)  \right)  \frac{e^{\bar{\mu}\left(
\tau,r\left(  \tau;\bar{\tau}\right)  \right)  -2\lambda\left(  \tau,r\left(
\tau;\bar{\tau}\right)  \right)  }w\left(  \tau;\bar{\tau}\right)  r\left(
\tau;\bar{\tau}\right)  }{\Xi\left(  \tau;\bar{\tau}\right)  }\nonumber\\
&  +\varphi_{w}\left(  t,r\left(  \tau;\bar{\tau}\right)  ,w\left(  r\left(
\tau;\bar{\tau}\right)  \right)  \right)  \left[  -\frac{e^{\bar{\mu}\left(
\tau,r\left(  \tau;\bar{\tau}\right)  \right)  }\left(  e^{2\lambda\left(
\tau,r\left(  \tau;\bar{\tau}\right)  \right)  }-1\right)  }{2\Xi\left(
\tau;\bar{\tau}\right)  }\left[  \frac{1}{\left(  r\left(  \tau;\bar{\tau
}\right)  \right)  ^{2}}+2e^{-2\lambda\left(  \tau,r\left(  \tau;\bar{\tau
}\right)  \right)  }\left(  w\left(  \tau;\bar{\tau}\right)  \right)
^{2}\right]  \right.  \nonumber\\
&  \left.  +\frac{e^{\bar{\mu}\left(  \tau,r\left(  \tau;\bar{\tau}\right)
\right)  }}{\left(  r\left(  \tau;\bar{\tau}\right)  \right)  ^{2}\Xi\left(
\tau;\bar{\tau}\right)  }\right]  \nonumber
\end{align}

On the other hand expanding the derivatives $\partial_{r}$ and $\partial
_{\bar{v}}$ in (\ref{S4E1a}) we obtain:%
\begin{equation}
\Delta\left(  t,r,\bar{v}\right)  =\partial_{t}\varphi\left(  t,r,\bar
{v}\right)  +e^{\mu-2\lambda}\frac{\bar{v}}{\tilde{E}}\varphi_{r}-\left(
-\frac{\lambda_{r}e^{\mu-2\lambda}\bar{v}^{2}}{\tilde{E}}+e^{\mu}\mu_{r}%
\tilde{E}-e^{\mu}\frac{1}{r^{3}\tilde{E}}\right)  \varphi_{\bar{v}%
}\ \label{J1E2}%
\end{equation}
where we have used that%
\[
\partial_{r}\left(  e^{\mu-2\lambda}\frac{\bar{v}}{\tilde{E}}\right)
-\partial_{\bar{v}}\left(  \left(  -\frac{\lambda_{r}e^{\mu-2\lambda}\bar
{v}^{2}}{\tilde{E}}+e^{\mu}\mu_{r}\tilde{E}-e^{\mu}\frac{1}{r^{3}\tilde{E}%
}\right)  \right)  =0
\]
something that follows from an explicit computation. Using (\ref{T1E3a}) we
can eliminate the derivatives $\lambda_{r},\ \mu_{r}$ from (\ref{J1E2}).
Moreover, using also (\ref{F3E3}) we arrive, after some computations, at:%
\begin{align}
&  e^{-\tau}\Delta\left(  t,r,\bar{v}\right) \label{J1E3}\\
&  =e^{-\tau}\partial_{t}\varphi\left(  t,r,\bar{v}\right)  +e^{\bar{\mu
}-2\lambda}\frac{\bar{v}}{\tilde{E}}\varphi_{r}-\left(  -\frac{\left(
1-e^{2\lambda}\right)  e^{\bar{\mu}-2\lambda}\bar{v}^{2}}{2r\tilde{E}}%
+e^{\bar{\mu}}\left(  e^{2\lambda}-1\right)  \frac{\tilde{E}}{2r}-e^{\bar{\mu
}}\frac{1}{r^{3}\tilde{E}}\right)  \varphi_{\bar{v}}\nonumber
\end{align}

The identity (\ref{J1E1}) then follows combining (\ref{J1E4}) and (\ref{J1E3}).
\end{proof}

We now prove the following:

\begin{lemma}
\label{LIntegral}Suppose that $\Delta,$ $\lambda,\ \mu$ are as in Lemma
\ref{Delta1}. Suppose also that $\left(  r,w,D\right)  \in\mathcal{X}_{L,T}$
and let us assume that $D,\ w$ and $\bar{D},\ \bar{w}$ are related as in
(\ref{A3E0a}), (\ref{A3E0b}). Let $R_{+},\ r_{+}$ as in Proposition
\ref{Rmas}. Then:%
\begin{align}
&  \int_{t_{0}}^{0}\int_{R_{+}\left(  t\right)  }^{\infty}D\left(  t,r\right)
\Delta\left(  t,r,\bar{w}\left(  t,r\right)  \right)  drdt\nonumber\\
&  =\int_{T}^{\infty}\frac{\partial r\left(  \bar{\tau};\bar{\tau}\right)
}{\partial\bar{\tau}}D\left(  \bar{\tau};\bar{\tau}\right)  \varphi\left(
\bar{t},r\left(  \bar{\tau};\bar{\tau}\right)  ,w\left(  \bar{\tau};\bar{\tau
}\right)  \right)  d\bar{\tau}\nonumber\\
&  -\int_{R_{+}\left(  T\right)  }^{\infty}\bar{D}\left(  T,r\right)
\varphi\left(  t,r,\bar{w}\left(  T,r\right)  \right)  dr \label{J1E5}%
\end{align}

\end{lemma}

\begin{proof}
Our goal is to compute the left-hand side of (\ref{J1E5}). To this end, we
replace the variable of integration $t$ by $\tau$ using (\ref{F3E5}). Due to
Lemma \ref{auxFunctions} we can also replace the variable $r$ by $\bar{\tau},$
using the change of variable $r=r\left(  \tau;\bar{\tau}\right)  .$ Then,
using that $\bar{w}\left(  t,r\left(  \tau;\bar{\tau}\left(  \tau,r\right)
\right)  \right)  =w\left(  \tau;\bar{\tau}\right)  $ with $\bar{\tau}\left(
\tau,r\right)  $ as in Lemma \ref{auxFunctions} we obtain:%
\begin{align*}
&  K\equiv\int_{t_{0}}^{0}\int_{R_{+}\left(  t\right)  }^{\infty}D\left(
t,r\right)  \Delta\left(  t,r,w\left(  t,r\right)  \right)  drdt\\
&  =\int_{\log\left(  \frac{1}{\left(  -t_{0}\right)  }\right)  }^{\infty}%
\int_{\tau}^{\infty}D\left(  \tau;\bar{\tau}\right)  \Delta\left(  t,r\left(
\tau;\bar{\tau}\right)  ,w\left(  \tau;\bar{\tau}\right)  \right)  e^{-\tau
}\frac{\partial r\left(  \tau;\bar{\tau}\right)  }{\partial\bar{\tau}}%
d\bar{\tau}d\tau
\end{align*}

Using Lemma \ref{Delta1} we then obtain:%
\[
K=\int_{\log\left(  \frac{1}{\left(  -t_{0}\right)  }\right)  }^{\infty}%
\int_{\tau}^{\infty}\left[  \frac{\partial r\left(  \tau;\bar{\tau}\right)
}{\partial\bar{\tau}}D\left(  \tau;\bar{\tau}\right)  \right]  \frac{d}{d\tau
}\left(  \varphi\left(  t,r\left(  \tau;\bar{\tau}\right)  ,w\left(  \tau
;\bar{\tau}\right)  \right)  \right)  d\bar{\tau}d\tau
\]

Using then Lemma \ref{geomConst} which implies that $\left[  \frac{\partial
r\left(  \tau;\bar{\tau}\right)  }{\partial\bar{\tau}}D\left(  \tau;\bar{\tau
}\right)  \right]  $ is constant almost everywhere we obtain:%
\begin{align}
K  &  =\int_{\log\left(  \frac{1}{\left(  -t_{0}\right)  }\right)  }^{\infty
}d\bar{\tau}\frac{\partial r\left(  \bar{\tau};\bar{\tau}\right)  }%
{\partial\bar{\tau}}D\left(  \bar{\tau};\bar{\tau}\right)  \varphi\left(
\bar{t},r\left(  \bar{\tau};\bar{\tau}\right)  ,w\left(  \bar{\tau};\bar{\tau
}\right)  \right) \label{J1E6}\\
&  -\int_{\log\left(  \frac{1}{\left(  -t_{0}\right)  }\right)  }^{\infty
}d\bar{\tau}\left.  \left[  \frac{\partial r\left(  \tau;\bar{\tau}\right)
}{\partial\bar{\tau}}D\left(  \tau;\bar{\tau}\right)  \right]  \varphi\left(
t,r\left(  \tau;\bar{\tau}\right)  ,w\left(  \tau;\bar{\tau}\right)  \right)
\right\vert _{\tau=\log\left(  \frac{1}{\left(  -t_{0}\right)  }\right)
}\nonumber
\end{align}

The last term in (\ref{J1E6}) can be transformed into an integral on
$r>r_{+}\left(  T\right)  $ using again the change of variables $r=r\left(
\tau;\bar{\tau}\right)  .$ We then obtain (\ref{J1E5}).
\end{proof}

We now need the following result, which is basically an auxiliary computation.

\begin{lemma}
\label{LalgId}Suppose that $\bar{w}_{k},\bar{D}_{k}\in C^{1}\left(  \left\{
t_{0}<t<0,\ y_{0}\sqrt{-t}<r<R_{+}\left(  t\right)  \right\}  \right)
,\ k=1,2$ satisfy:\
\begin{align}
\partial_{t}\bar{w}_{k}\left(  t,r\right)  +a_{k}\left(  t,r,\bar{w}\left(
t,r\right)  \right)  \partial_{r}\bar{w}_{k}\left(  t,r\right)   &
=Q_{1}\left(  t,r,\bar{w}_{k}\left(  t,r\right)  \right)  \ \label{J2E4}\\
\partial_{t}\bar{D}_{k}\left(  t,r\right)  +a_{k}\left(  t,r,\bar{w}%
_{k}\left(  t,r\right)  \right)  \partial_{r}\bar{D}_{k}\left(  t,r\right)
&  =Q_{2}\left(  t,r,\bar{w}_{k}\left(  t,r\right)  \right)  \bar{D}%
_{k}\left(  t,r\right)  \ \label{J2E5}%
\end{align}
where
\begin{equation}
Q_{1}\left(  t,r,\bar{w}_{k}\right)  =\frac{e^{\mu}}{r\tilde{E}_{k}}\left(
\frac{1}{r^{2}}-\frac{\left(  e^{2\lambda}-1\right)  }{2}\left(  \frac
{1}{r^{2}}+2e^{-2\lambda}\bar{w}_{k}^{2}\right)  \right)  \label{J2E5a}%
\end{equation}%
\begin{equation}
Q_{2}\left(  t,r,\bar{w}_{k}\right)  =-\frac{e^{\mu-2\lambda}}{\tilde{E}%
}\left(  \frac{\partial_{r}\bar{w}_{k}}{\tilde{E}_{k}^{2}r^{2}}-\frac
{\lambda_{r}\bar{w}_{k}}{\tilde{E}_{k}^{2}r^{2}}+\left(  \mu_{r}-\lambda
_{r}\right)  \bar{w}_{k}+\frac{\bar{w}_{k}}{r^{3}\tilde{E}_{k}^{2}}\right)
\label{J2E5b}%
\end{equation}%
\begin{equation}
Q_{3}\left(  t,r,\bar{w}_{k}\right)  =\left(  \frac{\lambda_{r}e^{\mu
-2\lambda}\bar{w}_{k}^{2}}{\tilde{E}}-e^{\mu}\mu_{r}\tilde{E}+e^{\mu}\frac
{1}{r^{3}\tilde{E}}\right)  \label{J2E5c}%
\end{equation}
$\ $%
\begin{equation}
a_{k}\left(  t,r,\bar{w}_{k}\right)  =e^{\mu-2\lambda}\frac{\bar{w}_{k}%
}{\tilde{E}_{k}} \label{J2E5d}%
\end{equation}
where $\lambda,\mu\in C^{1}\left(  \left\{  t_{0}<t<0,\ y_{0}\sqrt{-t}%
<r<R_{+}\left(  t\right)  \right\}  \right)  .$ Then:%
\begin{equation}
\frac{d}{dt}\left(  \varphi\left(  t,r,\bar{w}_{k}\left(  t,r\right)  \right)
\right)  +\frac{d}{dr}\left(  a\left(  t,r,\bar{w}_{k}\left(  t,r\right)
\right)  \varphi\left(  t,r,\bar{w}_{k}\left(  t,r\right)  \right)  \right)
=J_{k}\left(  t,r\right)  +\Delta_{k}\left(  t,r\right)  \label{J2E2}%
\end{equation}
with:%
\begin{align}
J_{k}\left(  t,r\right)   &  =\left[  Q_{1}\left(  t,r,\bar{w}_{k}\left(
t,r\right)  \right)  -Q_{3}\left(  t,r,\bar{w}_{k}\left(  t,r\right)  \right)
\right]  \left(  \partial_{w}\varphi\right)  \left(  t,r,\bar{w}_{k}\left(
t,r\right)  \right) \label{J2E3}\\
&  +\left[  a_{w}\left(  t,r,\bar{w}_{k}\left(  t,r\right)  \right)  \left(
\partial_{r}\bar{w}_{k}\right)  \left(  t,r\right)  -\left(  \partial_{w}%
Q_{3}\right)  \left(  t,r,\bar{w}_{k}\left(  t,r\right)  \right)  \right]
\varphi\left(  t,r,\bar{w}_{k}\left(  t,r\right)  \right) \nonumber
\end{align}
and $\Delta_{k}$ is defined for each $k=1,2$ as $\Delta_{k}\left(  t,r\right)
=\Delta\left(  t,r,\bar{w}_{k}\left(  t,r\right)  \right)  $ where $\Delta$ is
as in (\ref{Y1E8}).
\end{lemma}

\begin{remark}
Notice that we do not require $\lambda,\ \mu$ to solve the equations for the
fields. They can be general arbitrary fields.
\end{remark}

\begin{proof}
Using (\ref{J2E4}) we obtain:%
\begin{align}
&  \frac{d}{dt}\left(  \varphi\left(  t,r,\bar{w}_{k}\left(  t,r\right)
\right)  \right)  +\frac{d}{dr}\left(  a_{k}\left(  t,r,\bar{w}_{k}\left(
t,r\right)  \right)  \varphi\left(  t,r,\bar{w}_{k}\left(  t,r\right)
\right)  \right)  \nonumber\\
&  =\left(  \partial_{t}\varphi\right)  \left(  t,r,\bar{w}_{k}\left(
t,r\right)  \right)  +\left(  \partial_{w}\varphi\right)  \left(  t,r,\bar
{w}_{k}\left(  t,r\right)  \right)  Q_{1}\left(  t,r,\bar{w}_{k}\left(
t,r\right)  \right)  \nonumber\\
&  +a_{k,r}\left(  t,r,\bar{w}_{k}\left(  t,r\right)  \right)  \varphi\left(
t,r,\bar{w}_{k}\left(  t,r\right)  \right)  +a_{k,w}\left(  t,r,\bar{w}%
_{k}\left(  t,r\right)  \right)  \left(  \partial_{r}\bar{w}_{k}\right)
\left(  t,r\right)  \varphi\left(  t,r,\bar{w}_{k}\left(  t,r\right)  \right)
\nonumber\\
&  +a_{k}\left(  t,r,\bar{w}_{k}\left(  t,r\right)  \right)  \left(
\partial_{r}\varphi\right)  \left(  t,r,\bar{w}_{k}\left(  t,r\right)
\right)  \label{J2E6}%
\end{align}

We remark now that the definition of $\Delta$ in (\ref{Y1E8}) yields:%
\begin{align*}
\Delta_{k}\left(  t,r\right)   &  =\partial_{t}\varphi\left(  t,r,\bar{w}%
_{k}\right)  +\partial_{r}\left(  a_{k}\left(  t,r,\bar{w}_{k}\right)
\varphi\right)  +\partial_{w}\left(  Q_{3}\left(  t,r,\bar{w}_{k}\right)
\varphi\right) \\
&  =\partial_{t}\varphi\left(  t,r,\bar{w}_{k}\right)  +\left(  \partial
_{r}a_{k}\right)  \left(  t,r,\bar{w}_{k}\right)  \varphi\left(  t,r,\bar
{w}_{k}\right)  +a\left(  t,r,\bar{w}_{k}\right)  \left(  \partial_{r}%
\varphi\right)  \left(  t,r,\bar{w}\right)  +\partial_{w}\left(  Q_{3}\left(
t,r,\bar{w}_{k}\right)  \varphi\right)
\end{align*}

Combining this identity with (\ref{J2E6}) we obtain:%
\begin{align*}
&  \frac{d}{dt}\left(  \varphi\left(  t,r,\bar{w}_{k}\left(  t,r\right)
\right)  \right)  +\frac{d}{dr}\left(  a_{k}\left(  t,r,\bar{w}_{k}\left(
t,r\right)  \right)  \varphi\left(  t,r,\bar{w}_{k}\left(  t,r\right)
\right)  \right)  \\
&  =\left(  \partial_{w}\varphi\right)  \left(  t,r,\bar{w}_{k}\left(
t,r\right)  \right)  Q_{1}\left(  t,r,\bar{w}_{k}\left(  t,r\right)  \right)
\\
&  +a_{k,r}\left(  t,r,\bar{w}_{k}\left(  t,r\right)  \right)  \varphi\left(
t,r,\bar{w}_{k}\left(  t,r\right)  \right)  +a_{k,w}\left(  t,r,\bar{w}%
_{k}\left(  t,r\right)  \right)  \left(  \partial_{r}\bar{w}_{k}\right)
\left(  t,r\right)  \varphi\left(  t,r,\bar{w}_{k}\left(  t,r\right)  \right)
\\
&  +\Delta_{k}\left(  t,r\right)  -\left(  \partial_{w}\left(  Q_{3}%
\varphi\right)  \right)  \left(  t,r,\bar{w}_{k}\left(  t,r\right)  \right)
-\left(  \partial_{r}a_{k}\right)  \left(  t,r,\bar{w}_{k}\left(  t,r\right)
\right)  \varphi\left(  t,r,\bar{w}_{k}\left(  t,r\right)  \right)
\end{align*}
and, after some computations we arrive at:%
\begin{align*}
&  \frac{d}{dt}\left(  \varphi\left(  t,r,\bar{w}_{k}\left(  t,r\right)
\right)  \right)  +\frac{d}{dr}\left(  a_{k}\left(  t,r,\bar{w}_{k}\left(
t,r\right)  \right)  \varphi\left(  t,r,\bar{w}_{k}\left(  t,r\right)
\right)  \right)  \\
&  =\left[  Q_{1}\left(  t,r,\bar{w}_{k}\left(  t,r\right)  \right)
-Q_{3}\left(  t,r,\bar{w}_{k}\left(  t,r\right)  \right)  \right]  \left(
\partial_{w}\varphi\right)  \left(  t,r,\bar{w}_{k}\left(  t,r\right)
\right)  \\
&  +\left[  a_{k,w}\left(  t,r,\bar{w}_{k}\left(  t,r\right)  \right)  \left(
\partial_{r}\bar{w}_{k}\right)  \left(  t,r\right)  -\left(  \partial_{w}%
Q_{3}\right)  \left(  t,r,\bar{w}_{k}\left(  t,r\right)  \right)  \right]
\varphi\left(  t,r,\bar{w}_{k}\left(  t,r\right)  \right)  \\
&  +\Delta_{k}\left(  t,r\right)
\end{align*}

Using then (\ref{J2E3}) we obtain (\ref{J2E2}).
\end{proof}

We can now compute a simpler form for $J_{k}\left(  t,r\right)  .$

\begin{lemma}
\label{LalgIdbis}Suppose that the conditions of Lemma \ref{LalgId} are
satisfied. Let us assume also that the field equations (\ref{J3E6}) hold, with
$v_{k},\ B_{k}$ and $\bar{w}_{k},$ $\bar{D}_{k}$ related by means of
(\ref{Y2E3}). Then, the following identity holds:%
\begin{equation}
J_{k}\left(  t,r\right)  +Q_{2}\left(  t,r,\bar{w}_{k}\left(  t,r\right)
\right)  \varphi\left(  t,r,\bar{w}_{k}\left(  t,r\right)  \right)  =0
\label{J2E7}%
\end{equation}
where $J_{k}$ is as (\ref{J2E3}) and $Q_{2}$ is as in (\ref{J2E5b}).
\end{lemma}

\begin{proof}
Using (\ref{J2E5a}), (\ref{J2E5c}) we obtain:%
\begin{align*}
&  Q_{1}\left(  t,r,\bar{w}_{k}\right)  -Q_{3}\left(  t,r,\bar{w}_{k}\right)
\\
&  =\frac{e^{\mu}}{r\tilde{E}}\left(  -\frac{\left(  e^{2\lambda}-1\right)
}{2}\left(  \frac{1}{r^{2}}+2e^{-2\lambda}\bar{w}_{k}^{2}\right)
-r\lambda_{r}e^{-2\lambda}\bar{w}_{k}^{2}+r\mu_{r}\tilde{E}^{2}\right)
\end{align*}

Using now (\ref{T1E3a}) and (\ref{Y2E3}) to eliminate the derivatives
$\lambda_{r},\ \mu_{r}$ we obtain, after some computations:%
\begin{align}
&  Q_{1}\left(  t,r,\bar{w}_{k}\right)  -Q_{3}\left(  t,r,\bar{w}_{k}\right)
\nonumber\\
&  =\frac{e^{\mu}}{\tilde{E}}\left[  -\left(  \frac{\left(  e^{2\lambda
}-1\right)  }{2r}\left(  \frac{1}{r^{2}}+2e^{-2\lambda}\bar{w}_{k}^{2}\right)
\right)  +\left(  \frac{e^{2\lambda}-1}{2r}\right)  \left(  2\bar{w}_{k}%
^{2}e^{-2\lambda}+\frac{1}{r^{2}}\right)  \right] \nonumber\\
&  =0 \label{J3E1}%
\end{align}

We now compute $\left[  a_{k,w}\left(  t,r,\bar{w}\left(  t,r\right)  \right)
\left(  \partial_{r}\bar{w}_{k}\right)  \left(  t,r\right)  -\left(
\partial_{w}Q_{3}\right)  \left(  t,r,\bar{w}_{k}\left(  t,r\right)  \right)
+Q_{2}\left(  t,r,w_{k}\left(  t,r\right)  \right)  \right]  .$ Using
(\ref{J2E5b})-(\ref{J2E5d}) we arrive at:%
\begin{align}
&  \left[  a_{k,w}\left(  t,r,\bar{w}_{k}\left(  t,r\right)  \right)  \left(
\partial_{r}\bar{w}_{k}\right)  \left(  t,r\right)  -\left(  \partial_{w}%
Q_{3}\right)  \left(  t,r,\bar{w}_{k}\left(  t,r\right)  \right)
+Q_{2}\left(  t,r,\bar{w}_{k}\left(  t,r\right)  \right)  \right] \nonumber\\
&  =\frac{\lambda_{r}e^{\mu-4\lambda}\bar{w}_{k}^{3}}{\sqrt{\left(  \bar
{w}_{k}^{2}e^{-2\lambda}+\frac{1}{r^{2}}\right)  ^{3}}}-\frac{e^{\mu-4\lambda
}\bar{w}_{k}^{3}\lambda_{r}}{\tilde{E}_{k}^{3}}=0 \label{J3E2}%
\end{align}

Combining (\ref{J2E3}), (\ref{J3E1}), (\ref{J3E2}) we obtain (\ref{J2E7}).
\end{proof}

We now combine Lemmas \ref{LalgId}, \ref{LalgIdbis} in order to rewrite the
integrals  \linebreak$\int_{t_{0}}^{0}\int_{0}^{R_{+}\left(  t\right)  }\bar{D}%
_{k}\left(  t,r\right)  \Delta_{k}\left(  t,r,\bar{w}_{k}\left(  t,r\right)
\right)  drdt$ in terms of initial and boundary values.

\begin{lemma}
\label{BoundValInt}Suppose that $\bar{w}_{k},\bar{D}_{k},\ \Delta_{k}$ are as
in Lemma \ref{LalgId} and that the conditions of Lemmas \ref{LalgId},
\ref{LalgIdbis} hold. Then, the following identity holds:%
\begin{align}
&  \int_{t_{0}}^{0}\int_{0}^{R_{+}\left(  t\right)  }\bar{D}_{k}\left(
t,r\right)  \Delta_{k}\left(  t,r\right)  drdt\ \label{J3E8}\\
&  =-\int_{0}^{R_{+}\left(  t_{0}\right)  }\bar{D}_{k}\left(  t_{0},r\right)
\varphi\left(  t_{0},r,\bar{w}_{k}\left(  t_{0},r\right)  \right)
dr\nonumber\\
&  -\int_{t_{0}}^{0}D_{k}\left(  t,R_{k,+}\left(  t\right)  \right)
\varphi\left(  t,R_{+}\left(  t\right)  ,w_{k}\left(  t,R_{+}\left(  t\right)
\right)  \right)  \frac{dR_{k,+}\left(  t\right)  }{dt}dt\nonumber\\
&  +\int_{t_{0}}^{0}D_{k}\left(  t,R_{k,+}\left(  t\right)  \right)
a_{k}\left(  t,R_{k,+}\left(  t\right)  ,w_{k}\left(  t,R_{k,+}\left(
t\right)  \right)  \right)  \varphi\left(  t,R_{k,+}\left(  t\right)
,w_{k}\left(  t,R_{k,+}\left(  t\right)  \right)  \right)  dt\nonumber
\end{align}
for $k=1,2,$ where $R_{k,+}\left(  t\right)  =R_{\max}$ if $k=1$ and
$R_{k,+}\left(  t\right)  =R_{+}\left(  t\right)  $ if $k=2.$
\end{lemma}

\begin{proof}
Using (\ref{J2E2}) we can rewrite the left-hand side of (\ref{J3E8}) as:%
\begin{align*}
&  \int_{t_{0}}^{0}\int_{0}^{R_{+}\left(  t\right)  }\bar{D}_{k}\left(
t,r\right)  \Delta_{k}\left(  t,r\right)  drdt\\
&  =\int_{t_{0}}^{0}\frac{d}{dt}\left(  \int_{0}^{R_{k,+}\left(  t\right)
}\bar{D}_{k}\left(  t,r\right)  \varphi\left(  t,r,\bar{w}_{k}\left(
t,r\right)  \right)  dr\right)  dt\\
&  -\int_{t_{0}}^{0}\bar{D}_{k}\left(  t,R_{+}\left(  t\right)  \right)
\varphi\left(  t,R_{k,+}\left(  t\right)  ,\bar{w}_{k}\left(  t,R_{k,+}\left(
t\right)  \right)  \right)  \frac{dR_{k,+}\left(  t\right)  }{dt}dt\\
&  -\int_{t_{0}}^{0}\int_{0}^{R_{k,+}\left(  t\right)  }\left(  \partial
_{t}\bar{D}_{k}\left(  t,r\right)  \right)  \varphi\left(  t,r,\bar{w}%
_{k}\left(  t,r\right)  \right)  drdt\\
&  +\int_{t_{0}}^{0}\int_{0}^{R_{k,+}\left(  t\right)  }\bar{D}_{k}\left(
t,r\right)  \frac{d}{dr}\left(  a_{k}\left(  t,r,\bar{w}_{k}\left(
t,r\right)  \right)  \varphi\left(  t,r,\bar{w}\left(  t,r\right)  \right)
\right)  drdt\\
&  -\int_{t_{0}}^{0}\int_{0}^{R_{k,+}\left(  t\right)  }\bar{D}_{k}\left(
t,r\right)  J_{k}\left(  t,r\right)  drdt
\end{align*}

Integrating by parts and using (\ref{J2E5}) we obtain:%
\begin{align*}
&  \int_{t_{0}}^{0}\int_{R_{+}\left(  t\right)  }^{\infty}\bar{D}_{k}\left(
t,r\right)  \Delta_{k}\left(  t,r\right)  drdt\\
&  =-\int_{0}^{R_{+}\left(  t_{0}\right)  }\bar{D}_{k}\left(  t_{0},r\right)
\varphi\left(  t_{0},r,\bar{w}_{k}\left(  t_{0},r\right)  \right)  dr\\
&  -\int_{t_{0}}^{0}\bar{D}_{k}\left(  t,R_{k,+}\left(  t\right)  \right)
\varphi\left(  t,R_{k,+}\left(  t\right)  ,\bar{w}_{k}\left(  t,R_{k,+}\left(
t\right)  \right)  \right)  \frac{dR_{k,+}\left(  t\right)  }{dt}dt\\
&  +\int_{t_{0}}^{0}\bar{D}_{k}\left(  t,R_{k,+}\left(  t\right)  \right)
a\left(  t,R_{k,+}\left(  t\right)  ,\bar{w}_{k}\left(  t,R_{k,+}\left(
t\right)  \right)  \right)  \varphi\left(  t,R_{k,+}\left(  t\right)  ,\bar
{w}_{k}\left(  t,R_{k,+}\left(  t\right)  \right)  \right)  dt\\
&  -\int_{t_{0}}^{0}\int_{0}^{R_{+}\left(  t\right)  }Q_{2}\left(  t,r,\bar
{w}_{k}\left(  t,r\right)  \right)  \bar{D}_{k}\left(  t,r\right)
\varphi\left(  t,r,\bar{w}_{k}\left(  t,r\right)  \right)  drdt\\
&  -\int_{t_{0}}^{0}\int_{0}^{R_{k,+}\left(  t\right)  }\bar{D}_{k}\left(
t,r\right)  J_{k}\left(  t,r\right)  drdt
\end{align*}

Using then (\ref{J2E7}) we obtain (\ref{J3E8}) and the result follows.
\end{proof}

We can finally conclude the proof of Proposition \ref{ws}.

\begin{proof}
[End of the proof of Proposition \ref{ws}]Given a test function $\varphi$ we
can split it into the sum of two pieces $\varphi_{1},\ \varphi_{2}$ where the
support of $\varphi_{1}$ is in the region $r\leq\frac{2R_{+}\left(  t\right)
}{3}$ and the support of $\varphi_{2}$ is in the set $r\geq\frac{R_{+}\left(
t\right)  }{2}.$ Using Theorem \ref{RV} we obtain that the corresponding
formula (\ref{S4E2}) holds for $\varphi_{1}.$ We then need to check
(\ref{S4E2}) for $\varphi_{2}.$

The continuity of the function $\Psi\left(  t,r,\bar{v}\right)  $ in
$S\cap\left\{  r\geq R_{+}\left(  t\right)  \right\}  $ follows from the
properties of the functions $v_{1},\ B_{1},\ \lambda,\ \mu.$ On the other
hand, using Lemma \ref{LIntegral} we obtain:%
\begin{align}
&  \int\int_{S\cap\left\{  r>R_{+}\left(  t\right)  \right\}  }\tilde{\zeta
}\left(  t,r,\bar{v}\right)  \Delta\left(  t,r,\bar{v}\right)  drd\bar
{v}dt\label{J4E1}\\
&  =\int_{T}^{\infty}\frac{\partial r\left(  \bar{\tau};\bar{\tau}\right)
}{\partial\bar{\tau}}D\left(  \bar{\tau};\bar{\tau}\right)  \varphi\left(
\bar{t},r\left(  \bar{\tau};\bar{\tau}\right)  ,\bar{w}\left(  \bar{\tau}%
;\bar{\tau}\right)  \right)  d\bar{\tau}\nonumber\\
&  -\int_{r_{+}\left(  T\right)  }^{\infty}dr\bar{D}\left(  T,r\right)
\varphi\left(  t,r,\bar{w}\left(  r,T\right)  \right)  \nonumber
\end{align}

On the other hand, using the fact that $\Delta\left(  t,r,\bar{v}\right)  $ in
$S\cap\left\{  r<R_{+}\left(  t\right)  \right\}  $ can be written as the sum
$\sum_{k=1}^{2}\Delta_{k}\left(  t,r\right)  ,$ as well as Lemma
\ref{BoundValInt} we arrive at:%
\begin{align}
&  \int\int_{S\cap\left\{  r<R_{+}\left(  t\right)  \right\}  }\tilde{\zeta
}\left(  t,r,\bar{v}\right)  \Delta\left(  t,r,\bar{v}\right)  drd\bar
{v}dt\label{J4E2}\\
&  =-\sum_{k=1}^{2}\int_{0}^{R_{+}\left(  t_{0}\right)  }\bar{D}_{k}\left(
t_{0},r\right)  \varphi\left(  t_{0},r,\bar{w}_{k}\left(  t_{0},r\right)
\right)  dr\nonumber\\
&  -\sum_{k=1}^{2}\int_{t_{0}}^{0}\bar{D}_{k}\left(  t,R_{k,+}\left(
t\right)  \right)  \varphi\left(  t,R_{+}\left(  t\right)  ,\bar{w}_{k}\left(
t,R_{+}\left(  t\right)  \right)  \right)  \frac{dR_{k,+}\left(  t\right)
}{dt}dt\nonumber\\
&  +\sum_{k=1}^{2}\int_{t_{0}}^{0}\bar{D}_{k}\left(  t,R_{k,+}\left(
t\right)  \right)  a_{k}\left(  t,R_{k,+}\left(  t\right)  ,\bar{w}_{k}\left(
t,R_{k,+}\left(  t\right)  \right)  \right)  \varphi\left(  t,R_{k,+}\left(
t\right)  ,\bar{w}_{k}\left(  t,R_{k,+}\left(  t\right)  \right)  \right)
dt\nonumber
\end{align}

Adding (\ref{J4E1}), (\ref{J4E2}) we obtain the cancellation of several terms.
Indeed, the terms with $k=2$ in (\ref{J4E2}) cancel out due to the definition
of the function $R_{+}\left(  t\right)  =R_{2,+}\left(  t\right)  .$ On the
other hand, the terms with $k=1$ along the boundary $r=R_{+}\left(  t\right)
$ can be computed as follows. We have:%
\[
r\left(  \bar{\tau};\bar{\tau}\right)  =r_{+}\left(  \bar{\tau}\right)
\]

Differentiating this formula we obtain:%
\[
\frac{\partial r\left(  \bar{\tau};\bar{\tau}\right)  }{\partial\tau}%
+\frac{\partial r\left(  \bar{\tau};\bar{\tau}\right)  }{\partial\bar{\tau}%
}=\frac{dr_{+}\left(  \bar{\tau}\right)  }{d\bar{\tau}}%
\]

Using that $\frac{\partial r\left(  \bar{\tau};\bar{\tau}\right)  }%
{\partial\tau}=a_{1}$ we obtain:%
\[
a_{1}\left(  R_{+}\left(  \bar{t}\right)  ,w_{k}\left(  R_{+}\left(  \bar
{t}\right)  ,\bar{t}\right)  ,\bar{t}\right)  +\frac{\partial r\left(
\bar{\tau};\bar{\tau}\right)  }{\partial\bar{\tau}}=\frac{dr_{+}\left(
\bar{\tau}\right)  }{d\bar{\tau}}%
\]

It then follows that the sum of the terms in (\ref{J4E1}), (\ref{J4E2}) cancel
out. Therefore:%
\begin{align*}
&  \int\int_{S}\tilde{\zeta}\left(  t,r,\bar{v}\right)  \Delta\left(
t,r,\bar{v}\right)  drd\bar{v}dt\\
&  =-\int_{R_{+}\left(  t_{0}\right)  }^{\infty}D\left(  t_{0},r\right)
\varphi\left(  t_{0},r,w\left(  t_{0},r\right)  \right)  dr-\sum_{k=1}^{2}%
\int_{0}^{R_{+}\left(  t_{0}\right)  }D_{k}\left(  t_{0},r\right)
\varphi\left(  t_{0},r,w_{k}\left(  t_{0},r\right)  \right)  dr
\end{align*}
and Proposition \ref{ws} follows.
\end{proof}

Our next goal is to prove the following result.

\begin{proposition}
\label{ext}There exists $t_{0}<0$ and functions $\left(  r,w,D\right)
\in\mathcal{X}_{L,T}$ with $L>L_{0}$ and $T>T_{0},$ where $L_{0},\ T_{0}$ are
sufficiently large, as well as functions$\ \lambda,\ \mu$ defined as in Lemma
\ref{defFields} \ with $\rho$ and $p$ defined as in (\ref{S1E7}), (\ref{S1E8})
where the functions $\left(  r,w,D\right)  $ are a solution of (\ref{Y2E4}%
)-(\ref{Y2E7}) in the sense of characteristics, as in Definition
\ref{delSolChar}.
\end{proposition}

The proof of Proposition \ref{ext} will be carried out by means of a fixed
point argument. In order to formulate this argument in a precise manner, it is
convenient to reformulate the equations (\ref{T1E1}), (\ref{T1E2}) in a way,
that will allow to derive suitable regularity estimates for some of the
functions involved.

\subsection{Introducing a topology in $\mathcal{Y}_{L,T,a}$.\label{Topo}}

The next goal is to define a suitable weak topology in the set of functions
$\mathcal{Y}_{L,T,a}$. Suppose that $\tau^{\ast}\geq T.$ For any
$\varepsilon>0$ and $\varphi\in C_{0}^{\infty}\left[  r_{+}\left(  \tau^{\ast
}\right)  ,\infty\right)  ,$ we define a neighbourhood of a function $\bar
{D}^{\left(  \infty\right)  }$ in $\mathcal{A}_{\tau^{\ast}}=L^{\infty}\left[
r_{+}\left(  \tau^{\ast}\right)  ,\infty\right)  $\ as the set of functions
$\bar{D}$ such that $D$ satisfies (\ref{A1E1}) and
\begin{equation}
\left\vert \int_{\left[  r_{+}\left(  \tau^{\ast}\right)  ,\infty\right)
}\varphi\left(  r\right)  \left[  \bar{D}\left(  \tau^{\ast},r\right)
-\bar{D}^{\left(  \infty\right)  }\left(  \tau^{\ast},r\right)  \right]
dr\right\vert <\varepsilon\label{W2E1}%
\end{equation}

Classical results show that the topology defined by means of these functionals
is metrizable (cf. \cite{Br}). In particular convergence can be characterized
by sequential limits. Notice that the spaces (and the corresponding metric)
change with $\tau^{\ast}$ because the domain of the functions depends on
$\tau^{\ast}.$ In any case, we can define a topology for the functions
$\bar{D}.$ We will denote as $d_{\tau^{\ast}}$ the metric associated to the
weak topology defined in $\mathcal{A}_{\tau^{\ast}}.$ We will assume in all
the following that the space $\mathcal{A}_{\tau^{\ast}}$ is endowed with the
topology generated by the functionals (\ref{W2E1}). We then introduce a
topology for the functions $\bar{D}$ defined in a suitable space termed as:%

\begin{equation}
C\left(  \left[  T,\infty\right)  :\mathcal{A}_{\tau^{\ast}}\right)
\ \label{W2E1a}%
\end{equation}
where by this space we mean a set of functions $\bar{D}\left(  \tau^{\ast
},\cdot\right)  \in\mathcal{A}_{\tau^{\ast}}.$ Notice that these functions can
be also thought of as functions defined in the domain $\left\{  T\leq\tau^{\ast
}<\infty,\ r_{+}\left(  \tau^{\ast}\right)  <r<\infty\right\}  .$ The metric
of this space is given by means of:%
\begin{equation}
\operatorname*{dist}\left(  \bar{D}_{1}^{\left(  1\right)  },\bar{D}%
_{1}^{\left(  2\right)  }\right)  =\sup_{\tau^{\ast}\in\left[  T,\infty
\right)  }d_{\tau^{\ast}}\left(  \bar{D}_{1}^{\left(  1\right)  }\left(
\tau^{\ast},\cdot\right)  ,\bar{D}_{1}^{\left(  2\right)  }\left(  \tau^{\ast
},\cdot\right)  \right)  \label{W2E2}%
\end{equation}

Notice that there is not any problem defining the integrals in (\ref{W2E1})
because the test functions are compactly supported. The topology generated by
these functionals is metrizable due to (\ref{A1E1}).

We need to obtain a criteria of compactness in $\mathcal{Y}_{L,T,a}.$ To this
end, we introduce the following definition.

\begin{definition}
\label{equicont}We will say that the set of functions $\mathcal{F}%
\subset\mathcal{Y}_{L,T,a}$ is equicontinuous if for any $\varepsilon>0$ and
any test function $\varphi\in C_{0}^{\infty}\left(  \left[  0,\infty\right)
\right)  $ there exists $\delta=\delta\left(  \varepsilon,\varphi\right)  >0$
such that, for any $\tau_{1},\tau_{2}\in\left[  T,\infty\right)  $ satisfying
$\left\vert \tau_{1}-\tau_{2}\right\vert \leq\delta$ and any $\left(
r,w,D\right)  \in\mathcal{F}$ we have:%
\[
\left\vert \int_{\left[  r_{+}\left(  \tau_{1}\right)  ,\infty\right)  }%
\bar{D}\left(  \tau_{1},r\right)  \varphi\left(  r-r_{+}\left(  \tau
_{1}\right)  \right)  dr-\int_{\left[  r_{+}\left(  \tau_{2}\right)
,\infty\right)  }\bar{D}\left(  \tau_{2},r\right)  \varphi\left(
r-r_{+}\left(  \tau_{2}\right)  \right)  dr\right\vert \leq\varepsilon
\]
where $\bar{D}$ is the function defined in Lemma \ref{auxFunctions}.
\end{definition}

\begin{remark}
The functions $\bar{D}\left(  \cdot,\tau\right)  $ are defined in the domains
$\left(  r_{+}\left(  t_{2}\right)  ,\infty\right)  .$ A translation has been
made in order to bring the spaces to a fixed domain, and use test functions
$\varphi$ defined in a fixed domain to measure their distance.
\end{remark}

We have the following result:

\begin{proposition}
\label{CompEqui}Suppose that $\mathcal{F}\subset\mathcal{Y}_{L,T,a}$ is
equicontinuous in the sense of Definition \ref{equicont}. Then, there exists a
subsequence $\left\{  D_{n}\right\}  $ such that the corresponding functions
$\bar{D}_{n}$ defined by means of Lemma \ref{auxFunctions} converge in the
space $C\left(  \left[  T,\infty\right)  :\mathcal{A}_{\tau^{\ast}}\right)  $
to some function $\bar{D}\in L^{\infty}\left(  \left(  \tau,r\right)  :r\geq
r_{+}\left(  \tau\right)  ,\ \tau\geq T\right)  .$
\end{proposition}

\begin{proof}
The result is just a consequence of the classical Arzela-Ascoli Theorem in
arbitrary metric spaces (cf. (\cite{DS}), (\cite{F})).
\end{proof}

We now prove the following:

\begin{lemma}
Suppose that $\mathcal{F}\subset\mathcal{Y}_{L,T,a}$ is equicontinuous in the
sense of Definition \ref{equicont}. Then $\mathcal{F}$ is compact in the
topological space $\mathcal{X}_{L,T}$ endowed with the topology generated by
(\ref{W2E1}), (\ref{W2E2}).
\end{lemma}

\begin{proof}
We first prove that the space $\mathcal{Y}_{L,T,a}$ is closed with this
topology. To this end, we just notice that (\ref{A1E1a}) is preserved under
limits in this topology. The topology defined in $\mathcal{Y}_{L,T,a}$ is
metrizable. Therefore, we can restrict our analysis to the convergence of
sequences. The conservation of the inequality (\ref{A1E1}) follows from a
general argument yielding preservation of inequalities for weak
topologies.\ Indeed, given a test function $\varphi\in C_{0}^{\infty}\left(
\left\{  T\leq\tau\leq\bar{\tau}\right\}  \right)  $ suppose that we have a
sequence $\left\{  f_{n}\right\}  \subset L^{\infty}\left(  \left\{  r\geq
r_{+}\left(  \tau\right)  ,\ \tau\geq T\right\}  \right)  $ such that $\int
f_{n}\varphi\rightarrow\int f_{\infty}\varphi$ as $n\rightarrow\infty,$ with
$f_{\infty}\in L^{\infty}\left(  \left\{  r\geq r_{+}\left(  \tau\right)
,\ \tau\geq T\right\}  \right)  .$ Suppose that the functions $f_{n}$ satisfy
the inequalites $f_{n}\leq\psi$ for $a.e.$ $\left(  \tau,r\right)  \in\left\{
r\geq r_{+}\left(  \tau\right)  ,\ \tau\geq T\right\}  .$ Choosing any test
function $\varphi\in C_{0}^{\infty}\left(  \left\{  r\geq r_{+}\left(
\tau\right)  ,\ \tau\geq T\right\}  \right)  $ such that $\varphi\geq0,$ it
then follows that:%
\[
\int f_{n}\varphi d\tau d\bar{\tau}\leq\int\psi\varphi d\tau d\bar{\tau}%
\]
and taking the limit in this inequality we obtain:%
\[
\int f_{\infty}\varphi d\tau d\bar{\tau}\leq\int\psi\varphi d\tau d\bar{\tau}%
\]
since $\varphi$ is an arbitrary, compactly supported, nonnegative function we
then obtain $f_{\infty}\leq\psi,$ $a.e.$ $\left(  \tau,r\right)  \in\left\{
r\geq r_{+}\left(  \tau\right)  ,\ \tau\geq T\right\}  .$
\end{proof}

We will use the following auxiliary result.

\begin{lemma}
\label{CWFields}Given $\left(  r,w,D\right)  \in\mathcal{X}_{L,T}$ we define
$\bar{w},\ \bar{D}$ as in Lemma \ref{auxFunctions}. We then define the fields
$\bar{\mu},\ \lambda$ for $r\geq r_{+}\left(  \tau\right)  ,\ \tau>T,$ by
means of the following ODE problem:%
\begin{align}
e^{-2\lambda}\left(  2r\lambda_{r}-1\right)  +1  &  =8\pi^{2}e^{-\lambda
}\tilde{E}\bar{D}\ \ ,\ \ r\geq r_{+}\left(  \tau\right) \label{A2E1}\\
e^{-2\lambda}\left(  2r\bar{\mu}_{r}+1\right)  -1  &  =\frac{8\pi
^{2}e^{-\lambda}}{\tilde{E}}\bar{w}^{2}\bar{D}\ \ ,\ \ r\geq r_{+}\left(
\tau\right) \label{A2E2}\\
\tilde{E}  &  =\sqrt{w^{2}e^{-2\lambda}+\frac{1}{r^{2}}}\ \label{A2E3}\\
\lambda\left(  \tau,r_{+}\left(  \tau\right)  ^{+}\right)   &  =\lambda\left(
\tau,r_{+}\left(  \tau\right)  ^{-}\right)  \ \ ,\ \ \bar{\mu}\left(
\tau,r_{+}\left(  \tau\right)  ^{+}\right)  =\bar{\mu}\left(  \tau
,r_{+}\left(  \tau\right)  ^{-}\right)  \label{A2E4}%
\end{align}
where $\lambda\left(  \tau,r_{+}\left(  \tau\right)  ^{-}\right)  ,~\bar{\mu
}\left(  \tau,r_{+}\left(  \tau\right)  ^{-}\right)  $ are computed using
(\ref{F1E2}), (\ref{F1E3}), Theorem \ref{RV}. Given $L,$ there exists
$T_{0}=T_{0}\left(  L\right)  $ such that, if $T\geq T_{0},$ the functions
$\bar{\mu},\ \lambda$ defined by means of (\ref{A2E1})-(\ref{A2E4}) are
defined for $\tau\geq T,\ r\geq r_{+}\left(  \tau\right)  $. Moreover, if we
endow the space of functions $C\left(  \left\{  \tau\geq T,\ r\geq
r_{+}\left(  \tau\right)  \right\}  \right)  $ with the topology of uniform
convergence in compact sets, the functions $\bar{\mu},\ \lambda$ define a
mapping:%
\begin{equation}
\bar{\mu},\ \lambda:\mathcal{Y}_{L,T,a}\rightarrow C\left(  \left\{  \tau\geq
T,\ r\geq r_{+}\left(  \tau\right)  \right\}  \right)  \ \label{A2E4a}%
\end{equation}
which is continuous if $\mathcal{Y}_{L,T,a}$ is endowed with the topology
generated by means of the functionals (\ref{W2E1}).
\end{lemma}

\begin{remark}
Notice that the topology of $\mathcal{X}_{L,T}$ only provides weak convergence
along lines of constant $\tau$ for the functions $\bar{D}.$
\end{remark}

\begin{proof}
Local existence of solutions follows from standard ODE Theory. In order to
show that the solution is defined for arbitrarily large values of $r,$ we use
the fact that, arguing as in the proof of Lemma \ref{defFields}, we can obtain
the following representation formula for the fields $\bar{\mu},\ \lambda$ for
$r>r_{+}\left(  \tau\right)  $%
\begin{equation}
\lambda\left[  \left(  r,w,D\right)  \right]  \left(  \tau,r\right)
=\lambda\left(  \tau,r\right)  =\frac{1}{2}\log\left(  \frac{r}{r-R_{0}\left(
\tau,r\right)  }\right)  =\frac{1}{2}\log\left(  r\right)  -\frac{1}{2}%
\log\left(  r-R_{0}\left(  \tau,r\right)  \right)  \label{A2E5}%
\end{equation}%
\begin{align}
\bar{\mu}\left[  \left(  r,w,D\right)  \right]  \left(  \tau,r\right)   &
=\bar{\mu}\left(  \tau,r\right)  =\bar{\mu}\left(  \tau,r_{+}\left(
\tau\right)  \right) \label{A2E6}\\
&  +\int_{r_{+}\left(  \tau\right)  }^{r}\frac{4\pi^{2}\exp\left(
-\lambda\left(  \tau,\xi\right)  \right)  \bar{w}^{2}\left(  \tau,\xi\right)
\bar{D}\left(  \tau,\xi\right)  }{\tilde{E}\left(  \tau,\xi\right)  }%
\frac{d\xi}{\xi-R_{0}}\nonumber\\
&  +\frac{1}{2}\int_{r_{+}\left(  \tau\right)  }^{r}\frac{d\xi}{\left(
\xi-R_{0}\left(  \tau,\xi\right)  \right)  }-\frac{1}{2}\log\left(  \frac
{r}{r_{+}\left(  \tau\right)  }\right) \nonumber
\end{align}
where:%
\begin{align}
R_{0}\left(  \tau,r\right)   &  =r_{+}\left(  \tau\right)  \left(
1-\exp\left(  -2\lambda\left(  \tau,r_{+}\left(  \tau\right)  \right)
\right)  \right) \label{A2E7}\\
&  +\int_{r_{+}\left(  \tau\right)  }^{r}8\pi^{2}\tilde{E}\left(  \tau
,\xi\right)  \exp\left(  -\lambda\left(  \tau,\xi\right)  \right)  \bar
{D}\left(  \tau,\xi\right)  d\xi\ \ ,\ \ \nonumber
\end{align}
for $\ r\geq r_{+}\left(  \tau\right)  $ and:%
\begin{equation}
\tilde{E}\left(  \tau,\xi\right)  =\tilde{E}=\sqrt{\left(  \bar{w}\left(
\tau,\xi\right)  \right)  ^{2}e^{-2\lambda\left(  \tau,\xi\right)  }+\frac
{1}{r^{2}}} \label{A2E8}%
\end{equation}

Notice that (\ref{A2E1}) implies that $\lambda\left(  \tau,r\right)  >0$ for
$r\geq r_{+}\left(  \tau\right)  .$ We can then estimate the exponential terms
in (\ref{A2E6})-(\ref{A2E8}) by $1.$ We now use the fact that $0\leq\bar
{D}\left(  \tau,r\right)  \leq Ce^{-ar}$ (cf. (\ref{A3E1})) to obtain that, if
$T_{0}$ is sufficiently large the functions $\lambda$ and $\bar{\mu}$ are
defined for arbitrary values $r\geq r_{+}\left(  \tau\right)  ,\ \tau\geq$
$T\geq T_{0}.$

In order to prove the continuity of the mappings (\ref{A2E4a}) we consider the
difference of the fields $\bar{\mu},\ \lambda$ associated to functions
$\left(  r_{1},w_{1},D_{1}\right)  $ and $\left(  r_{2},w_{2},D_{2}\right)  .$
Notice that the ODEs (\ref{A2E1}), (\ref{A2E2})\ imply estimates for the
derivatives of $\bar{\mu},\ \lambda.$ We also have uniform estimates for the
derivatives of the function $\bar{w}.$ Suppose that we consider the difference
of the functions $\lambda_{1},\ \lambda_{2}$ associated to $\left(
r_{1},w_{1},D_{1}\right)  $ and $\left(  r_{2},w_{2},D_{2}\right)  $
respectively. We then obtain, for each $\ \tau\geq$ $T$ the following estimate
(using (\ref{A2E5}), (\ref{A2E7}), (\ref{A2E8})):%
\begin{align*}
\left\vert \lambda_{1}\left(  \tau,r\right)  -\ \lambda_{2}\left(
\tau,r\right)  \right\vert  &  \leq C\int_{r_{+}\left(  \tau\right)  }%
^{r}\left\vert \lambda_{1}\left(  \tau,\xi\right)  -\ \lambda_{2}\left(
\tau,\xi\right)  \right\vert d\xi\\
&  +\left\vert \int_{r_{+}\left(  \tau\right)  }^{r}\Psi\left(  \tau
,\xi\right)  \left[  \bar{D}_{1}\left(  \tau,\xi\right)  -\bar{D}_{2}\left(
\tau,\xi\right)  \right]  d\xi\right\vert +\bar{\varepsilon}%
\end{align*}
where \ $\Psi$ is a continuous function with $\left\vert \Psi_{r}\right\vert $
bounded. The error $\bar{\varepsilon}$ is due to the differences of terms
$\bar{w}_{1},\bar{w}_{2}.$ It can be made arbitrarily small if $\left(
r_{1},w_{1},D_{1}\right)  $ and $\left(  r_{2},w_{2},D_{2}\right)  $ are close
in the topology of $\mathcal{X}_{L,T}$. Due to Arzela-Ascoli, for any
$\varepsilon>0$ and any $\tau\geq$ $T,$ there exists a finite set of functions
$\psi_{1},\ \psi_{2},...\psi_{L}$ such that for any $\Psi$ as above with
$\left\vert \Psi_{r}\right\vert \leq A,$ we have $\min_{k}\sup_{r\in\left[
0,R\right]  }\left\vert \Psi\left(  \tau,r\right)  -\ \psi_{k}\left(
\tau,r\right)  \right\vert \leq\varepsilon.$\ Using also the boundedness of
$\bar{D}_{1},\ \bar{D}_{2}$ in order to show that approximating the integral
$\int_{r_{+}\left(  \tau\right)  }^{r}\Psi\left(  \tau,\xi\right)  \left[
\bar{D}_{1}\left(  \tau,\xi\right)  -\bar{D}_{2}\left(  \tau,\xi\right)
\right]  d\xi$ for a finite set of values $r_{1},r_{2},...r_{M}$ we obtain an
approximation in the whole interval $\left[  0,R\right]  ,$ as well as the
fact that $\bar{D}_{1},$ $\bar{D}_{2}$ are close in the weak topology, we
obtain that the difference $\left\vert \lambda_{1}\left(  \tau,r\right)
-\ \lambda_{2}\left(  \tau,r\right)  \right\vert $ can be made arbitrarily
small for $r\in\left[  0,R\right]  .$ The difference $\left\vert \bar{\mu}%
_{1}\left(  \tau,r\right)  -\ \bar{\mu}_{2}\left(  \tau,r\right)  \right\vert
$ can be estimated in a similar form.

We notice also that we can prove that the functions $\lambda,\bar{\mu}$ are
continuous in the variable $\tau$ using a similar argument. Indeed, taking the
difference of $\lambda$ at two different times $\tau_{1},\tau_{2}\geq T$ we
would obtain:%
\begin{align*}
\left\vert \lambda\left(  \tau_{1},r\right)  -\ \lambda\left(  \tau
_{2},r\right)  \right\vert  &  \leq C\int_{r_{+}\left(  \tau_{1}\right)  }%
^{r}\left\vert \lambda\left(  \tau_{1},\xi\right)  -\ \lambda\left(  \tau
_{2},\xi\right)  \right\vert d\xi\\
&  +\left\vert \int_{r_{+}\left(  \tau_{2}\right)  }^{r}\Psi\left(  \tau
,\xi\right)  \left[  \bar{D}\left(  \tau_{1},\xi\right)  -\bar{D}\left(
\tau_{2},\xi\right)  \right]  d\xi\right\vert +\bar{\varepsilon}%
\end{align*}
where the error term contains contributions due to the differences of $\bar
{w}_{1},\bar{w}_{2}$ as well as the modulus of continuity of $r_{+}\left(
\tau\right)  .$ Using the fact that $\bar{D}\in C\left(  \left[
T,\infty\right)  ;\mathcal{A}_{\tau^{\ast}}\right)  $ we obtain that $\lambda$
is continuous in the variable $\tau.$ The continuity of $\bar{\mu}$ can be
proved in a similar manner.
\end{proof}

\section{FIXED POINT ARGUMENT.}

\subsection{Definition of the operator $\mathcal{T}$. \label{OperatorT}}

We now define an operator in the space $\mathcal{Y}_{L,T,a}$ as follows. Given
$\bar{D}\in\mathcal{Y}_{L,T,a}$ we define fields $\lambda,\ \bar{\mu}$ as in
Lemma \ref{CWFields}. We then define an operator $\mathcal{T}$ as:%
\begin{equation}
\mathcal{T}:\bar{D}\rightarrow\bar{D}_{n} \label{F8E5}%
\end{equation}
where the functions $\bar{D}_{n}$ are defined as follows ($n$ means new). We
first define the functions $r_{n},w_{n}$ as:
\begin{align}
\frac{\partial r_{n}\left(  \tau;\bar{\tau}\right)  }{\partial\tau}  &
=\frac{e^{\bar{\mu}\left(  \tau,r_{n}\left(  \tau;\bar{\tau}\right)  \right)
-2\lambda\left(  \tau,r_{n}\left(  \tau;\bar{\tau}\right)  \right)  }%
w_{n}\left(  \tau;\bar{\tau}\right)  r_{n}\left(  \tau;\bar{\tau}\right)
}{\Xi_{n}\left(  \tau;\bar{\tau}\right)  }\ ,\ \ r_{n}\left(  \bar{\tau}%
;\bar{\tau}\right)  =r_{+}\left(  \bar{\tau}\right) \label{F9E2}\\
\frac{\partial w_{n}\left(  \tau;\bar{\tau}\right)  }{\partial\tau}  &
=-\frac{e^{\bar{\mu}\left(  \tau,r_{n}\left(  \tau;\bar{\tau}\right)  \right)
}\left(  e^{2\lambda\left(  \tau,r_{n}\left(  \tau;\bar{\tau}\right)  \right)
}-1\right)  }{2\Xi_{n}\left(  \tau;\bar{\tau}\right)  }\left[  \frac
{1}{\left(  r_{n}\left(  \tau;\bar{\tau}\right)  \right)  ^{2}}+2e^{-2\lambda
\left(  \tau,r_{n}\left(  \tau;\bar{\tau}\right)  \right)  }\left(
w_{n}\left(  \tau;\bar{\tau}\right)  \right)  ^{2}\right] \nonumber\\
&  +\frac{e^{\bar{\mu}\left(  \tau,r_{n}\left(  \tau;\bar{\tau}\right)
\right)  }}{\left(  r_{n}\left(  \tau;\bar{\tau}\right)  \right)  ^{2}\Xi
_{n}\left(  \tau;\bar{\tau}\right)  }\label{F9E3}\\
w_{n}\left(  \bar{\tau};\bar{\tau}\right)   &  =\frac{e^{\lambda\left(
\tau,r_{n}\left(  \bar{\tau};\bar{\tau}\right)  \right)  }}{\left(  -\bar
{t}\right)  }V_{1}\left(  \frac{r_{+}\left(  \bar{\tau}\right)  }{\left(
-\bar{t}\right)  }\right)  \label{F9E4}%
\end{align}
where:%
\begin{equation}
\Xi_{n}\left(  \tau;\bar{\tau}\right)  =\sqrt{1+e^{-2\lambda\left(
r_{n}\left(  \tau;\bar{\tau}\right)  ,\tau\right)  }\left(  w_{n}\left(
\tau;\bar{\tau}\right)  \right)  ^{2}\left(  r_{n}\left(  \tau;\bar{\tau
}\right)  \right)  ^{2}} \label{F9rho}%
\end{equation}

We then define $D_{n}$ by means of:%
\begin{align}
&  \frac{\partial D_{n}\left(  \tau;\bar{\tau}\right)  }{\partial\tau
}\nonumber\\
&  =-\left(  \frac{e^{\bar{\mu}\left(  \tau,r_{n}\left(  \tau;\bar{\tau
}\right)  \right)  -2\lambda\left(  \tau,r_{n}\left(  \tau;\bar{\tau}\right)
\right)  }\cdot\left[  \frac{\left(  \frac{\partial w_{n}}{\partial\bar{\tau}%
}\left(  \tau,\bar{\tau}\right)  \right)  }{\left(  \frac{\partial r_{n}%
}{\partial\bar{\tau}}\left(  \tau,\bar{\tau}\right)  \right)  }-\lambda
_{r}\left(  r_{n}\left(  \tau;\bar{\tau}\right)  ,\tau\right)  w_{n}\left(
\tau;\bar{\tau}\right)  \right]  r_{n}\left(  \tau;\bar{\tau}\right)
}{\left(  \Xi_{n}\left(  \tau;\bar{\tau}\right)  \right)  ^{3}}\right.
\label{F7E4}\\
&  +\frac{e^{\bar{\mu}\left(  \tau,r_{n}\left(  \tau;\bar{\tau}\right)
\right)  -2\lambda\left(  \tau,r_{n}\left(  \tau;\bar{\tau}\right)  \right)
}r_{n}\left(  \tau;\bar{\tau}\right)  w_{n}\left(  \tau;\bar{\tau}\right)
\left[  \bar{\mu}_{r}\left(  \tau,r_{n}\left(  \tau;\bar{\tau}\right)
\right)  -\lambda_{r}\left(  \tau,r_{n}\left(  \tau;\bar{\tau}\right)
\right)  \right]  }{\Xi_{n}\left(  \tau;\bar{\tau}\right)  }\nonumber\\
&  \left.  +\frac{e^{\bar{\mu}\left(  \tau,r_{n}\left(  \tau;\bar{\tau
}\right)  \right)  -2\lambda\left(  \tau,r_{n}\left(  \tau;\bar{\tau}\right)
\right)  }w_{n}\left(  \tau;\bar{\tau}\right)  }{\Xi_{n}\left(  \tau;\bar
{\tau}\right)  }D_{n}\left(  \tau;\bar{\tau}\right)  \right) \nonumber\\
D_{n}\left(  \bar{\tau};\bar{\tau}\right)   &  =\left(  -\bar{t}\right)
b_{1}\left(  \frac{r_{+}\left(  \tau\right)  }{\left(  -\bar{t}\right)
}\right)  \exp\left(  \lambda\left(  \tau,r_{n}\left(  \tau;\bar{\tau}\right)
\right)  \right) \nonumber
\end{align}

Notice that in these equations we have replaced $\left(  \partial_{r}\bar
{w}_{n}\right)  \left(  r\left(  \tau,\bar{\tau}\right)  ,\tau\right)  $ by
$\left(  \frac{\partial w_{n}}{\partial\bar{\tau}}\left(  \tau,\bar{\tau
}\right)  \right)  /\left(  \frac{\partial r_{n}}{\partial\bar{\tau}}\left(
\tau,\bar{\tau}\right)  \right)  .$

We then construct a function $\bar{D}_{n}$ taking as starting point the
functions $\left(  r_{n},w_{n},D_{n}\right)  ,$ which will be shown to be in
the space $\mathcal{X}_{L,T},$ by means of Lemma \ref{auxFunctions}.

The main result that we will prove in the following is that the operator
$\mathcal{T}$ is well defined and it maps $\mathcal{Y}_{L,T,a}$ into itself
compactly if we assume that $L$ and $T$ are sufficiently large.

\begin{proposition}
\label{mTit}There exists $L_{0}$ sufficiently large such that, for $L>L_{0}$
there exists$\ T_{0}=T_{0}\left(  L\right)  $ such that for$\ T>T_{0}$ the
operator $\mathcal{T}$ defined as in (\ref{F8E5}) is well defined for
any\ $D\in\mathcal{Y}_{L,T,a}$ and it maps $\mathcal{Y}_{L,T,a}$ into itself.
The operator $\mathcal{T}$ is compact.
\end{proposition}

We first remark that the solvability of (\ref{F9E2})-(\ref{F9E4}),
(\ref{F7E4}) can be obtained using standard ODE arguments.

\begin{lemma}
For any $\bar{D}\in\mathcal{Y}_{L,T,a}$ there exist $\delta>0$ and a unique
solution of (\ref{F9E2})-(\ref{F9E4}), (\ref{F7E4}) defined in $\left\{
\bar{\tau}\geq\tau\geq\max\left\{  T,\bar{\tau}-\delta\right\}  \right\}  $
which satisfy these equations $a.e.$ with $\bar{\mu},\ \lambda$ as in Lemma
\ref{CWFields}. The following estimates hold:%
\begin{align*}
0  &  <c_{0}\leq\left\vert \frac{\partial r_{n}}{\partial\bar{\tau}}\left(
\tau,\bar{\tau}\right)  \right\vert \leq c_{1}\ \ ,\ \ \bar{\tau}\geq\tau
\geq\max\left\{  T,\bar{\tau}-\delta\right\} \\
\left\vert \frac{\partial w_{n}}{\partial\bar{\tau}}\left(  \tau,\bar{\tau
}\right)  \right\vert  &  \leq c_{1}\ \ ,\ \ \bar{\tau}\geq\tau\geq
\max\left\{  T,\bar{\tau}-\delta\right\}
\end{align*}
We assume that $\left(  \partial_{r}w_{n}\right)  \left(  r\left(  \tau
,\bar{\tau}\right)  ,\tau\right)  =\left(  \frac{\partial w_{n}}{\partial
\bar{\tau}}\left(  \tau,\bar{\tau}\right)  \right)  /\left(  \frac{\partial
r_{n}}{\partial\bar{\tau}}\left(  \tau,\bar{\tau}\right)  \right)  $. Moreover
$r_{n}\left(  \tau;\bar{\tau}\right)  \geq r_{+}\left(  \bar{\tau}\right)  $
if $\max\left\{  T,\bar{\tau}-\delta\right\}  \leq\tau\leq\bar{\tau}.$ The
solutions can be extended as long as $w_{n},\ r_{n}$ remain bounded and
$r_{n}\left(  \tau,\bar{\tau}\right)  >0,$ $\left\vert \frac{\partial r_{n}%
}{\partial\bar{\tau}}\left(  \tau,\bar{\tau}\right)  \right\vert >0.$
\end{lemma}

\begin{proof}
Due to (\ref{A2E1}), (\ref{A2E2}) as well as the boundedness of $\lambda$ we
have that $\lambda_{r},\ \bar{\mu}_{r}$ are uniformly bounded in $\left\{
\tau\geq T,\ r\geq r_{+}\left(  \tau\right)  \right\}  .$ Therefore, the
right-hand side of (\ref{F9E2}), (\ref{F9E3}) is Lipschitz for bounded values
of $r_{n},w_{n}.$ Then, the functions $r_{n},\ w_{n}$ can be defined by means
of (\ref{F9E2})-(\ref{F9E4}) due to local ODE theory. Notice that the global
boundedness of $r_{+}\left(  \bar{\tau}\right)  ,\ \frac{e^{\lambda\left(
r_{n}\left(  \bar{\tau};\bar{\tau}\right)  ,\tau\right)  }}{\left(  -\bar
{t}\right)  }V_{1}\left(  \frac{r_{+}\left(  \bar{\tau}\right)  }{\left(
-\bar{t}\right)  }\right)  $ for $\bar{\tau}\geq T$ imply that the time of
existence of solutions $\delta$ can be chosen uniformly in $\bar{\tau}.$ We
can now rewrite (\ref{F9E2})-(\ref{F9E4}) in integral form and differentiate
with respect to $\bar{\tau}.$ We can then obtain uniform estimates for
$\left\vert \frac{\partial r_{n}}{\partial\bar{\tau}}\left(  \tau,\bar{\tau
}\right)  \right\vert $ and $\left\vert \frac{\partial w_{n}}{\partial
\bar{\tau}}\left(  \tau,\bar{\tau}\right)  \right\vert $ using Gronwall's
Lemma. Moreover, differentiating the identity $r_{n}\left(  \bar{\tau}%
;\bar{\tau}\right)  =r_{+}\left(  \bar{\tau}\right)  ,$ and using the fact
that $\frac{dr_{+}\left(  \bar{\tau}\right)  }{d\bar{\tau}}$ is small for $T$
sufficiently large, and $\left\vert \frac{\partial r_{n}\left(  \bar{\tau
};\bar{\tau}\right)  }{\partial\tau}\right\vert \geq c_{2}>0$ for $\bar{\tau
}\geq T$ and $T$ large (due to (\ref{F9E2}) and (\ref{F9E4})) it then follows
that $c_{0}\leq\left\vert \frac{\partial r_{n}}{\partial\bar{\tau}}\left(
\tau,\bar{\tau}\right)  \right\vert $ if $\bar{\tau}\geq\tau\geq\max\left\{
T,\bar{\tau}-\delta\right\}  $ and $\delta$ is small.

We can then define $\left(  \partial_{r}w_{n}\right)  \left(  r\left(
\tau,\bar{\tau}\right)  ,\tau\right)  $ as $\left(  \frac{\partial w_{n}%
}{\partial\bar{\tau}}\left(  \tau,\bar{\tau}\right)  \right)  /\left(
\frac{\partial r_{n}}{\partial\bar{\tau}}\left(  \tau,\bar{\tau}\right)
\right)  $ for $\bar{\tau}\geq\tau\geq\max\left\{  T,\bar{\tau}-\delta
\right\}  .$ Then (\ref{F7E4}) has the form $\frac{\partial D_{n}\left(
\tau;\bar{\tau}\right)  }{\partial\tau}=F\left(  \tau;\bar{\tau}\right)
D_{n}\left(  \tau;\bar{\tau}\right)  $ where $F\left(  \tau;\bar{\tau}\right)
$ is bounded for $\bar{\tau}\geq\tau\geq\max\left\{  T,\bar{\tau}%
-\delta\right\}  $. We can then obtain $D_{n}\left(  \tau;\bar{\tau}\right)  $
as $D_{n}\left(  \tau;\bar{\tau}\right)  =D_{n}\left(  \bar{\tau};\bar{\tau
}\right)  \exp\left(  -\int_{\tau}^{\bar{\tau}}F\left(  s;\bar{\tau}\right)
ds\right)  .$ Notice that the function $D_{n}\left(  \tau;\bar{\tau}\right)  $
solves (\ref{F7E4}) $a.e.$ We remark also that the right-hand side of
(\ref{F9E2}) is negative if $r_{n}>0$ and $w_{n}<0.$ Due to (\ref{F9E4}) we
have $w_{n}<0$ if $\bar{\tau}\geq\tau\geq\max\left\{  T,\bar{\tau}%
-\delta\right\}  $ and $\delta>0$ is small enough. It then follows that
$r_{n}\left(  \tau;\bar{\tau}\right)  \geq r_{+}\left(  \bar{\tau}\right)  $
if $\max\left\{  T,\bar{\tau}-\delta\right\}  \leq\tau\leq\bar{\tau}$ and
$\delta>0$ is small. The fact that the solution can be extended as long as
$w_{n},\ r_{n}$ are bounded and $r_{n}\left(  \tau,\bar{\tau}\right)  >0,$
$\left\vert \frac{\partial r_{n}}{\partial\bar{\tau}}\left(  \tau,\bar{\tau
}\right)  \right\vert >0$ follows from the form of the right-hand side of
(\ref{F9E2}), (\ref{F9E3}), (\ref{F7E4}).
\end{proof}

\subsection{The operator $\mathcal{T}$ maps $\mathcal{Y}_{L,T,a}$ into itself
if $L$ and $T$ are large.}

In order to prove Proposition \ref{mTit} we need to derive estimates for the
functions $\rho_{n}^{\ast}\left(  \tau;\bar{\tau}\right)  ,$ $z_{n}\left(
\tau;\bar{\tau}\right)  $ which will be defined by means of:%
\begin{equation}
\rho_{n}^{\ast}\left(  \tau;\bar{\tau}\right)  =\left(  r_{n}\left(  \tau
;\bar{\tau}\right)  -\mathcal{R}\left(  \tau-\bar{\tau}\right)  \right)
\ \ ,\ \ z_{n}\left(  \tau;\bar{\tau}\right)  =\left(  w_{n}\left(  \tau
;\bar{\tau}\right)  -\mathcal{W}\left(  \tau-\bar{\tau}\right)  \right)
\ \ ,\ \ \bar{\tau}\geq\tau\geq T \label{G9E3}%
\end{equation}
(cf. (\ref{F1E1})). As a first step we obtain estimates for the difference
between the functions $\lambda\left(  r_{n}\left(  \tau;\bar{\tau}\right)
,\tau\right)  ,$ $\bar{\mu}\left(  r_{n}\left(  \tau;\bar{\tau}\right)
,\tau\right)  $ and their asymptotic values $\lambda_{0}\left(  \mathcal{R}%
\left(  \tau-\bar{\tau}\right)  \right)  ,$ $\bar{\mu}_{0}\left(
\mathcal{R}\left(  \tau-\bar{\tau}\right)  \right)  .$

\begin{notation}
\label{Not1}We will write $f\left(  \tau;\bar{\tau}\right)  =O\left(  g\left(
\tau;\bar{\tau}\right)  \right)  $ to indicate that $\left\vert f\left(
\tau;\bar{\tau}\right)  \right\vert \leq Cg\left(  \tau;\bar{\tau}\right)  $
for $\tau\geq T,$ where $C$ is a constant independent of $T,\ \tau,\ \bar
{\tau},\ L,\ r$ but perhaps depending on $y_{0},\ R_{\max}.$
\end{notation}

The function $g\left(  \tau\right)  $ used in the Notation \ref{Not1} will be
usually $\frac{1}{1+\left(  \tau-\bar{\tau}\right)  ^{m}},$ $e^{-a\tau}$ for
some $m>0,\ a>0.$

\begin{lemma}
\label{DifFields}Let $\lambda_{0},\ \bar{\mu}_{0}$ be given by (\ref{F4E5}),
(\ref{F4E6}). Suppose that $\bar{D}\in\mathcal{Y}_{L,T,a}$ and define
$\lambda,\ \bar{\mu},$ $R_{0}$ as in (\ref{A2E5}), (\ref{A2E8}). Let
$r_{n}\left(  \tau;\bar{\tau}\right)  ,\ w_{n}\left(  \tau;\bar{\tau}\right)
$ defined by means of (\ref{F9E2})-(\ref{F9E4}) and $\rho_{n}^{\ast},z_{n}%
\ $be given by (\ref{G9E3}). Then:%
\begin{align}
&  \lambda\left(  \tau,r_{n}\left(  \tau;\bar{\tau}\right)  \right)
-\lambda_{0}\left(  \mathcal{R}\left(  \tau-\bar{\tau}\right)  \right)
\label{F7E8}\\
&  =\frac{1}{2}\frac{\rho_{n}^{\ast}\left(  \tau;\bar{\tau}\right)
}{\mathcal{R}\left(  \tau-\bar{\tau}\right)  }-\frac{1}{2}\frac{\rho_{n}%
^{\ast}\left(  \tau;\bar{\tau}\right)  -\left(  R_{0}\left(  \tau,r\right)
-\frac{2R_{\max}}{3}\right)  }{\mathcal{R}\left(  \tau-\bar{\tau}\right)
-\frac{2R_{\max}}{3}}+R_{1}\left(  \tau;\bar{\tau}\right) \nonumber
\end{align}
where:%
\begin{equation}
\left\vert R_{1}\left(  \tau;\bar{\tau}\right)  \right\vert \leq C\left(
\left(  \frac{\rho_{n}^{\ast}\left(  \tau;\bar{\tau}\right)  }{\mathcal{R}%
\left(  \tau-\bar{\tau}\right)  }\right)  ^{2}+\left(  \frac{\rho_{n}^{\ast
}\left(  \tau;\bar{\tau}\right)  -\left(  R_{0}\left(  \tau,r\right)
-\frac{2R_{\max}}{3}\right)  }{\mathcal{R}\left(  \tau-\bar{\tau}\right)
-\frac{2R_{\max}}{3}}\right)  ^{2}\right)  \ \label{F7E8a}%
\end{equation}
with $C$ depending on $y_{0}$ but independent of $L.$ We have also:%
\begin{equation}
\bar{\mu}\left(  \tau,r_{n}\left(  \tau;\bar{\tau}\right)  \right)  -\bar{\mu
}_{0}\left(  \mathcal{R}\left(  \tau-\bar{\tau}\right)  \right)  =\frac{1}%
{2}\frac{\rho_{n}^{\ast}\left(  \tau;\bar{\tau}\right)  }{\mathcal{R}\left(
\tau-\bar{\tau}\right)  -\frac{2R_{\max}}{3}}-\frac{1}{2}\frac{\rho_{n}^{\ast
}\left(  \tau;\bar{\tau}\right)  }{\mathcal{R}\left(  \tau-\bar{\tau}\right)
}+R_{2}\left(  \tau;\bar{\tau}\right)  \ \label{F7E9}%
\end{equation}
with%
\[
\left\vert R_{2}\left(  \tau;\bar{\tau}\right)  \right\vert \leq C\left[
\left(  \frac{\rho_{n}^{\ast}\left(  \tau;\bar{\tau}\right)  }{\mathcal{R}%
\left(  \tau-\bar{\tau}\right)  }\right)  ^{2}+e^{-\delta\tau}\right]
\]
where $\mathcal{R}$ is as in Subsection \ref{auxF} and $a>0$ is as in Lemma
\ref{auxFunctions} (cf. also (\ref{A3E1})). The constant $C$ depends on
$y_{0}$ but is independent of $L.$

Moreover, the following estimate holds:%
\begin{equation}
0\leq\lambda\left(  \tau,r_{n}\left(  \tau;\bar{\tau}\right)  \right)
\leq\frac{C}{r_{n}\left(  \tau;\bar{\tau}\right)  }\ \ ,\ \ T\leq\tau\leq
\bar{\tau} \label{F7E9a}%
\end{equation}

\end{lemma}

\begin{proof}
Using Lemmas \ref{auxFunctions}, \ref{CWFields} we obtain:%
\begin{equation}
\int_{r_{+}\left(  \tau\right)  }^{r}8\pi^{2}\tilde{E}\left(  \tau,\xi\right)
\exp\left(  -\lambda\left(  \tau,\xi\right)  \right)  \bar{D}\left(  \tau
,\xi\right)  d\xi\leq CLe^{-2\tau}\ \ ,\ \ \tau\geq T\ \label{B3E1}%
\end{equation}
where $C$ does not depend on $L.$ Using (\ref{A2E7}), (\ref{F2E2}) in Theorem
\ref{RV}, Proposition \ref{Rmas} and (\ref{B3E1}) we obtain:%
\begin{equation}
\left\vert R_{0}\left(  \tau,r\right)  -\frac{2R_{\max}}{3}\right\vert \leq
Ce^{-\delta\tau}+CLe^{-2\tau}\ \ ,\ \ \tau\geq T\ \label{G1E2}%
\end{equation}
where $\delta>0$ is as in Theorem \ref{RV}. Therefore the difference
$\left\vert R_{0}\left(  \tau,r\right)  -\frac{2R_{\max}}{3}\right\vert $ is
small if $T$ is large enough. Using (\ref{A2E5}) we obtain:%
\begin{align*}
&  \lambda\left(  \tau,r_{n}\left(  \tau;\bar{\tau}\right)  \right)
-\lambda_{0}\left(  \mathcal{R}\left(  \tau-\bar{\tau}\right)  \right) \\
&  =\frac{1}{2}\log\left(  \frac{r_{n}\left(  \tau;\bar{\tau}\right)
}{\mathcal{R}\left(  \tau-\bar{\tau}\right)  }\right)  +\frac{1}{2}\log\left(
\frac{\mathcal{R}\left(  \tau-\bar{\tau}\right)  -\frac{2R_{\max}}{3}}%
{r_{n}\left(  \tau;\bar{\tau}\right)  -R_{0}\left(  \tau,r\right)  }\right) \\
&  =\frac{1}{2}\log\left(  1+\frac{\rho^{\ast}\left(  \tau;\bar{\tau}\right)
}{\mathcal{R}\left(  \tau-\bar{\tau}\right)  }\right)  -\frac{1}{2}\log\left(
1+\frac{\rho^{\ast}\left(  \tau;\bar{\tau}\right)  -\left(  R_{0}\left(
\tau,r\right)  -\frac{2R_{\max}}{3}\right)  }{\mathcal{R}\left(  \tau
-\bar{\tau}\right)  -\frac{2R_{\max}}{3}}\right)
\end{align*}
and (\ref{F7E8}), (\ref{F7E8a}) follow using Taylor's Theorem.

We now compute, using (\ref{A2E6}) and (\ref{F4E6}), the difference:%
\begin{align*}
&  \bar{\mu}\left(  \tau,r\left(  \tau;\bar{\tau}\right)  \right)  -\bar{\mu
}_{0}\left(  \mathcal{R}\left(  \tau-\bar{\tau}\right)  \right)  \\
&  =\bar{\mu}\left(  \tau,r_{+}\left(  \tau\right)  \right)  +\int
_{r_{+}\left(  \tau\right)  }^{r\left(  \tau;\bar{\tau}\right)  }\frac
{4\pi^{2}\exp\left(  -\lambda\left(  \tau,\xi\right)  \right)  \bar{w}%
^{2}\left(  \tau,\xi\right)  \bar{D}\left(  \tau,\xi\right)  }{\tilde
{E}\left(  \tau,\xi\right)  }\frac{d\xi}{\xi-R_{0}}\\
&  +\frac{1}{2}\int_{r_{+}\left(  \tau\right)  }^{r\left(  \tau;\bar{\tau
}\right)  }\frac{d\xi}{\left(  \xi-R_{0}\left(  \tau,\xi\right)  \right)
}-\frac{1}{2}\log\left(  \frac{r\left(  \tau;\bar{\tau}\right)  }{r_{+}\left(
\tau\right)  }\right)  \\
&  -\bar{\mu}_{0}\left(  \infty\right)  -\frac{1}{2}\log\left(  1-\frac
{2R_{\max}}{3\mathcal{R}\left(  \tau-\bar{\tau}\right)  }\right)
\end{align*}
where $\bar{\mu}_{0}\left(  \infty\right)  =\log\left(  \frac{R_{\max}%
\sqrt{3\left(  1-y_{0}^{2}\right)  }}{y_{0}}\right)  .$ Then, adding and
subtracting in the right-hand side the term $\frac{1}{2}\int_{r_{+}\left(
\tau\right)  }^{r\left(  \tau;\bar{\tau}\right)  }\frac{d\xi}{\left(
\xi-\frac{2R_{\max}}{3}\right)  }$ we obtain, after some computations:%
\begin{align*}
&  \bar{\mu}\left(  \tau,r\left(  \tau;\bar{\tau}\right)  \right)  -\bar{\mu
}_{0}\left(  \mathcal{R}\left(  \tau-\bar{\tau}\right)  \right)  \\
&  =\bar{\mu}\left(  \tau,r_{+}\left(  \tau\right)  \right)  +\int
_{r_{+}\left(  \tau\right)  }^{r\left(  \tau;\bar{\tau}\right)  }\frac
{4\pi^{2}\exp\left(  -\lambda\left(  \tau,\xi\right)  \right)  \bar{w}%
^{2}\left(  \tau,\xi\right)  \bar{D}\left(  \tau,\xi\right)  }{\tilde
{E}\left(  \tau,\xi\right)  }\frac{d\xi}{\xi-R_{0}}\\
&  +\frac{1}{2}\int_{r_{+}\left(  \tau\right)  }^{r\left(  \tau;\bar{\tau
}\right)  }\frac{\left(  R_{0}\left(  \tau,\xi\right)  -\frac{2R_{\max}}%
{3}\right)  d\xi}{\left(  \xi-\frac{2R_{\max}}{3}-\left(  R_{0}\left(
\tau,\xi\right)  -\frac{2R_{\max}}{3}\right)  \right)  \left(  \xi
-\frac{2R_{\max}}{3}\right)  }\\
&  +\frac{1}{2}\log\left(  \frac{r\left(  \tau;\bar{\tau}\right)
-\frac{2R_{\max}}{3}}{\mathcal{R}\left(  \tau-\bar{\tau}\right)
-\frac{2R_{\max}}{3}}\right)  +\frac{1}{2}\log\left(  \frac{r_{+}\left(
\tau\right)  }{r_{+}\left(  \tau\right)  -\frac{2R_{\max}}{3}}\right)  \\
&  -\frac{1}{2}\log\left(  \frac{r\left(  \tau;\bar{\tau}\right)
}{\mathcal{R}\left(  \tau-\bar{\tau}\right)  }\right)  -\bar{\mu}_{0}\left(
\infty\right)
\end{align*}

Using (\ref{F3E4}) and (\ref{F2E1}), Proposition \ref{Rmas} that implies
$\hat{U}\left(  \frac{r_{+}\left(  t\right)  }{\left(  -t\right)  }\right)
=\log\left(  \frac{\sqrt{1-y_{0}^{2}}}{y_{0}}\right)  +O\left(  \left(
-t\right)  ^{\delta}\right)  $ as well as the identity $\log\left(
\frac{R_{\max}\sqrt{1-y_{0}^{2}}}{y_{0}}\right)  -\bar{\mu}_{0}\left(
\infty\right)  =-\frac{\log\left(  3\right)  }{2}$ we obtain:%
\begin{align}
&  \bar{\mu}\left(  \tau,r_{n}\left(  \tau;\bar{\tau}\right)  \right)
-\bar{\mu}_{0}\left(  \mathcal{R}\left(  \tau-\bar{\tau}\right)  \right)
\label{G1E1}\\
&  =\log\left(  \frac{r_{+}\left(  t\right)  }{R_{\max}}\right)  +\int
_{r_{+}\left(  \tau\right)  }^{r\left(  \tau;\bar{\tau}\right)  }\frac
{4\pi^{2}\exp\left(  -\lambda\left(  \tau,\xi\right)  \right)  \bar{w}%
^{2}\left(  \tau,\xi\right)  \bar{D}\left(  \tau,\xi\right)  }{\tilde{E}%
_{1}\left(  \tau,\xi\right)  }\frac{d\xi}{\xi-R_{0}}\nonumber\\
&  +\frac{1}{2}\int_{r_{+}\left(  \tau\right)  }^{r\left(  \tau;\bar{\tau
}\right)  }\frac{\left(  R_{0}\left(  \tau,\xi\right)  -\frac{2R_{\max}}%
{3}\right)  d\xi}{\left(  \xi-\frac{2R_{\max}}{3}-\left(  R_{0}\left(
\tau,\xi\right)  -\frac{2R_{\max}}{3}\right)  \right)  \left(  \xi
-\frac{2R_{\max}}{3}\right)  }\nonumber\\
&  +\frac{1}{2}\log\left(  \frac{r_{n}\left(  \tau;\bar{\tau}\right)
-\frac{2R_{\max}}{3}}{\mathcal{R}\left(  \tau-\bar{\tau}\right)
-\frac{2R_{\max}}{3}}\right)  +\frac{1}{2}\log\left(  \frac{r_{+}\left(
\tau\right)  }{3\left(  r_{+}\left(  \tau\right)  -\frac{2R_{\max}}{3}\right)
}\right)  \nonumber\\
&  -\frac{1}{2}\log\left(  \frac{r_{n}\left(  \tau;\bar{\tau}\right)
}{\mathcal{R}\left(  \tau-\bar{\tau}\right)  }\right)  +O\left(
e^{-\delta\tau}\right)
\end{align}

Using (\ref{G1E2}) we obtain:%
\begin{align}
&  \left\vert \frac{1}{2}\int_{r_{+}\left(  \tau\right)  }^{r\left(  \tau
;\bar{\tau}\right)  }\frac{\left(  R_{0}\left(  \tau,\xi\right)
-\frac{2R_{\max}}{3}\right)  d\xi}{\left(  \xi-\frac{2R_{\max}}{3}-\left(
R_{0}\left(  \tau,\xi\right)  -\frac{2R_{\max}}{3}\right)  \right)  \left(
\xi-\frac{2R_{\max}}{3}\right)  }\right\vert \nonumber\\
&  \leq C\left[  e^{-\delta\tau}+Le^{-2\tau}\right]  \ \label{G1E3}%
\end{align}

On the other hand, using (\ref{Y6E3d}), (\ref{A1E2}), (\ref{F1E1}), and Lemmas
\ref{auxFunctions}, \ref{CWFields} we obtain:%
\begin{equation}
\left\vert \int_{r_{+}\left(  \tau\right)  }^{r\left(  \tau;\bar{\tau}\right)
}\frac{4\pi^{2}\exp\left(  -\lambda\left(  \tau,\xi\right)  \right)  \bar
{w}^{2}\left(  \tau,\xi\right)  \bar{D}\left(  \tau,\xi\right)  }{\tilde
{E}_{1}\left(  \tau,\xi\right)  }\frac{d\xi}{\xi-R_{0}}\right\vert \leq
CLe^{-2\tau} \label{G1E4}%
\end{equation}

Using Proposition \ref{Rmas} and Taylor's Theorem we arrive at:%
\begin{equation}
\frac{1}{2}\log\left(  \frac{r_{+}\left(  \tau\right)  }{3\left(  r_{+}\left(
\tau\right)  -\frac{2R_{\max}}{3}\right)  }\frac{r_{+}\left(  \tau\right)
}{R_{\max}}\right)  \leq Ce^{-4\tau} \label{G1E5}%
\end{equation}

Plugging (\ref{G1E3})-(\ref{G1E5}) into (\ref{G1E1}) and using also
(\ref{G9E3}) we obtain:%
\begin{align*}
\bar{\mu}\left(  \tau,r\left(  \tau;\bar{\tau}\right)  \right)  -\bar{\mu}%
_{0}\left(  \mathcal{R}\left(  \tau-\bar{\tau}\right)  \right)    & =\frac
{1}{2}\log\left(  1+\frac{\rho^{\ast}\left(  \tau;\bar{\tau}\right)
}{\mathcal{R}\left(  \tau-\bar{\tau}\right)  -\frac{2R_{\max}}{3}}\right)  \\
& -\frac{1}{2}\log\left(  1+\frac{\rho^{\ast}\left(  \tau;\bar{\tau}\right)
}{\mathcal{R}\left(  \tau-\bar{\tau}\right)  }\right)  +O\left(
e^{-\delta\tau}\right)
\end{align*}
and using Taylor's Theorem we arrive at (\ref{F7E9}).

Using (\ref{A2E5}), (\ref{G1E2}) as well as Taylor's Theorem we obtain
(\ref{F7E9a}) due to the fact that $r_{+}\left(  \tau\right)  \rightarrow
R_{\max}$ as $\tau\rightarrow\infty.$
\end{proof}

We now write for further reference the differential equations satisfied by
$\rho_{n}^{\ast}\left(  \tau;\bar{\tau}\right)  $ and $z_{n}\left(  \tau
;\bar{\tau}\right)  $ defined in (\ref{G9E3}).

\begin{lemma}
Suppose that $r_{n}\left(  \tau;\bar{\tau}\right)  $ and $w_{n}\left(
\tau;\bar{\tau}\right)  $ solve (\ref{F9E2})-(\ref{F9E4}). Then, the functions
$\rho_{n}^{\ast}\left(  \tau;\bar{\tau}\right)  ,\ z_{n}\left(  \tau;\bar
{\tau}\right)  $ defined in (\ref{G9E3}) solve the following system of
equations for any $\left(  \tau;\bar{\tau}\right)  \in\mathcal{U}\left(
T\right)  :$%
\begin{align}
\frac{\partial\rho_{n}^{\ast}\left(  \tau;\bar{\tau}\right)  }{\partial\tau}
&  =\frac{e^{\bar{\mu}\left(  \tau,r_{n}\left(  \tau;\bar{\tau}\right)
\right)  -2\lambda\left(  \tau,r_{n}\left(  \tau;\bar{\tau}\right)  \right)
}w_{n}\left(  \tau;\bar{\tau}\right)  r_{n}\left(  \tau;\bar{\tau}\right)
}{\sqrt{1+e^{-2\lambda\left(  \tau,r_{n}\left(  \tau;\bar{\tau}\right)
\right)  }\left(  w_{n}\left(  \tau;\bar{\tau}\right)  \right)  ^{2}\left(
r_{n}\left(  \tau;\bar{\tau}\right)  \right)  ^{2}}}\nonumber\\
&  -\frac{e^{\bar{\mu}_{0}\left(  \mathcal{R}\left(  \tau\right)  \right)
-2\lambda_{0}\left(  \mathcal{R}\left(  \tau\right)  \right)  }\mathcal{W}%
\left(  \tau-\bar{\tau}\right)  \mathcal{R}\left(  \tau-\bar{\tau}\right)
}{\sqrt{1+e^{-2\lambda_{0}\left(  \mathcal{R}\left(  \tau\right)  \right)
}\left(  \mathcal{R}\left(  \tau-\bar{\tau}\right)  \right)  ^{2}\left(
\mathcal{W}\left(  \tau-\bar{\tau}\right)  \right)  ^{2}}}\ \nonumber\\
\rho_{n}^{\ast}\left(  \bar{\tau};\bar{\tau}\right)   &  =r_{+}\left(
\tau\right)  -R_{\max}\nonumber\\
\frac{\partial z_{n}\left(  \tau;\bar{\tau}\right)  }{\partial\tau}  &
=-\frac{e^{\bar{\mu}\left(  \tau,r_{n}\left(  \tau;\bar{\tau}\right)  \right)
}\left(  e^{2\lambda\left(  \tau,r_{n}\left(  \tau;\bar{\tau}\right)  \right)
}-1\right)  }{2\sqrt{1+e^{-2\lambda\left(  \tau,r_{n}\left(  \tau;\bar{\tau
}\right)  \right)  }\left(  r_{n}\left(  \tau;\bar{\tau}\right)  \right)
^{2}\left(  w_{n}\left(  \tau;\bar{\tau}\right)  \right)  ^{2}}}A\left[
r_{n}\left(  \tau;\bar{\tau}\right)  ,w_{n}\left(  \tau;\bar{\tau}\right)
\right] \nonumber\\
&  +\frac{e^{\bar{\mu}_{0}\left(  \mathcal{R}\left(  \tau\right)  \right)
}\left(  e^{2\lambda_{0}\left(  \mathcal{R}\left(  \tau-\bar{\tau}\right)
\right)  }-1\right)  }{2\sqrt{1+e^{-2\lambda_{0}\left(  \mathcal{R}\left(
\tau\right)  \right)  }\left(  \mathcal{R}\left(  \tau-\bar{\tau}\right)
\right)  ^{2}\left(  \mathcal{W}\left(  \tau-\bar{\tau}\right)  \right)  ^{2}%
}}A\left[  \mathcal{R}\left(  \tau-\bar{\tau}\right)  ,\mathcal{W}\left(
\tau-\bar{\tau}\right)  \right] \nonumber\\
&  +\frac{e^{\bar{\mu}\left(  \tau,r_{n}\left(  \tau;\bar{\tau}\right)
\right)  }}{\left(  r_{n}\left(  \tau;\bar{\tau}\right)  \right)  ^{2}%
\sqrt{1+e^{-2\lambda\left(  \tau,r_{n}\left(  \tau;\bar{\tau}\right)  \right)
}\left(  r_{n}\left(  \tau;\bar{\tau}\right)  \right)  ^{2}\left(
w_{n}\left(  \tau;\bar{\tau}\right)  \right)  ^{2}}}\nonumber\\
&  -\frac{e^{\bar{\mu}_{0}\left(  \mathcal{R}\left(  \tau-\bar{\tau}\right)
\right)  }}{\left(  \mathcal{R}\left(  \tau-\bar{\tau}\right)  \right)
^{2}\sqrt{1+e^{-2\lambda_{0}\left(  \mathcal{R}\left(  \tau-\bar{\tau}\right)
\right)  }\left(  \mathcal{R}\left(  \tau-\bar{\tau}\right)  \right)
^{2}\left(  \mathcal{W}\left(  \tau-\bar{\tau}\right)  \right)  ^{2}}%
}\label{F7E5}\\
z_{n}\left(  \tau;\bar{\tau}\right)   &  =\frac{e^{\lambda\left(
\tau,r\left(  \tau;\bar{\tau}\right)  \right)  }}{\left(  -\bar{t}\right)
}V_{1}\left(  \frac{r_{+}\left(  \tau\right)  }{\left(  -\bar{t}\right)
}\right)  +\frac{6y_{0}\sqrt{\left(  1-y_{0}^{2}\right)  }}{\left(
1-4y_{0}^{2}\right)  R_{\max}}\nonumber
\end{align}
where we define:%
\[
A\left[  r,w\right]  =\left[  \frac{1}{\left(  r\left(  \tau;\bar{\tau
}\right)  \right)  ^{2}}+2e^{-2\lambda\left(  r\left(  \tau;\bar{\tau}\right)
,\tau\right)  }\left(  w\left(  \tau;\bar{\tau}\right)  \right)  ^{2}\right]
\]

\end{lemma}

\begin{proof}
The Lemma follows just subtracting (\ref{F5E5}), (\ref{F5E6}), (\ref{F5E7})
from (\ref{F9E2}), (\ref{F9E3}), (\ref{F9E4}) respectively.
\end{proof}

\begin{remark}
These equations are satisfied at any point of $\mathcal{U}\left(  T\right)  ,$
not only $a.e.$
\end{remark}

\begin{lemma}
\label{linDiff}Given $L>0,$ there exists $T_{0}=T_{0}\left(  L\right)  >0$
with the following property. Suppose that $T\geq T_{0}.$ Let $\bar{D}%
\in\mathcal{Y}_{L,T,a}.$ Let $\left(  r_{n},w_{n},D_{n}\right)  $ be as in
(\ref{F8E5}) (cf. also (\ref{F9E2})-(\ref{F9E4}), (\ref{F7E4})). Then, the
following estimate holds:%
\[
\sup_{\left\{  T\leq\tau\leq\bar{\tau}\right\}  }L\left(  \frac{\left\vert
\rho_{n}^{\ast}\left(  \tau;\bar{\tau}\right)  \right\vert }{1+\left(
\bar{\tau}-\tau\right)  }+\left\vert z_{n}\left(  \tau;\bar{\tau}\right)
\right\vert \right)  \leq\frac{1}{5}%
\]

\end{lemma}

\begin{proof}
We first remark that, as long as we have the inequality%
\begin{equation}
w_{n}\left(  \tau;\bar{\tau}\right)  \leq-\frac{\Gamma\left(  y_{0,}R_{\max
}\right)  }{2}\ \label{B4E1}%
\end{equation}
with $\Gamma\left(  y_{0,}R_{\max}\right)  $ as in Proposition
\ref{appFunctions}, the following estimate holds
\begin{equation}
\left\vert \frac{\partial w_{n}\left(  \tau;\bar{\tau}\right)  }{\partial\tau
}\right\vert \leq\frac{C}{\left(  r_{n}\left(  \tau;\bar{\tau}\right)
\right)  ^{2}} \label{B4E2}%
\end{equation}

Indeed, due to (\ref{B4E1}) we have $\sqrt{1+e^{-2\lambda\left(  r_{n}\left(
\tau;\bar{\tau}\right)  ,\tau\right)  }\left(  r_{n}\left(  \tau;\bar{\tau
}\right)  \right)  ^{2}\left(  w_{n}\left(  \tau;\bar{\tau}\right)  \right)
^{2}}\geq Cr_{n}\left(  \tau;\bar{\tau}\right)  .$ Combining this estimate
also with the fact that (\ref{F7E9a}) implies $\left(  e^{2\lambda\left(
r_{n}\left(  \tau;\bar{\tau}\right)  \right)  }-1\right)  \leq\frac{C}%
{r_{n}\left(  \tau;\bar{\tau}\right)  }$ we obtain (\ref{B4E2}).

Using (\ref{F7E5}) as well as Lemma \ref{DifFields} we can derive a system of
equations which depends, to the leading order, linearly on $\left(  \rho
_{n}^{\ast},z_{n}\right)  $ and contains source terms proportional to
$e^{-\delta\tau}.$ Then $z_{n},\ \rho_{n}^{\ast}$ can be made arbitrarily
small in any bounded region if $T$ is large enough. In particular this implies
that on such a time interval time the inequality (\ref{B4E1}) holds. Moreover,
due to Proposition \ref{appFunctions} we can obtain that for any $R>0$ there
exists $\hat{T}$ such that, for $\left(  \bar{\tau}-\tau\right)  \geq L$ we
have $r_{n}\left(  \tau;\bar{\tau}\right)  \geq R$ if $T_{0}$ is sufficiently
large. Choosing $R$ large enough, as well as estimate (\ref{B1E2a}) it would
then follow that the change of $w_{n}\left(  \tau;\bar{\tau}\right)  $ can be
made arbitrarily small in the whole set of values $T\leq\tau\leq\bar{\tau}.$ A
similar estimate can be proved for $\mathcal{W}\left(  \tau\right)  .$\ The
estimate (\ref{B4E1}) would be proved by means of a continuation argument to
that set of values. It then follows that%
\begin{equation}
\left\vert z_{n}\left(  \tau;\bar{\tau}\right)  \right\vert \leq\frac{1}{10L}
\label{B4E3}%
\end{equation}
if $T_{0}$ is sufficiently large.

In order to obtain the estimate of $\rho_{n}^{\ast}$ we use the first equation
of (\ref{F7E5}). The differences of functions containing the fields
$\lambda,\ \bar{\mu}$ or their limit values $\lambda_{0},\ \bar{\mu}_{0}$ are
small if $T_{0}$ is large. Actually these differences could contain terms like
$\left(  \frac{\rho_{n}^{\ast}}{\mathcal{R}}\right)  ^{2}$ that are smaller
than the expected contribution of $\frac{\rho_{n}^{\ast}}{1+\left(  \bar{\tau
}-\tau\right)  }.$ Notice that in terms like the ones coming from
(\ref{F7E8}), (\ref{F7E9}) we obtain some contributions with the form
$C\frac{\rho_{n}^{\ast}}{\mathcal{R}^{2}}.$ The contribution due to these
terms can be estimated using Gronwall arguments, and due to the integrability
of $\frac{1}{\mathcal{R}^{2}}$ the corresponding effect in $\rho_{n}^{\ast}$
would be small. Some terms that must be estimated carefully are the
differences of the form:%
\[
\frac{r_{n}\left(  \tau;\bar{\tau}\right)  }{\sqrt{1+e^{-2\lambda_{0}\left(
\mathcal{R}\left(  \tau\right)  \right)  }\left(  r_{n}\left(  \tau;\bar{\tau
}\right)  \right)  ^{2}\left(  \mathcal{W}\left(  \tau-\bar{\tau}\right)
\right)  ^{2}}}-\frac{\mathcal{R}\left(  \tau-\bar{\tau}\right)  }%
{\sqrt{1+e^{-2\lambda_{0}\left(  \mathcal{R}\left(  \tau\right)  \right)
}\left(  \mathcal{R}\left(  \tau-\bar{\tau}\right)  \right)  ^{2}\left(
\mathcal{W}\left(  \tau-\bar{\tau}\right)  \right)  ^{2}}}%
\]

The differences or the other terms can be estimated easily. These differences
give terms of the form:%
\[
\frac{r_{n}\left(  \tau;\bar{\tau}\right)  J_{1}-\mathcal{R}\left(  \tau
-\bar{\tau}\right)  J_{2}}{J_{1}J_{2}}%
\]
where:%
\begin{align*}
J_{1}  &  =\sqrt{1+e^{-2\lambda_{0}\left(  \mathcal{R}\left(  \tau\right)
\right)  }\left(  \mathcal{R}\left(  \tau-\bar{\tau}\right)  \right)
^{2}\left(  \mathcal{W}\left(  \tau-\bar{\tau}\right)  \right)  ^{2}}\\
J_{2}  &  =\sqrt{1+e^{-2\lambda_{0}\left(  \mathcal{R}\left(  \tau\right)
\right)  }\left(  r_{n}\left(  \tau;\bar{\tau}\right)  \right)  ^{2}\left(
\mathcal{W}\left(  \tau-\bar{\tau}\right)  \right)  ^{2}}%
\end{align*}

Taking conjugates (of the roots) we obtain differences with orders of
magnitude:%
\[
\frac{\rho_{n}^{\ast}\left(  \tau;\bar{\tau}\right)  }{\left(  \mathcal{R}%
\left(  \tau-\bar{\tau}\right)  \right)  ^{3}}%
\]

Notice that terms like $\left(  r_{n}\left(  \tau;\bar{\tau}\right)  \right)
^{2}e^{-2\lambda_{0}\left(  \mathcal{R}\left(  \tau\right)  \right)  }\left(
\mathcal{R}\left(  \tau-\bar{\tau}\right)  \right)  ^{2}\left(  \mathcal{W}%
\left(  \tau-\bar{\tau}\right)  \right)  ^{2}$ and $\left(  \mathcal{R}\left(
\tau-\bar{\tau}\right)  \right)  ^{2}e^{-2\lambda_{0}\left(  \mathcal{R}%
\left(  \tau\right)  \right)  }\left(  r_{n}\left(  \tau;\bar{\tau}\right)
\right)  ^{2}\left(  \mathcal{W}\left(  \tau-\bar{\tau}\right)  \right)  ^{2}$
cancel out. Therefore, these terms do not modify the order of magnitude of
$\rho_{n}^{\ast}\left(  \tau;\bar{\tau}\right)  .$ We then obtain that, taking
$T_{0}$ large enough, we would obtain the estimate:%
\begin{equation}
\left\vert \rho_{n}^{\ast}\left(  \tau;\bar{\tau}\right)  \right\vert
\leq\frac{1+\left(  \bar{\tau}-\tau\right)  }{10L} \label{B4E4}%
\end{equation}

Combining (\ref{B4E3}), (\ref{B4E4}) we conclude the proof of the\ Lemma.
\end{proof}

As a next step we obtain estimates for the derivatives of the functions
$r_{n}\left(  \tau;\bar{\tau}\right)  ,\ w_{n}\left(  \tau;\bar{\tau}\right)
.$

\begin{lemma}
\label{EstDer}Given $L>0,$ there exists $T_{0}=T_{0}\left(  L\right)  >0$ such
that, for any $\bar{D}\in\mathcal{Y}_{L,T,a}$ if we define $\left(
r_{n},w_{n},D_{n}\right)  $ as in Subsection \ref{OperatorT} the following
estimate holds $a.e.$ $\left(  \tau,\bar{\tau}\right)  \in\mathcal{U}\left(
T\right)  :$%
\begin{equation}
L\left[  \left\vert \frac{\partial r_{n}\left(  \tau;\bar{\tau}\right)
}{\partial\bar{\tau}}+\mathcal{R}^{\prime}\left(  \tau-\bar{\tau}\right)
\right\vert +\left\vert \frac{\partial w_{n}\left(  \tau;\bar{\tau}\right)
}{\partial\bar{\tau}}+\mathcal{W}^{\prime}\left(  \tau-\bar{\tau}\right)
\right\vert \right]  \leq\frac{1}{5} \label{G1E6}%
\end{equation}

\end{lemma}

\begin{proof}
To this end we need to differentiate (\ref{F5E1})-(\ref{F5E3}) with respect to
$\bar{\tau}$. The formulas for the derivatives are rather long, but we can
estimate the form of the resulting linear equations for large values of
$r_{n}\left(  \tau;\bar{\tau}\right)  .$ This is the only range of values that
we need to estimate in detail, since for bounded values of $\left(  \bar{\tau
}-\tau\right)  $ we can obtain estimates using Gronwall-like arguments. It is
worth to examine the type of terms that we obtain in the equations
(\ref{F9E2})-(\ref{F9E4}).

First, we have terms coming from the derivatives of $\bar{\mu},\ \lambda.$ We
differentiate with respect to $\bar{\tau}.$ Therefore, we do not need to
differentiate with respect to the terms in the previous iteration, since they
depend on $\tau.$ We need to estimate the derivatives of the functions
$\lambda\left(  \tau,r\right)  $,\ $\bar{\mu}\left(  \tau,r\right)  $ which
can be computed by means of:
\[
\lambda\left(  \tau,r\right)  =\frac{1}{2}\log\left(  \frac{r}{r-R_{0}\left(
\tau,r\right)  }\right)  =\frac{1}{2}\log\left(  r\right)  -\frac{1}{2}%
\log\left(  r-R_{0}\left(  \tau,r\right)  \right)
\]%
\begin{align*}
\bar{\mu}\left(  \tau,r\right)   &  =\bar{\mu}\left(  \tau,r_{+}\left(
\tau\right)  \right)  +\int_{r_{+}\left(  \tau\right)  }^{r}\frac{4\pi
^{2}v_{1}^{2}\left(  \tau,\xi\right)  B_{1}\left(  \tau,\xi\right)  }%
{\tilde{E}_{1}\left(  \tau,\xi\right)  }\frac{d\xi}{\xi-R_{0}}\\
&  +\frac{1}{2}\int_{r_{+}\left(  \tau\right)  }^{r}\frac{d\xi}{\left(
\xi-R_{0}\left(  \tau,\xi\right)  \right)  }-\frac{1}{2}\log\left(  \frac
{r}{r_{+}\left(  \tau\right)  }\right)
\end{align*}

It is easily seen that the derivative of both functions with respect to $r$
decreases like $\frac{1}{r^{2}}.$ Therefore, this yields in the equations for
$r_{n}\left(  \tau;\bar{\tau}\right)  $ and $w_{n}\left(  \tau;\bar{\tau
}\right)  $ terms decreasing like $\frac{1}{1+\left(  \bar{\tau}-\tau\right)
^{2}}\frac{\partial r}{\partial\bar{\tau}}.$ On the other hand, the
derivatives of $r_{n}\left(  \tau;\bar{\tau}\right)  $ in both equations
result in terms decreasing rather fast too. The reason is that in the first
term the asymptotics as $r\rightarrow\infty$ is like a constant. The
correction is like $\frac{1}{r},$ and the derivative gives again terms of
order $\frac{1}{r^{2}}.$ In the first equation this results in terms like
$\frac{1}{1+\left(  \bar{\tau}-\tau\right)  ^{2}}$ multiplied by
$\frac{\partial r}{\partial\bar{\tau}}$ and in the equation for $w$ this
results in terms like $\frac{1}{1+\left(  \bar{\tau}-\tau\right)  ^{3}}$
multiplied by $\frac{\partial r}{\partial\bar{\tau}}.$

We now consider the derivatives of $w.$ They give terms that are multiplied at
least by $\frac{1}{1+\left(  \bar{\tau}-\tau\right)  ^{2}}$ in the second
equation, since we have basically the same contributions as in the equation
without derivatives. On the other hand, the effect of the terms containing $w$
in the first equation is more subtle. Notice that for large values of $r$ the
function
\[
\frac{e^{\bar{\mu}\left(  \tau,r_{n}\left(  \tau;\bar{\tau}\right)  \right)
-2\lambda\left(  \tau,r_{n}\left(  \tau;\bar{\tau}\right)  \right)  }%
w_{n}\left(  \tau;\bar{\tau}\right)  r_{n}\left(  \tau;\bar{\tau}\right)
}{\sqrt{1+e^{-2\lambda\left(  \tau,r_{n}\left(  \tau;\bar{\tau}\right)
\right)  }\left(  w_{n}\left(  \tau;\bar{\tau}\right)  \right)  ^{2}\left(
r_{n}\left(  \tau;\bar{\tau}\right)  \right)  ^{2}}}%
\]
is independent of $w.$ This means that this dependence does not appear for
large values. Using Taylor we obtain a dependence with the form $\frac
{G\left(  w\right)  }{\left(  r_{n}\left(  \tau;\bar{\tau}\right)  \right)
^{2}}$ at least. This results in terms with the form $\frac{1}{1+\left(
\bar{\tau}-\tau\right)  ^{2}}\frac{\partial w_{n}\left(  \tau;\bar{\tau
}\right)  }{\partial\bar{\tau}}$ Therefore, the linearized equation has the
following form:%
\begin{align*}
\frac{\partial}{\partial\tau}\left(  \frac{\partial r_{n}\left(  \tau
;\bar{\tau}\right)  }{\partial\bar{\tau}}\right)   &  =O\left(  \frac
{1}{1+\left(  \bar{\tau}-\tau\right)  ^{2}}\right)  \frac{\partial
w_{n}\left(  \tau;\bar{\tau}\right)  }{\partial\bar{\tau}}+O\left(  \frac
{1}{1+\left(  \bar{\tau}-\tau\right)  ^{2}}\right)  \frac{\partial
r_{n}\left(  \tau;\bar{\tau}\right)  }{\partial\bar{\tau}}\\
\frac{\partial}{\partial\tau}\left(  \frac{\partial w_{n}\left(  \tau
;\bar{\tau}\right)  }{\partial\bar{\tau}}\right)   &  =O\left(  \frac
{1}{1+\left(  \bar{\tau}-\tau\right)  ^{3}}\right)  \frac{\partial
r_{n}\left(  \tau;\bar{\tau}\right)  }{\partial\bar{\tau}}+O\left(  \frac
{1}{1+\left(  \bar{\tau}-\tau\right)  ^{2}}\right)  \frac{\partial
w_{n}\left(  \tau;\bar{\tau}\right)  }{\partial\bar{\tau}}%
\end{align*}
where these formulas must be understood in weak form, as in the proof of Lemma
\ref{geomConst}. Therefore the functions $\frac{\partial r_{n}\left(
\tau;\bar{\tau}\right)  }{\partial\bar{\tau}}$,\ $\frac{\partial w_{n}\left(
\tau;\bar{\tau}\right)  }{\partial\bar{\tau}}$ are close to constant. They are
determined then by the boundary values. We use the equations:
\[
r_{n}\left(  \tau;\bar{\tau}\right)  =r_{+}\left(  \tau\right)
\ \ \ ,\ \ \ w_{n}\left(  \tau;\bar{\tau}\right)  =\frac{e^{\lambda\left(
\tau,r_{n}\left(  \tau;\bar{\tau}\right)  \right)  }}{\left(  -\bar{t}\right)
}V_{1}\left(  \frac{r_{+}\left(  \tau\right)  }{\left(  -\bar{t}\right)
}\right)
\]

Then:%
\begin{align*}
\frac{\partial r_{n}\left(  \tau;\bar{\tau}\right)  }{\partial\tau}%
+\frac{\partial r_{n}\left(  \tau;\bar{\tau}\right)  }{\partial\bar{\tau}}  &
=O\left(  e^{-\delta\bar{\tau}}\right) \\
\frac{\partial w_{n}\left(  \tau;\bar{\tau}\right)  }{\partial\tau}%
+\frac{\partial w_{n}\left(  \tau;\bar{\tau}\right)  }{\partial\bar{\tau}}  &
=O\left(  e^{-\delta\bar{\tau}}\right)
\end{align*}

We could have much better estimates using the estimates for the self-similar
solution. On the other hand we can compute $\frac{\partial r_{n}\left(
\tau;\bar{\tau}\right)  }{\partial\tau},\ \frac{\partial w_{n}\left(
\tau;\bar{\tau}\right)  }{\partial\tau}$ using the differential equation
itself. This would give an approximation of order:%
\begin{align*}
\frac{\partial r_{n}\left(  \tau;\bar{\tau}\right)  }{\partial\tau}  &
=\mathcal{R}^{\prime}\left(  \bar{\tau}-\bar{\tau}\right)  +O\left(
e^{-\delta\bar{\tau}}\right) \\
\frac{\partial w_{n}\left(  \tau;\bar{\tau}\right)  }{\partial\tau}  &
=\mathcal{W}^{\prime}\left(  \bar{\tau}-\bar{\tau}\right)  +O\left(
e^{-\delta\bar{\tau}}\right)
\end{align*}

We then obtain, using standard continuous dependence results, that
$\frac{\partial r_{n}\left(  \tau;\bar{\tau}\right)  }{\partial\bar{\tau}}$
and $\frac{\partial w_{n}\left(  \tau;\bar{\tau}\right)  }{\partial\bar{\tau}%
}$ can be approximated for $\left(  \bar{\tau}-\tau\right)  $ of order one, by
means of $-\mathcal{R}^{\prime}\left(  \tau-\bar{\tau}\right)  $ and
$-\mathcal{W}^{\prime}\left(  \tau-\bar{\tau}\right)  $ respectively. For
large values of $\left(  \bar{\tau}-\tau\right)  $ the integrable decay of the
terms in the linearized equations imply a small change of the values of
$\frac{\partial r_{n}\left(  \tau;\bar{\tau}\right)  }{\partial\bar{\tau}}$
and $\frac{\partial w_{n}\left(  \tau;\bar{\tau}\right)  }{\partial\bar{\tau}%
}.$ We then have that:%
\[
\left\vert \frac{\partial r_{n}\left(  \tau;\bar{\tau}\right)  }{\partial
\bar{\tau}}+\mathcal{R}^{\prime}\left(  \tau-\bar{\tau}\right)  \right\vert
\ \ ,\ \ \ \left\vert \frac{\partial w_{n}\left(  \tau;\bar{\tau}\right)
}{\partial\bar{\tau}}+\mathcal{W}^{\prime}\left(  \tau-\bar{\tau}\right)
\right\vert
\]
are uniformly small if $T$ is large enough. Therefore, the corresponding terms
in the norm can be estimated by a quantity arbitrarily small if $T$ is large.

It is important to take into account that we can differentiate the equations
satisfied by $r_{n}\left(  \tau;\bar{\tau}\right)  ,\ w_{n}\left(  \tau
;\bar{\tau}\right)  $ only $a.e.$ The argument must be done in the integrated
version of the differential equations satisfied by $r_{n}\left(  \tau
;\bar{\tau}\right)  ,\ w_{n}\left(  \tau;\bar{\tau}\right)  .$ We take
derivatives with respect to $\bar{\tau}$ which can be introduced inside the
integrals. The estimates are done then using a Gronwall argument and this
gives estimates $a.e.$
\end{proof}

\begin{lemma}
\label{rLower}Given $L>0,$ there exists $T_{0}=T_{0}\left(  L\right)  >0$ such
that, for any $\bar{D}\in\mathcal{Y}_{L,T,a}$ if we define $\left(
r_{n},w_{n},D_{n}\right)  $ as in Subsection \ref{OperatorT} we have:%
\[
\frac{\partial r_{n}}{\partial\tau}\left(  \tau;\bar{\tau}\right)  \leq
-\frac{1}{L}\ \ \ ,\ \ \ a.e.\text{ }\left(  \tau,\bar{\tau}\right)
\in\mathcal{U}\left(  T\right)
\]

\end{lemma}

\begin{proof}
It is just a consequence of (\ref{F9E2}) and Lemma \ref{linDiff}.
\end{proof}

Finally we estimate the function $D_{n}$ using (\ref{F7E4}) as well as the
estimates obtained for the functions $r_{n},\ w_{n}.$

\begin{lemma}
\label{EstD}There exists $L_{0}>0,$ such that, for any $L>L_{0},$ there exists
$T_{0}=T_{0}\left(  L\right)  $ such that, for $T>T_{0},$ given $\bar{D}%
\in\mathcal{Y}_{L,T,a}$ if we define $\left(  r_{n},w_{n},D_{n}\right)  $ as
in (\ref{F8E5}) (cf. also (\ref{F9E2})-(\ref{F9E4}), (\ref{F7E4})) we have:
\[
D_{n}\left(  \tau;\bar{\tau}\right)  \exp\left(  2\bar{\tau}\right)
\leq1\ \ ,\ a.e.\text{ }\left(  \tau,\bar{\tau}\right)  \in\mathcal{U}\left(
T\right)
\]

\end{lemma}

\begin{proof}
We need to obtain estimates for $D_{n}.$ To this end we use (\ref{F7E4}) which
is satisfied $a.e.$ in $\mathcal{U}\left(  T\right)  .$ It is relevant to
examine the dominant terms in this equation as $r\rightarrow\infty.$ The term
$\partial_{r}\bar{w}_{n}\left(  r\left(  \tau;\bar{\tau}\right)  ,\tau\right)
=\left(  \frac{\partial w_{n}}{\partial\bar{\tau}}\left(  \tau,\bar{\tau
}\right)  \right)  /\left(  \frac{\partial r_{n}}{\partial\bar{\tau}}\left(
\tau,\bar{\tau}\right)  \right)  $ is bounded due to (\ref{G1E6}). On the
other hand, using the field equations (\ref{A2E1}), (\ref{A2E2}) we can
estimate the terms $\bar{\mu}_{r}\left(  r_{n}\left(  \tau;\bar{\tau}\right)
,\tau\right)  ,\ \lambda_{r}\left(  r_{n}\left(  \tau;\bar{\tau}\right)
,\tau\right)  $ in (\ref{F7E4}). The contribution due to terms in
(\ref{A2E1}), (\ref{A2E2}) containing $\bar{D}$ can be estimated, using
(\ref{A3E1}), as $Ce^{-2\tau}e^{-ar\left(  \tau;\bar{\tau}\right)  }.$ On the
other hand, the equations (\ref{A2E1}), (\ref{A2E2}) also yield terms
$\frac{\left(  1-e^{2\lambda}\right)  }{r}$ which can be estimated as
$\frac{C}{r^{2}}$ using (\ref{F7E9a}). Combining all these estimates for
$\partial_{r}\bar{w}_{n}\left(  r\left(  \tau;\bar{\tau}\right)  ,\tau\right)
,\ \bar{\mu}_{r}\left(  r_{n}\left(  \tau;\bar{\tau}\right)  ,\tau\right)
,\ \lambda_{r}\left(  r_{n}\left(  \tau;\bar{\tau}\right)  ,\tau\right)  $ and
using also Proposition \ref{appFunctions} and Lemma \ref{linDiff} we can then
rewrite (\ref{F7E4}) as:%
\begin{equation}
\frac{\partial D_{n}\left(  \tau;\bar{\tau}\right)  }{\partial\tau}=K\left(
\tau;\bar{\tau}\right)  D_{n}\left(  \tau;\bar{\tau}\right)
\ \ ,\ \ \ \left\vert K\left(  \tau;\bar{\tau}\right)  \right\vert \leq
\frac{C}{1+\left(  \bar{\tau}-\tau\right)  ^{2}}\ \label{B4E5}%
\end{equation}
for some constant $C$ independent of $L,\ T.$ Using now the boundary condition
$D_{n}\left(  \bar{\tau};\bar{\tau}\right)  =\left(  -\bar{t}\right)
b_{1}\left(  \frac{r_{+}\left(  \tau\right)  }{\left(  -\bar{t}\right)
}\right)  \exp\left(  \lambda\left(  r_{n}\left(  \tau;\bar{\tau}\right)
,\tau\right)  \right)  $ (cf. (\ref{F7E4})) and using (\ref{F2E5}) as well as
the fact that $\gamma\left(  y_{0}\right)  >2$ if $y_{0}$ is sufficiently
small (cf. Remark \ref{gam2}), we obtain, integrating (\ref{B4E5}):%
\begin{equation}
0\leq D_{n}\left(  \tau;\bar{\tau}\right)  \leq e^{-2\bar{\tau}}%
\ \ ,\ \ a.e.\text{ }\left(  \tau,\bar{\tau}\right)  \in\mathcal{U}\left(
T\right)  \ \label{B4E5a}%
\end{equation}
if $T$ is sufficiently large.
\end{proof}

\begin{proof}
[Proof of Proposition \ref{mTit}]Due to Lemma \ref{auxFunctions} and
(\ref{A1E1a}) if we prove that the function $\left(  r_{n},w_{n},D_{n}\right)
$ is in the space $\mathcal{X}_{L,T}$ if $L$ and $T_{0}\left(  L\right)  $ are
sufficiently large, we would obtain that the corresponding function $\bar
{D}_{n}\in\mathcal{Y}_{L,T,a}$ if we assume that $L$ is large enough. We will
then check that $\left(  r_{n},w_{n},D_{n}\right)  \in\mathcal{X}_{L,T}.$ The
fact that $D_{n}$ satisfies (\ref{A1E1}) follows from Lemma \ref{EstD}. The
inequalities in (\ref{A1E2}) are satisfied by the corresponding functions
$\rho_{n}^{\ast},\ z_{n}$ associated to $r_{n},w_{n}$ due to Lemma
\ref{linDiff}. The inequalities (\ref{A1E3}) are a consequence of Lemma
\ref{EstDer}. The estimate (\ref{A1E3a}) follows from Lemma \ref{rLower}. The
inequalities (\ref{A1E5}) follow from (\ref{F9E2}), (\ref{F9E3}),
(\ref{F9E4}), (\ref{F7E4}) if $L$ is chosen sufficiently large (independently
on $T$), since the right hand side of (\ref{F9E2}), (\ref{F9E3}),
(\ref{F9E4}), (\ref{F7E4}) is uniformly bounded by a constant independent of
$T$ if $T$ is large.

The fact that $r_{n}$ satisfies (\ref{Y3E6}) follows by construction (cf.
(\ref{F9E2})) as well as the uniform derivative estimates for $r_{n}.$
\end{proof}

\subsection{Weak continuity and compactness of the operator $\mathcal{T}$.}

We need to prove that the operator $\mathcal{T}$ is continuous and compact in
the topology of the space $\mathcal{Y}_{L,T,a}.$

\begin{lemma}
\label{compactness}The operator $\mathcal{T}$ defined in (\ref{F8E5}) is
continuous and compact in the topology of the space $\mathcal{Y}_{L,T,a}$
defined in Subsection \ref{Topo}.
\end{lemma}

\begin{proof}
This result is a consequence of the identity (\ref{Q1E5}). The functions
$r_{n},\ w_{n}$ depend continuously on $\bar{D}$ in the topology uniform in
$t$ and weak in $r$ defined in Section \ref{Topo}. Suppose that $\mathcal{T}%
:\bar{D}\rightarrow\bar{D}_{n}$ Then:
\begin{equation}
\bar{D}_{n}\left(  t,r_{n}\left(  t,\bar{t}\right)  \right)  \frac{\partial
r_{n}\left(  t,\bar{t}\right)  }{\partial\bar{t}}=D_{n}\left(  t,\bar
{t}\right)  \frac{\partial r_{n}\left(  t,\bar{t}\right)  }{\partial\bar{t}%
}=D_{n}\left(  \bar{t},\bar{t}\right)  \frac{\partial r_{n}\left(  \bar
{t},\bar{t}\right)  }{\partial\bar{t}}=\omega\left(  \bar{t}\right)
\label{L1E2}%
\end{equation}

In order to study the continuity of the operator $\mathcal{T}$ we need to
study its action with respect to test functions. Notice that:%
\begin{align}
\int_{r_{+}\left(  t\right)  }^{\infty}\bar{D}_{n}\left(  t,r\right)
\varphi\left(  t,r\right)  dr  &  =\int_{t}^{\infty}\bar{D}_{n}\left(
t,r_{n}\left(  t,\bar{t}\right)  \right)  \frac{\partial r_{n}\left(
t,\bar{t}\right)  }{\partial\bar{t}}\varphi\left(  t,r_{n}\left(  t,\bar
{t}\right)  \right)  d\bar{t}\label{L1E1}\\
&  =\int_{t}^{\infty}\omega\left(  \bar{t}\right)  \varphi\left(
t,r_{n}\left(  t,\bar{t}\right)  \right)  d\bar{t}\nonumber
\end{align}

The continuity of the functions $r_{n}$ in the continuous-weak topology in
$\bar{D}$ implies the continuity of $\mathcal{T}$.

On the other hand, in order to prove compactness of the operator $\mathcal{T}$
we need to prove equicontinuity of $\int_{r_{+}\left(  t\right)  }^{\infty
}\bar{D}_{n}\left(  t,r\right)  \varphi\left(  t,r\right)  dr.$ To this end we
use the identity (\ref{L1E1}). The function $\gamma\left(  \bar{t}\right)  $
is smooth. The equicontinuity of the functional then follows from the
equicontinuity of $r_{n}\left(  t,\bar{t}\right)  $ with respect to $t,$ which
is a consequence of the differential equation satisfied by this function. The
right-hand side is uniformly bounded and therefore we have equicontinuity.
This gives the desired compactness.
\end{proof}

\subsection{Fixed point argument. End of the proof of Theorem \ref{main}.}

\begin{proof}
[End of the proof of Proposition \ref{ext}]This result follows from Schauder's
fixed point Theorem (cf. \cite{Evans}) combined with Lemma \ref{compactness}
as well as the fact that every fixed point of the operator $\mathcal{T}$
allows us to obtain a solution of (\ref{Y2E4})-(\ref{Y2E7}) in the sense of
characteristics by means of the solution of the characteristic equations
(\ref{F9E2})-(\ref{F7E4}).
\end{proof}

\begin{proof}
[Proof of Theorems \ref{main} and \ref{mainF}]Theorem \ref{main} follows from
\ Proposition \ref{ws} and Proposition \ref{ext}. Theorem \ref{mainF} then
follows from Proposition \ref{fExt}.
\end{proof}

\section{Geometrical properties of the solution.}

In this section we summarize several geometrical properties of the spacetime
constructed in the previous sections. In particular we will prove that the
corresponding metric is geodesically incomplete. Moreover, we will rewrite the
metric in double-null coordinates. This will allow us to clarify the causal
relations between the different regions of the spacetime. A consequence of
this will be a proof of the fact that the spacetime obtained does not contain
any horizon separating the regions where the curvature of the spacetime is
unbounded and the regions at infinity where $r=\infty$. We summarize the
results in the following theorem. In this section we follow the common use of
the letters $\left(  u,v\right)  $ to denote the double null coordinates.
Therefore, $v$ is not the coordinate introduced in (\ref{S3E1}).

\begin{theorem}
\label{doubnull}There exists a diffeomorphism $\left(  t,r\right)
\rightarrow\left(  u,v\right)  $ which transforms the portion of spacetime
$\left\{  \left(  t,r,\theta,\varphi\right)  :0\leq r<\infty,\ t_{0}%
<t<0,\ \theta\in\left[  0,\pi\right]  ,\varphi\in\left[  0,2\pi\right]
\right\}  $ into the region$\left\{  \left(  u,v,\theta,\varphi\right)
:U_{-}\left(  v\right)  \leq u<1\text{\ if\ }v\in\left(  -1,0\right]
,U_{+}\left(  v\right)  \leq u<1\ \text{if\ }v\in\left[  0,1\right)  \right.
$\linebreak$\ \left.  \theta\in\left[  0,\pi\right]  ,\ \varphi\in\left[
0,2\pi\right]  \right\}  $ for suitable functions $U_{-},\ U_{+}\in
C^{1}\left(  \left[  0,1\right]  \right)  $ satisfying
\begin{align}
U_{-}\left(  0\right)   &  =U_{+}\left(  0\right)  <1\ \ ,\ \ U_{-}\left(
-1\right)  =1\ \ \ U_{+}\left(  1\right)  =1\ \label{M2E2}\\
&  U_{-}\text{ is a decreasing function},\ U_{+}\text{ is an increasing
function.} \label{M2E3}%
\end{align}

The metric $ds^{2}$ of the spacetime in the coordinates $\left(
u,v,\theta,\varphi\right)  $ has the form:%
\begin{equation}
ds^{2}=-\left(  \Omega\left(  u,v\right)  \right)  ^{2}dudv+\left[  r\left(
u,v\right)  \right]  ^{2}\left(  d\theta^{2}+\sin^{2}\theta d\varphi
^{2}\right)  . \label{M2E1}%
\end{equation}
for a suitable function $\Omega\left(  u,v\right)  .$ The system of
coordinates $\left(  u,v\right)  $ will be denoted as double-null coordinates.

The curves $\left\{  u=const\right\}  $ and $\left\{  v=const\right\}  $ are
radial light rays. The center $r=0$ is, in the coordinates $\left(
u,v\right)  $ the line $u=U_{+}\left(  v\right)  ,\ v\in\left(  0,1\right)  .$
On the other hand, the limit $r\rightarrow\infty,$ where the spacetime
obtained is asymptotically flat, is represented by the line $u=1,\ -1=U_{-}%
\left(  1\right)  \leq v<U_{+}\left(  1\right)  .$ The point $\left(
u,v\right)  =\left(  U_{+}\left(  1\right)  ,1\right)  $ is a singularity for
the spacetime obtained, since the curvature of the metric becomes unbounded.

The spacetime obtained is geodesically incomplete. More specifically, the
curve $\left\{  r=0,\ t_{0}\leq t<0\right\}  $ is a geodesic along which the
proper time is the coordinate $t.$ Therefore the singular point is reached in
a finite proper time along this line.

There exists a light ray connecting any point in the spacetime region
\[
\left\{  \left(  u,v\right)  :u\in\left(  0,1\right)  ,\ U_{-}\left(
u\right)  \leq v\leq U_{+}\left(  u\right)  \right\}
\]
with the set $\left\{  r=\infty\right\}  =\left\{  u=1,\ U_{-}\left(
1\right)  \leq v<U_{+}\left(  1\right)  \right\}  .$ Therefore, no horizon
appears in any part of the spacetime considered.
\end{theorem}

We will use the following asymptotic description of the fields $\lambda$ and
$\mu$.

\begin{lemma}
\label{asF}Suppose that $\zeta$ is a solution (\ref{S1E3})-(\ref{S1E6}),
(\ref{S3E7}), (\ref{S3E10}) in the sense of Definition \ref{zetaWeak} as
obtained in Theorem \ref{main} and let $\lambda,\ \bar{\mu}$ the corresponding
fields. Then, the following asymptotics holds:%
\begin{equation}
\left\vert \lambda\left(  \tau,r\right)  -\frac{1}{2}\log\left(  \frac
{r}{r-R_{0}\left( \tau, \infty\right)  }\right)  \right\vert \leq\bar{C}%
\exp\left(  -ar\right)  \exp\left(  -b\tau\right)  \ \label{lam1}%
\end{equation}%
\begin{equation}
\left\vert \bar{\mu}\left(  \tau,r\right)  -\bar{\mu}\left(  \tau
,\infty\right)  -\frac{1}{2}\log\left(  1-\frac{2R_{0}\left(  \tau
,\infty\right)  }{3r}\right)  \right\vert \leq\bar{C}\exp\left(  -ar\right)
\exp\left(  -b\tau\right)  \ \label{mu1}%
\end{equation}
uniformly in $r\geq R_{\max},\ \tau\geq\tau_{0}$ where:%
\begin{align}
\left\vert R_{0}\left(  \tau,\infty\right)  -\frac{2R_{\max}}{3}\right\vert
&  \leq\bar{C}\exp\left(  -b\tau\right) \label{a1}\\
\left\vert \bar{\mu}\left(  \tau,\infty\right)  -\bar{\mu}_{0}\left(
\infty\right)  \right\vert  &  \leq\bar{C}\exp\left(  -b\tau\right)
\ \ \ ,\ \ \bar{\mu}_{0}\left(  \infty\right)  =\log\left(  \frac{R_{\max
}\sqrt{3\left(  1-y_{0}^{2}\right)  }}{y_{0}}\right)  \label{a2}%
\end{align}

\end{lemma}

\begin{proof}
These formulas are a consequence of the fact that $\bar{D}$ satisfies
(\ref{A1E1a}). We have similar estimates for $\rho,\ p$ due to (\ref{T1E3}),
(\ref{Y2E3}). The estimates in the Lemma then follow from Lemma
\ref{defFields}.
\end{proof}

\begin{proof}
[Proof of Theorem \ref{doubnull}]In order to construct the double-null
coordinates we need to solve the differential equations:%
\[
0=-e^{2\mu\left(  t,r\right)  }dt^{2}+e^{2\lambda\left(  t,r\right)  }dr^{2}%
\]
or equivalently:%
\begin{align}
\frac{dr}{dt}  &  =e^{\mu-\lambda}\ \label{R1E1}\\
\frac{dr}{dt}  &  =-e^{\mu-\lambda} \label{R1E2}%
\end{align}

We now use the fact that the functions $\lambda,\ \mu$ have the following
self-similar form for $r\leq R_{+}\left(  t\right)  :$%
\[
\mu\left(  t,r\right)  =\lambda\left(  t,r\right)  =0\ \ \text{if\ \ }r\leq
y_{0}\left(  -t\right)  \ \ ,\ \ t_{0}\leq t<0
\]%
\begin{equation}
\mu\left(  t,r\right)  =U\left(  y\right)  \;\;,\;\;\lambda\left(  t,r\right)
=\Lambda\left(  y\right)  \ \ ,\ \ y=\frac{r}{\left(  -t\right)
}\ \ ,\ \ \ r\leq R_{+}\left(  t\right)  \label{R1E2a}%
\end{equation}

On the other hand, $\lambda,\ \mu$ are given for $r>R_{+}\left(  t\right)  $
by means of Lemma \ref{CWFields}, with $\left(  r,w,D\right)  $ given by the
fixed point obtained in Proposition \ref{ext}. Using (\ref{F3E3}) we can
reformulate (\ref{R1E1}) as:%
\begin{equation}
\frac{dr}{d\tau}=e^{\bar{\mu}\left(  \tau,r\right)  -\lambda\left(
\tau,r\right)  } \label{R1E1a}%
\end{equation}

Using then the asymptotics (\ref{lam1})-(\ref{a2}), as well as standard ODE
arguments, we obtain that for any $v\in\left[  0,1\right)  $ there exists a
unique solution of (\ref{R1E1a}) denoted as $r_{+}\left(  \tau;v\right)  $ and
satisfying:%
\begin{equation}
r_{+}\left(  \tau_{0}+\tanh^{-1}\left(  v\right)  ;v\right)  =0\ \ ,\ v\in
\left[  0,1\right)  \label{R1E3a}%
\end{equation}
and for any $v\in\left(  -1,0\right]  $ there exists a unique solution of
(\ref{R1E1a}) such that:%
\begin{equation}
r_{+}\left(  \tau_{0};v\right)  =-\tanh^{-1}\left(  v\right)  \ \ ,\ v\in
\left(  -1,0\right]  \label{R1E3b}%
\end{equation}

The function $\tau\rightarrow r_{+}\left(  \tau;v\right)  $ is increasing. If
$r_{+}\left(  \tau_{0};v\right)  >r_{+}\left(  \tau_{0}\right)  $ (cf.
(\ref{F2E7}), (\ref{F2E7d})) we have that the intersection between the curves
$\left\{  r=r_{+}\left(  \tau;v\right)  ,\ \tau\right\}  $, $\left\{
r=r_{+}\left(  \tau_{0}\right)  \right\}  $ is empty. On the other hand, if
$r_{+}\left(  \tau_{0};v\right)  \leq r_{+}\left(  \tau_{0}\right)  $ it
follows from Lemma \ref{asF}, (\ref{F6E4a}) and (\ref{R1E1a}) that for
$\left\vert t_{0}\right\vert $ sufficiently small there is a unique
intersection between the curves $\left\{  r=r_{+}\left(  \tau;v\right)
\right\}  $, $\left\{  r=r_{+}\left(  \tau_{0}\right)  \right\}  .$ In both
cases the function $r_{+}\left(  \tau;v\right)  $ is globally defined for
$\tau\geq\min\left\{  \tau_{0}+\tanh^{-1}\left(  v\right)  ,\tau_{0}\right\}
$ and $\lim_{\tau\rightarrow\infty}r_{+}\left(  \tau;v\right)  =\infty.$ The
family of disjoint curves
\[
\left\{  r=r_{+}\left(  \tau;v\right)  :\tau\geq\tau_{0}+\min\left\{  \tau
_{0}+\tanh^{-1}\left(  v\right)  ,\tau_{0}\right\}  \right\}  \ \text{with
}v\in\left(  -1,-1\right)
\]
covers the whole set $\left\{  \tau\geq\tau_{0},\ r\geq0\right\}  .$ We can
then use these curves to define a function $v:\left\{  \tau\geq\tau
_{0},\ r\geq0\right\}  \rightarrow\left(  -1,-1\right)  $ which assigns to
each pair $\left(  r,\tau\right)  $ a value $v\left(  r,\tau\right)  .$

On the other hand we define functions $r_{-}\left(  \tau;u\right)  $ by means
of the problem:%
\begin{equation}
\frac{dr}{d\tau}=-e^{\bar{\mu}\left(  \tau,r\right)  -\lambda\left(
\tau,r\right)  }\ \ ,\ \ \tau\geq\tau_{0}\ \ ,\ \ r_{-}\left(  \tau
_{0};u\right)  =\tanh^{-1}\left(  u\right)  \ \ ,\ \ u\in\left[  0,1\right)
\ \label{R1E4a}%
\end{equation}

The function $\tau\rightarrow r_{-}\left(  \tau;u\right)  $ is decreasing for
each $u\in\left[  0,1\right)  .$ Using Lemma \ref{asF}, (\ref{F6E4a}) and
(\ref{R1E4a}) it follows that there is at most one intersection between the
curves $\left\{  r=r_{-}\left(  \tau;u\right)  \right\}  $, $\left\{
r=r_{+}\left(  \tau\right)  \right\}  .$ Moreover, there exists exactly one
intersection for each $u\in\left[  0,1\right)  $ such that $\tanh^{-1}\left(
u\right)  \geq r_{+}\left(  \tau_{0}\right)  .$ We then define $\tau_{-}%
^{\ast}\left(  u\right)  \geq\tau_{0}$ for $u\geq\tanh\left(  r_{+}\left(
\tau_{0}\right)  \right)  $ by means of:%
\[
r_{-}\left(  \tau_{-}^{\ast}\left(  u\right)  ;u\right)  =r_{+}\left(
\tau_{-}^{\ast}\left(  u\right)  \right)
\]

In order to obtain $r_{-}\left(  \tau;u\right)  $ for $\tau>\tau_{-}^{\ast
}\left(  u\right)  $ we use the self-similar form $r_{-}\left(  \tau;u\right)
=e^{-\tau}y_{-}\left(  \tau;u\right)  $ which yields the following equation
for $y_{-}\left(  \tau;u\right)  :$%
\[
-y_{-}+\frac{dy_{-}}{d\tau}=-e^{U\left(  y_{-}\right)  -\Lambda\left(
y_{-}\right)  }\ \ ,\ \ \tau\geq\tau_{-}^{\ast}\left(  u\right)
\ \ ,\ \ y_{-}\left(  \tau_{-}^{\ast}\left(  u\right)  ;u\right)  =e^{\tau
_{-}^{\ast}\left(  u\right)  }r_{+}\left(  \tau_{-}^{\ast}\left(  u\right)
\right)
\]

The solution of this equation is given by:%
\begin{equation}
F_{-}\left(  e^{\tau_{-}^{\ast}\left(  u\right)  }r_{+}\left(  \tau_{-}^{\ast
}\left(  u\right)  \right)  \right)  -F_{-}\left(  y_{-}\left(  \tau;u\right)
\right)  =\tau-\tau_{-}^{\ast}\left(  u\right)  \ \label{R1E8b}%
\end{equation}
where:%
\begin{equation}
F_{-}\left(  y\right)  =\int_{0}^{y}\frac{d\xi}{e^{U\left(  \xi\right)
-\Lambda\left(  \xi\right)  }-\xi} \label{R1E8a}%
\end{equation}

We now remark that the analysis of the self-similar solutions in \cite{RV}
imply that $F_{-}\left(  y\right)  $ is well defined and it is increasing if
$y>0.$ In particular, it then follows that (\ref{R1E8b}) that for any
$u\geq\tanh\left(  r_{+}\left(  \tau_{0}\right)  \right)  $ there exists
$\tau_{-}^{\ast\ast}\left(  u\right)  >\tau_{-}^{\ast}\left(  u\right)  ,$
$\tau_{-}^{\ast\ast}\left(  u\right)  <\infty$ such that $y_{-}\left(
\tau_{-}^{\ast}\left(  u\right)  ;u\right)  =0.$ It then follows that the
light rays $\left\{  r=r_{-}\left(  \tau;u\right)  \right\}  $ reach the
center $r=0$ for a finite value of $\tau=\tau_{-}^{\ast\ast}\left(  u\right)
.$ In particular the disjoint curves $\left\{  r=r_{-}\left(  \tau;u\right)
\right\}  ,\ u\in\left[  0,1\right)  $ cover the whole domain $\left\{
\tau\geq\tau_{0},\ r\geq0\right\}  $ and they can be used to define a function
$\left(  r,\tau\right)  \rightarrow v\left(  r,\tau\right)  .$We define the
functions $U_{-}\left(  v\right)  ,\ U_{+}\left(  v\right)  $ by means of the
identities:%
\begin{equation}
r_{+}\left(  \tau;v\right)  =r_{-}\left(  \tau;U_{+}\left(  v\right)  \right)
=0\ \ ,\ \ v\in\left[  0,1\right)  \ \label{R1E9a}%
\end{equation}%
\begin{equation}
r_{+}\left(  \tau_{0};v\right)  =r_{-}\left(  \tau_{0};U_{-}\left(  v\right)
\right)  =-\tanh^{-1}\left(  v\right)  \ \ ,\ \ v\in\left(  -1,0\right]
\label{R1E9b}%
\end{equation}

The properties (\ref{M2E2}), (\ref{M2E3}) then follow from these formulas.
Indeed, the key property which need to be checked is:%
\begin{equation}
\inf\left\{  u\left(  \tau,r\right)  :v\left(  \tau,r\right)  \geq
1-\varepsilon\right\}  \rightarrow1\ \label{R1E9d}%
\end{equation}
as $\varepsilon\rightarrow0^{+}.$ This can be seen as follows. We first claim
that $\tau_{-}^{\ast\ast}\left(  u\right)  \rightarrow\infty$ as
$u\rightarrow1^{-}.$ To see this we just solve the ODE (\ref{R1E2}) for
$\tau_{0}\leq\tau\leq T$ with the initial condition $r\left(  T\right)  =0$
and denote the corresponding solution as $\bar{r}\left(  \cdot;T\right)  .$
Due to Lemma \ref{asF} this solution is defined in the whole interval
$\tau\in\left[  \tau_{0},T\right]  .$ Moreover, we have
\begin{equation}
\bar{r}\left(  \tau_{0};T\right)  \rightarrow\infty\label{R1E9c}%
\end{equation}

Indeed, notice that the uniqueness theorem for ODEs imply that the function
$T\rightarrow\bar{r}\left(  \tau_{0};T\right)  $ is increasing. Suppose that
$\lim_{T\rightarrow\infty}r\left(  \tau_{0};T\right)  =R_{\infty}<\infty.$
Then, due to the uniqueness theorem for ODEs we would obtain that $\tau
_{-}^{\ast\ast}\left(  u\right)  =\infty$ if $u>\tanh\left(  R_{\infty
}\right)  ,$ but this contradicts the fact that $r_{-}\left(  \tau_{-}%
^{\ast\ast}\left(  u\right)  ;u\right)  <\infty$ for any $u<1$ as indicated
above. Therefore $R_{\infty}=\infty$ and this implies (\ref{R1E9c}). We can
now prove (\ref{R1E9d}). Notice that the function $f\left(  \varepsilon
\right)  =\inf\left\{  u\left(  \tau,r\right)  :v\left(  \tau,r\right)
\geq1-\varepsilon\right\}  $ is decreasing in $\varepsilon$ and by
construction $f\left(  \varepsilon\right)  \leq1.$ Suppose that $\lim
_{\varepsilon\rightarrow0^{+}}f\left(  \varepsilon\right)  <1.$ Then there
exists $\delta>0$ such that for any $\varepsilon>0$ there exists at least one
point $\left(  \tau_{\varepsilon},r_{\varepsilon}\right)  $ such that
$v\left(  \tau_{\varepsilon},r_{\varepsilon}\right)  \geq1-\varepsilon$ and
$u\left(  \tau_{\varepsilon},r_{\varepsilon}\right)  \leq1-\delta.$ Moreover,
the definition of $v$ implies that $\tau_{\varepsilon}\geq\tanh^{-1}\left(
v\left(  \tau_{\varepsilon},r_{\varepsilon}\right)  \right)  \geq\tanh
^{-1}\left(  1-\varepsilon\right)  .$ However, this is not possible, because
due to the monotonicity of the function $u\rightarrow r_{-}\left(
\tau;u\right)  $ we would have that $\tau_{\varepsilon}\leq\tau_{-}^{\ast\ast
}\left(  u\left(  \tau_{\varepsilon},r_{\varepsilon}\right)  \right)  \leq
\tau_{-}^{\ast\ast}\left(  1-\delta\right)  <\infty$ and this gives a
contradiction if $\varepsilon$ is sufficiently small.

Using (\ref{R1E9d}) we obtain the identity $U_{+}\left(  1\right)  =1$ in
(\ref{M2E2}). The rest of the identities in (\ref{M2E2}), (\ref{M2E3}) follow
easily from the definitions of $U_{+},U_{-}$ in (\ref{R1E9a}), (\ref{R1E9b}).
\end{proof}

\bigskip

Notice that the causal structure of the spacetime obtained shows that no
horizon is formed, since any point in the spacetime can send a light ray
reaching infinity. However, the singularity obtained is not the type of
singularity usually termed a naked singularity. Indeed, we remark that there
is not any light ray starting at the singularity and reaching infinity. The
spacetime obtained is geodesically incomplete, since the particles at $r=0,$
reach the singularity in a finite proper time.\ On the other hand, light rays
emitted at the center $r=0,$ at times in which the curvatures are arbitrarily
large reach infinity, but no light ray emitted from the singular point reaches
infinity. The solutions obtained have also a causal structure different from
the one associated to the so-called "collapsed cones" which were obtained in
\cite{Chr2} for the Einstein equations coupled to a scalar field. The most
distinctive feature of the solutions constructed in this paper is the fact
that the singularity cannot be reached by any light ray having as starting
point any point of the space-time obtained. It turns out, however, that the
singular point can be reached by some specific time-like trajectories. In the
case of the singularities in \cite{Chr2} it is possible to connect points of
the space-time and points of the singularity by means of light-rays. Figure 1
contains a Penrose diagram of the space-time in order to explain the causal
structure of the metric. Notice that the only singular point in this space
time is the point $P.$ On the other hand, the point $P$ cannot be reached by
any radial light ray, although it can be reached by some time-like curves in
finite proper time, for instance $\left\{  r=0\right\}  .$

\begin{figure}[ptb]
\caption{Penrose diagram of the spacetime obtained.}
\centering
\includegraphics[width=0.7\textwidth]{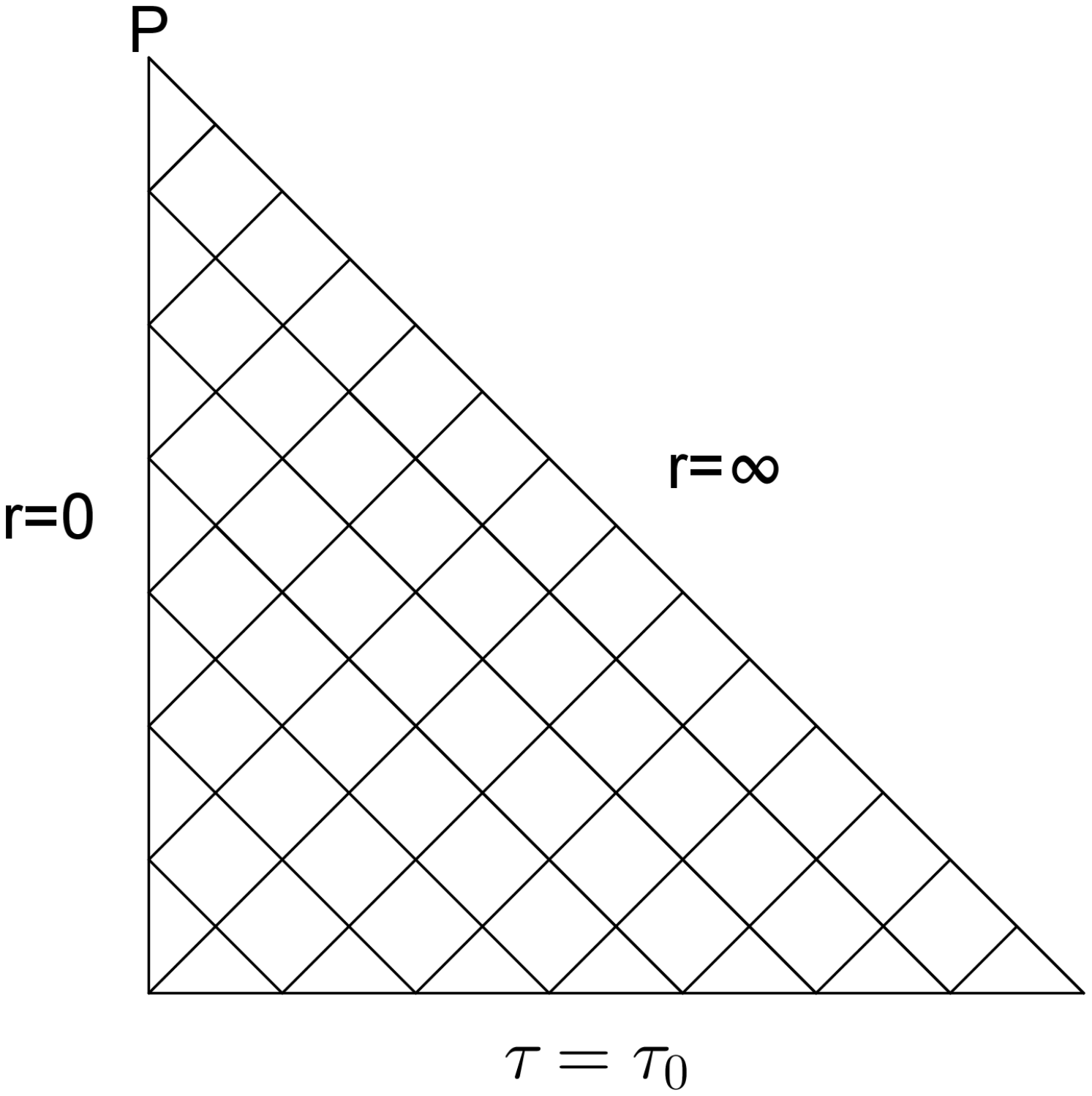}
\end{figure}

\bigskip

Notice that the functions $\rho,\ p$ become unbounded at the surface $\left\{
r=y_{0}\left(  -t\right)  \right\}  ,$ and they are discontinuous at $\left\{
r=R_{+}\left(  t\right)  \right\}  .$ Therefore, the derivatives of the fields
$\lambda,\mu$ are unbounded in the surface $\left\{  r=y_{0}\left(  -t\right)
\right\}  $ and they are discontinuous at $\left\{  r=R_{+}\left(  t\right)
\right\}  $ (cf. (\ref{S1E3}), (\ref{S1E4})). In particular the so-called
Kretschmann scalar (cf. (\cite{Rendall})) becomes unbounded in a neighbourhood
of the set of turning points $\left\{  r=y_{0}\left(  -t\right)  \right\}  .$
Nevertheless, the integrability of the functions $p,\ \rho$ in a neighbourhood
of $\left\{  r=y_{0}\left(  -t\right)  \right\}  $ suggests that the
singularity at he turning points is not a true singularity induced by the
nonlinear character of Einstein equations, but that it is more a fictitious
type of singularity due to the singular type of matter used (dust-like
solutions). A discussion about singularities which are due to the presence of
singular behaviours in the matter model under consideration can be found in
\cite{HE}, Subsection 8.4. It is indicated in \cite{HE} that the minimal
condition that must be requested to a spacetime to consider it singularity
free is timelike and null geodesic completeness. From this point of view, it
might be considered that the spacetime constructed in this paper is
singularity free for $t<0.$ Indeed, a careful analysis of the geodesic
equations $\frac{d^{2}x_{\alpha}}{d\zeta^{2}}+\Gamma_{\alpha\gamma}^{\beta
}\frac{dx^{\alpha}}{d\zeta}\frac{dx^{\gamma}}{d\zeta}=0$ shows that their
evolution is well defined for every time-like and null characteristic in a
neighbourhood of the surface $\left\{  r=y_{0}\left(  -t\right)
,\ t<0\right\}  .$ Away from this surface the geodesic completeness follows
from the smoothness of the metric.

It is natural to ask in which sense the solution described in this paper
represents a singularity of the spacetime which has worse properties than the
spacetime for $t=t_{0}<0,$ which is already singular, due to the divergence of
$\rho$ and $p$ at $r=y_{0}\left(  -t\right)  .$ Seemingly the answer to this
question is that the singular character of the spacetime is only apparent due
to collapse of the whole structure towards $r=0$ as $t\rightarrow0^{-}.$ More
precisely, we can construct several quantities which exhibit this collapsing
behaviour as $t\rightarrow0^{-}.$ One possibility is the following. Suppose
that we denote as $\operatorname*{Kr}$ the Kretschmann scalar given by
$\operatorname*{Kr}=R^{\alpha\beta\gamma\delta}R_{\alpha\beta\gamma\delta}$
(cf. \cite{Rendall}). It has been seen in \cite{RV} that $\operatorname*{Kr}%
\geq0.$ Suppose that we take $A>0$ sufficiently large, but fixed. It then
follows that:%
\begin{equation}
\frac{1}{\left(  -t\right)  }\int_{0}^{A\left(  -t\right)  }\left(
\operatorname*{Kr}\right)  ^{\theta}dr\geq\frac{C_{\theta}}{\left(  -t\right)
^{2\theta}}\ \label{ExpF}%
\end{equation}
for any $\theta\in\left(  0,1\right)  $ with $C_{\theta}>0.$ This inequality
indicates that the curvature near the self-similar region is divergent in some
suitable average sense. Notice that we cannot take $\theta\geq1$ due to the
singular behaviour of $\rho,p$ near $r=y_{0}\left(  -t\right)  ,$ since in
that case the left-hand side of (\ref{ExpF}) would be infinity. Another
quantity which shows the singular character of the spacetime constructed in
this paper can be constructed in terms of the so-called Hawking mass $m_{H}$.
This quantity, and more precisely $\frac{m_{H}}{r}$ measures the degree of
deformation of the spacetime up to some radius $r.$ In the setting of this
paper we have (cf. \cite{Rendall}):%
\[
\frac{2m_{H}}{r}=1-e^{-2\lambda}%
\]
and given that for small $r$ and $t$ we have $\lambda=\Lambda\left(  \frac
{r}{\left(  -t\right)  }\right)  ,$ with $\Lambda\left(  0\right)
=\Lambda\left(  \infty\right)  =0,$ $\Lambda\left(  y\right)  =0$ for $0\leq
y\leq y_{0},$ $\Lambda\left(  y\right)  >0$ if $y_{0}<y<\infty,$ we obtain:%
\[
\lim_{t\rightarrow0}\sup_{0\leq r\leq A\left(  -t\right)  }\frac{m_{H}}{r}\geq
c_{1}>0
\]
where $A>y_{0}$ is a fixed constant. This formula indicates also the presence
of a concentration of mass-energy for small $r$ and $t.$ Note that the Hawking
mass is continuous at the turning point.

\bigskip

\textbf{Acknowledgement. }Most of this work was done when ADR was at the MPI
for Gravitational Physics. JJLV acknowledges support through the CRC 1060 The
mathematics of emergent effects at the University of Bonn, that is funded
through the German Science Foundation (DFG).

\end{document}